\title{Combined-BayesGOF-Jan2018}

\documentclass[12pt]{article}
\usepackage{enumitem}
\usepackage[english]{babel}
\usepackage[utf8]{inputenc}
\usepackage{amssymb}
\usepackage[intlimits]{amsmath}
\usepackage{mathabx}
\usepackage{setspace}
\usepackage{bm}
\usepackage{amsthm,mathrsfs} 
\usepackage{enumitem}
\usepackage{graphicx,graphics,epsfig}
\usepackage{subcaption}
\usepackage[colorinlistoftodos]{todonotes}
\usepackage{caption} 
\usepackage{url}
\usepackage{verbatim}
\usepackage{booktabs}
\usepackage{tabularx}
\usepackage{multirow}
\usepackage{float}
\usepackage[most]{tcolorbox}
\newcolumntype{Y}{>{\centering\arraybackslash}X}

\usepackage[margin=.9in,a4paper]{geometry}
\usepackage[compact]{titlesec}
\usepackage{bibunits}
\usepackage{jabbrv}
\theoremstyle{plain}
\newtheorem{thm}{Theorem}

\newcommand{\bthm}{\begin{thm}}
\newcommand{\ethm}{\end{thm}}

\newcommand{\bpf}{\begin{proof}}
\newcommand{\epf}{\end{proof}}
\usepackage{arydshln}
\theoremstyle{definition}
\newtheorem{defn}{Definition}

\newcommand{\DS}{\mbox{DS}}
\newcommand{\wcte}{\widecheck{\te}}

\numberwithin{equation}{section}
\newcommand{\wtte}{\widetilde \te}

\usepackage{deep-macro1}
\begin{document}
\begin{bibunit}[naturemag-doi]
\begin{center}
{\Large{\bf Bayesian Modeling \textit{via} Goodness-of-fit}} \\[.15in] %
Subhadeep Mukhopadhyay, Douglas Fletcher\\  
Temple University, Department of Statistical Science \\ Philadelphia, Pennsylvania, 19122, U.S.A.
\vspace{-.25em}
\end{center}
\begin{center}
\textit{{\large Dedicated to 80th birthday anniversary of Brad Efron}} \end{center}
\begin{abstract}
\vspace{-.25em}
The two key issues of modern Bayesian statistics are: (i) establishing principled approach for \textit{distilling} statistical prior that is \textit{consistent} with the given data from an initial believable scientific prior; and (ii) development of a \textit{consolidated} Bayes-frequentist data analysis workflow that is more effective than either of the two separately. In this paper, we propose the idea of  ``Bayes \textit{via} goodness-of-fit'' as a framework for exploring these fundamental questions, in a way that is general enough to embrace almost all of the familiar probability models. Several examples, spanning application areas such as clinical trials, metrology, insurance, medicine, and ecology show the unique benefit of this new point of view as a practical data science tool.
\end{abstract}
\renewcommand{\baselinestretch}{1.12}
\setlength{\parskip}{1.4ex}
\section{Introduction}\label{Intro}
Bayesians and frequentists have long been ambivalent toward each other \cite{efron1986isn,sims2010,stigler1982thomas}. The concept of ``prior'' remains the center of this 250 years old tug-of-war: frequentists view prior as a weakness that can hamper scientific objectivity and can corrupt the final statistical inference, whereas Bayesians view it as a \textit{strength} to incorporate relevant domain-knowledge into the data analysis. The question naturally arises: how can we develop a consolidated Bayes-frequentist data analysis workflow \cite{robbins1956, good1992bayes, rubin1984aos, efron2003robbins} 
that enjoys the best of both worlds? The objective of this paper is to develop one such modeling framework.  

We observe samples $y=(y_1,\ldots,y_k)$ from a known probability distribution $f(y|\te)$, where the unobserved parameters $\te=(\te_1,\ldots,\te_k)$ are independent realizations from unknown $\pi(\te)$. Given such a model, Bayesian inference typically aims at answering the following two questions:
\begin{itemize}[itemsep=1pt,topsep=1.24pt]
\item MacroInference: How should we combine $k$ model parameters to come up with an overall, macro-level aggregated statistical behavior of $\te_1,\ldots,\te_k$?
\item MicroInference: Given the observables $y_i$, how should we simultaneously estimate individual micro-level parameters $\te_i$?
\end{itemize}
Thanks to Bayes' rule, answers to these questions are fairly straightforward and automatic once we have the observed data $\{y_i\}_{i=1}^k$ and a specific choice for $\pi(\theta)$. A common practice is to choose $\pi$ as the parametric conjugate prior $g(\te;\al,\be)$, where the hyper-parameters are either selected based on an investigator's expert input or estimated from the data (current/historical) when little prior information is available .
\vskip.65em

{\bf Motivating Questions}. However, an applied Bayesian statistician may find it unsatisfactory to work with an initial believable prior $g(\te)$ at its face value, without being able to interrogate its credibility in the light of the observed data \cite{dempster1975sub,Berger1994robust} as this choice unavoidably shapes his or her final inferences and decisions. A good statistical practice thus demands greater transparency to address this trust-deficit. What is needed is a justifiable class of prior distributions to answer the following \textit{pre}-inferential modeling questions: Why should I believe your prior? How to check its appropriateness (self-diagnosis)? How to quantify and characterize the uncertainty of the a priori selected $g$? Can we use that information to ``refine'' the starting prior (\textit{auto}-correction), which is to be used for subsequent inference? In the end, the question remains: how can we develop a systematic and principled approach to go from a \textit{scientific} prior to a \textit{statistical} prior that is consistent with the current data? A resolution of these questions is necessary to develop a ``dependable and defensible'' Bayesian data analysis workflow, which is the goal of the ``Bayes \textit{via} goodness-of-fit'' technology.
\vskip.65em

{\bf Summary of Contributions}. This paper provides some practical strategies for addressing these questions by introducing a general modeling framework, along with concrete guidelines for applied users. The major practical advantages of our proposal are: (i) computational ease (it does not require Markov chain Monte Carlo (MCMC), variational methods, or any other sophisticated computational techniques); (ii) simplicity and interpretability of the underlying theoretical framework which is general enough to include \textit{almost all} commonly encountered models; and (iii) easy integration with mainframe Bayesian analysis that makes it readily applicable to a wide range of problems. The next section introduces a new class of nonparametric priors $\DS(G,m)$ along with its role in exploratory graphical diagnostic and uncertainty quantification. The estimation theory, algorithm, and real data examples are discussed in Section \ref{sec:EstMeth}. Consequences for inference are discussed in Section \ref{sec:Inference}, which include methods of combining heterogeneous studies and a generalized nonparametric Stein-prediction formula that selectively borrows strength from `similar' experiments in an automated manner. Section \ref{sec:InfLearningArs} describes a new theory of `learning from uncertain data,' which is an important problem in many application fields including metrology, physics, and chemistry. Section \ref{sec:PoiSmo} solves a long-standing puzzle of modern empirical Bayes,  originally posed by Herbert Robbins \cite{robbins1980}. We conclude the paper with some final remarks in Section \ref{sec:Disc}. Connections with other Bayesian cultures are presented in the supplementary material to ensure the smooth flow of main ideas.  
\vskip.5em
{\bf Real-data Applications.} To demonstrate the versatility of the proposed ``Bayes \textit{via} goodness-of-fit'' data analysis scheme, we selected examples from a wide range of models including normal, Poisson, and Binomial distributions. The full catalog of datasets is presented in Supplementary Table \ref{tbl:DataSummary}.
\vskip.5em
{\bf Notation.} The notation $g$ and $G$ denote the density and distribution function of the starting prior, while $\pi$ and $\Pi$ denote the density and distribution function of the unknown oracle prior. We will denote the conjugate prior with hyperparameters $\al$ and $\be$ by $g(\te;\al,\be)$. Let $\sL^2(\mu)$ be the space of square integrable functions with inner product $\int f(u) g(u)\dd\mu(u)$. $\Leg_j(u)$ denotes $j$th shifted orthonormal Legendre polynomials on $[0,1]$. They form a complete orthonormal basis for $\sL^2(0,1)$. Whereas $T_j(\te;G) := \Leg_j [G(\te)]$ is the modified shifted Legendre polynomials of rank-G transform $G(\te)$, which are basis of the Hilbert space $\sL^2(G)$. The composition of functions is denoted by the usual `$\circ$' sign. 
\section{The Model} \label{sec:Model}
Our model-building approach proceeds sequentially as follows: (i) it starts with a scientific (or empirical) parametric prior $g(\te;\al,\be)$, (ii) inspects the adequacy and the remaining uncertainty of the elicited prior using a graphical exploratory tool, (iii) estimates the necessary ``correction'' for assumed $g$ by looking at the data, (iv) generates the final statistical estimate $\hat \pi(\te)$, and (v) executes macro and micro-level inference. We seek a method that can yield answers to all five of the phases using only a \textit{single} algorithm.

\subsection{New Family of Prior Densities}\label{sec:DSGm}
This section serves two purposes: it provides a universal class of prior density models, followed by its Fourier non-parametric representation in a specialized orthonormal basis.

\begin{defn}
The Skew-G class of density models is given by
\beq \label{eq:DP-prior} \pi(\theta) = g(\theta; \alpha, \beta)\,d[G(\theta);G,\Pi],\eeq
where $d(u;G,\Pi)=\pi(G^{-1}(u))/g(G^{-1}(u))$ for $0<u<1$ and consequently $\int_0^1 d(u;G,\Pi)=1$.
\end{defn}
A few notes on the model specification:
\begin{itemize}[itemsep=1pt,topsep=1.24pt]
\item It has a unique \textit{two-component} structure that combines assumed parametric $g$ with the $d$-function. The function $d$ can be viewed as a ``correction'' density to counter the possible misspecification bias of $g$. 
\item The density function $d(u;G,\Pi)$ can also be viewed as describing the ``excess'' \textit{uncertainty} of the assumed $g(\te;\al,\be)$. For that reason we call it the U-function.
\item The motivation behind the representation \eqref{eq:DP-prior} stems from the observation that $d[G(\te);G,\Pi]$ is in fact the prior density-ratio $\pi(\te)/g(\te)$. Hence, it is straightforward to verify that the scheme \eqref{eq:DP-prior} always yields a proper density, i.e., $\int_\te g(\theta)\,d[G(\theta);G,\Pi]=1$.
\end{itemize}

Since the square integrable $d[G(\te);G,\Pi]$ lives in the Hilbert space $\sL^2(G)$, we can approximate it by projecting into the orthonormal basis $\{T_j\}$ satisfying $\int T_i(\te;G) T_j(\te;G)\dd G=\delta_{ij}$. We choose $T_j(\te;G)$ to be $\Leg_j \hspace{-.08em}\circ \,G(\te)$, a member of the LP-class of rank-polynomials \cite{mukhopadhyay2014lp}. The system $\{T_j\}$ possesses two attractive properties: they are polynomials of rank transform $G(\te)$ thus constitutes a robust basis, and they are orthonormal with respect to $\sL^2(G)$, for \textit{any} arbitrary $G$ (continuous). This is not to be confused with standard Legendre polynomials $\Leg_j(u), 0 <u<1$, which are orthonormal with respect to $\mathrm{Uniform}[0,1]$ measure. For more details, see Supplementary Appendix B. The above discussion paves the way for the following definition. 
\begin{defn}
$\Theta \sim \DS(G,m)$ distribution if it admits the following representation: 
\beq \label{eq:LPds}
\pi(\te)\,=\,g(\te;\al,\be)\, \Big[ 1+\sum_{j=1}^m \LP[j;G,\Pi] \,T_j(\te;G)   \Big].
\eeq
The LP-Fourier coefficients $\LP[j;G,\Pi]$ are the key parameters that help us to express mathematically the ``gap'' between a priori anticipated $G$ and the true prior $\Pi$. When all the expansion coefficients are zero, we automatically recover $g$.
\end{defn}
We will now spend a few words on the LP-DS($G,m$) class of prior models:
\begin{itemize}[itemsep=1pt,topsep=1.24pt]
\item When $\pi(\te)$ is a member of $\DS(G,m)$ class of priors, the orthogonal LP-transform coefficients \eqref{eq:LPds} satisfy 
 \beq \label{eq:lpt}\LP[j;G,\Pi]\,=\,\langle d, T_j\circ G^{-1}\rangle_{\sL^2(0,1)}\,=\,\Ex[T_j(\Theta;G);\Pi].\eeq
 Thus, given a random sample $\te_1,\ldots,\te_k$  from $\pi(\te)$, we could easily estimate the unknown LP-coefficients, and, thus, $d$ and $\pi$, by computing the sample mean $k^{-1}\sum_{i=1}^k T_j(\te_i;G)$. \textit{But unfortunately, the $\te_i$'s are unobserved}. Section \ref{sec:EstMeth} describes an estimation strategy that can deal with the situation at hand. Before introducing this technique, however, we must acclimate the reader with the role played by the U-function $d(u;G,\Pi)$ for uncertainty quantification and characterization of the initial believable prior $g$. That's the objective of the next Section \ref{sec:Ufunc}.
\item Under definition 2, we have $\DS(G,m=0) \equiv g(\te;\al,\be)$. The truncation point $m$ in \eqref{eq:LPds} reflects the \textit{concentration} of permissible $\pi$ around a known $g$. While this class of priors is rich enough to approximate any reasonable prior with the desired accuracy in the large-$m$ limit, one can easily exclude absurdly rough densities and focus on a neighborhood around the domain-knowledge-based $g$ by choosing $m$ not ``too big.''
\item The motivations behind the name `DS-Prior' are twofold.  First, our formulation operationalizes  I. J. Good's `\underline{S}uccessive \underline{D}eepening' idea \cite{good1983EDA} for Bayesian data analysis:
\begin{quote}
\vspace{-.25em}
 \begin{spacing}{1.1}
{\footnotesize \textrm{\textit{A hypothesis is formulated, and, if it explains enough, it is judged to be probably approximately correct. The next stage is to try to improve it. The form that this approach often takes in EDA is to examine residuals for patterns, or to treat them as if they were original data}} (I. J. Good, 1983, p. 289).}
\end{spacing}
\vspace{-.25em}
\end{quote}
Secondly, our prior has two components: A \underline{S}cientific $g$ that encodes an expert's knowledge and a \underline{D}ata-driven $d$. That is to say that our framework embraces data and science, both, in a \textit{testable} manner \cite{gelman2017prior}.
\end{itemize}
 \subsection{Exploratory Diagnostics and U-Function} \label{sec:Ufunc}
Is your data compatible with the pre-selected $g(\te)$? If yes, the job is done without getting into the arduous business of nonparametric estimation. If no, we can model the ``gap'' between the parametric $g$ and the true unknown prior $\pi$, which is often \textit{far easier} than modeling $\pi$ from scratch (hence, one can learn from small number of cases)! If the observed $y_1,\ldots,y_k$ look very unexpected given $g(\theta;\al,\be)$, it is completely reasonable to question the sanctity of such a self-selected prior. Here we provide a formal nonparametric exploratory procedure to describe comprehensively the uncertainty about the choice of $g$.  Using the algorithm detailed in the next section, we estimate U-functions for four real data sets.  Among them, the first three are binomial variate and the last one normal. The results are shown in Fig. \ref{fig:sec3_3CD}. 
\begin{itemize}[itemsep=1pt,topsep=1.24pt]
\item The rat tumor data \cite{gelman2013bayesian} consists of observations of endometrial stromal polyp incidence in $k=70$ groups of female rats.  For each group, $y_i$ is the number of rats with polyps and $n_i$ is the total number of rats in the experiment. \vspace{-.15em}
\item The terbinafine data \cite{young2008pooling} comprise $k=41$ studies, which investigate the proportion of patients whose treatment terminated early due to some adverse effect of an oral anti-fungal agent: $y_i$ is the number of terminated treatments and $n_i$ is the total number of patients in the experiment. \vspace{-.15em}
\item The rolling tacks \cite{beckett1994spectral} data involve flipping a common thumbtack $9$ times.  It consists of $320$ pairs, $(9, y_i)$, where $y_i$ represents the number of times the thumbtack landed point up. 
\item The ulcer data consist of forty randomized trials of a surgical treatment for stomach ulcers conducted between 1980 and 1989 \cite{sacks1990endoscopic,efron1996empirical}.  Each of the $40$ trials has an estimated log-odds ratio $y_i|\te_i \sim \cN(\te_i,s_i^2)$ that measures the rate of occurrence of recurrent bleeding given the surgical treatment. 
\end{itemize}
Throughout, we have used the  maximum likelihood estimates (MLE) for estimating the initial starting value of the hyperparameters. However, one can use any other reasonable choice, which may involve expert's judgment. What is important to note is the \textit{shape} of the $\whd$; more specifically, its departure from uniformity, indicates the assumed conjugate prior $g(\te;\al,\be)$ needs a `repair' to resolve the prior-data conflict. For example, the flat shape of the estimated $\whd$ in Fig. \ref{fig:sec3_3CD}(b)  indicates that our initial selection of $g(\theta; \alpha,\beta)$ is appropriate for the terbinafine and ulcer data.  Therefore, one can proceed in turning the ``Bayesian crank'' with confidence using the parametric beta and normal prior respectively.
\begin{figure}[t]
\begin{subfigure}{.33\textwidth}
  \centering
  \includegraphics[width=\linewidth]{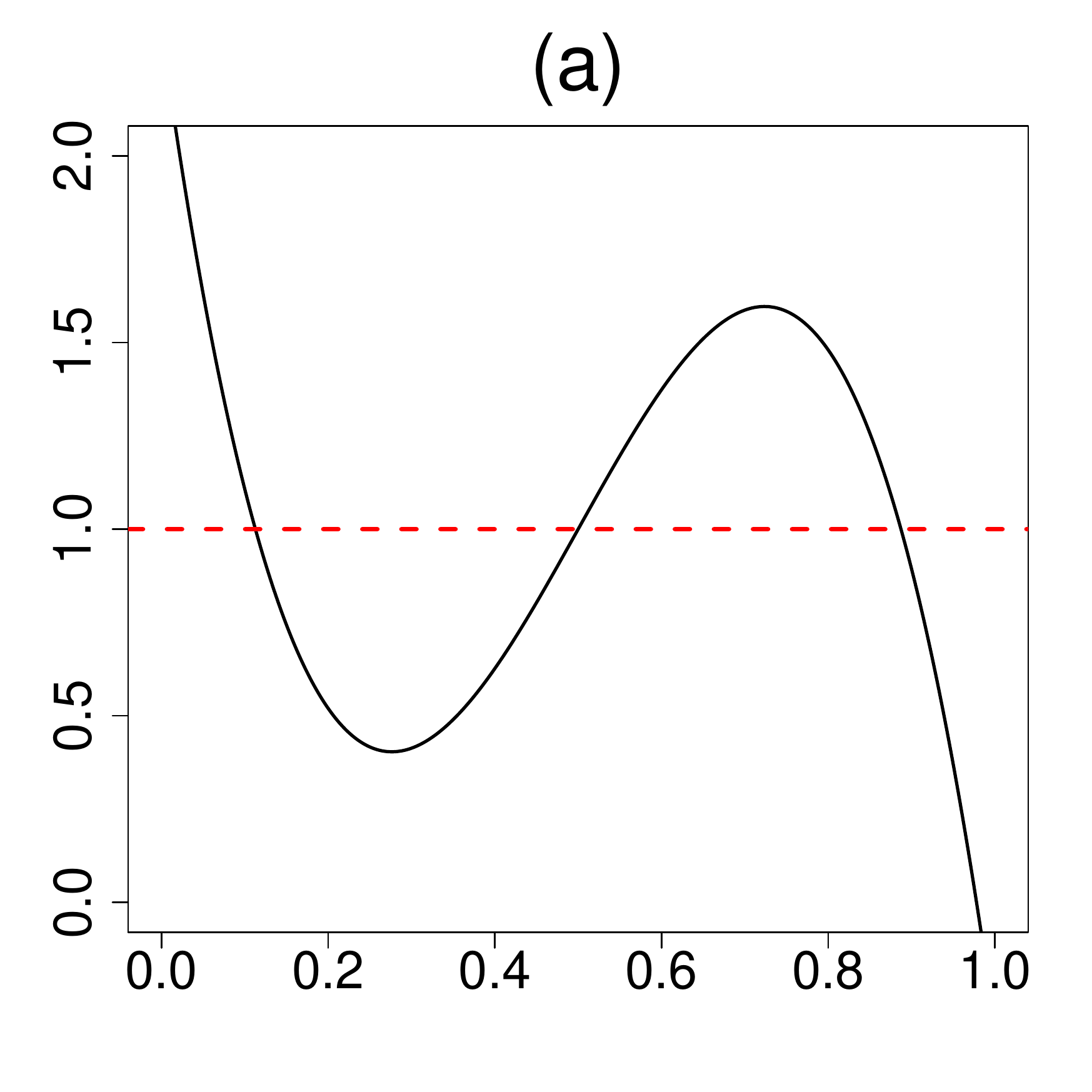}
\end{subfigure}\hspace{1mm}%
\begin{subfigure}{.33\textwidth}
  \centering
\includegraphics[width=\linewidth]{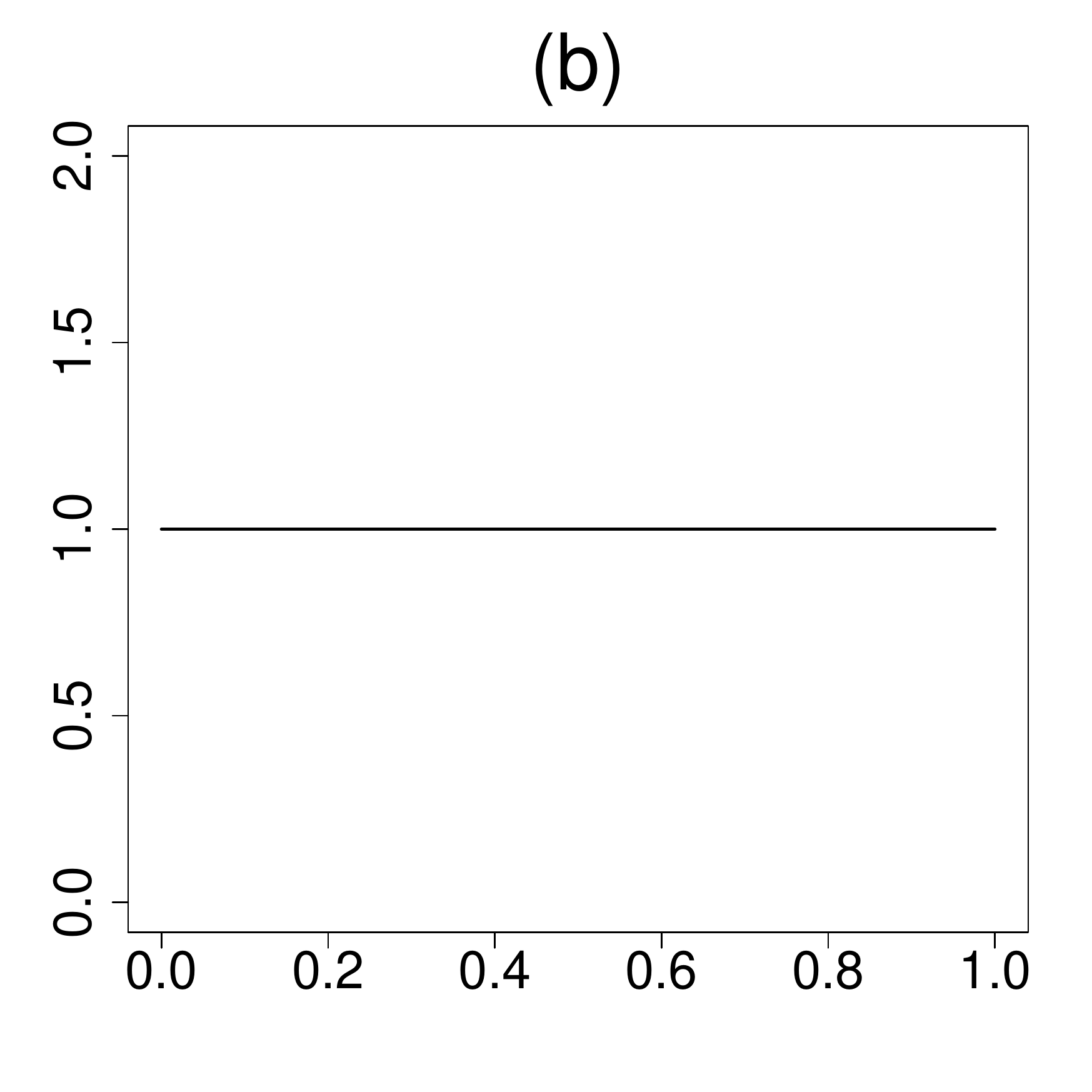}
\end{subfigure}\hspace{1mm}%
\begin{subfigure}{.33\textwidth}
  \centering
  \includegraphics[width=\linewidth]{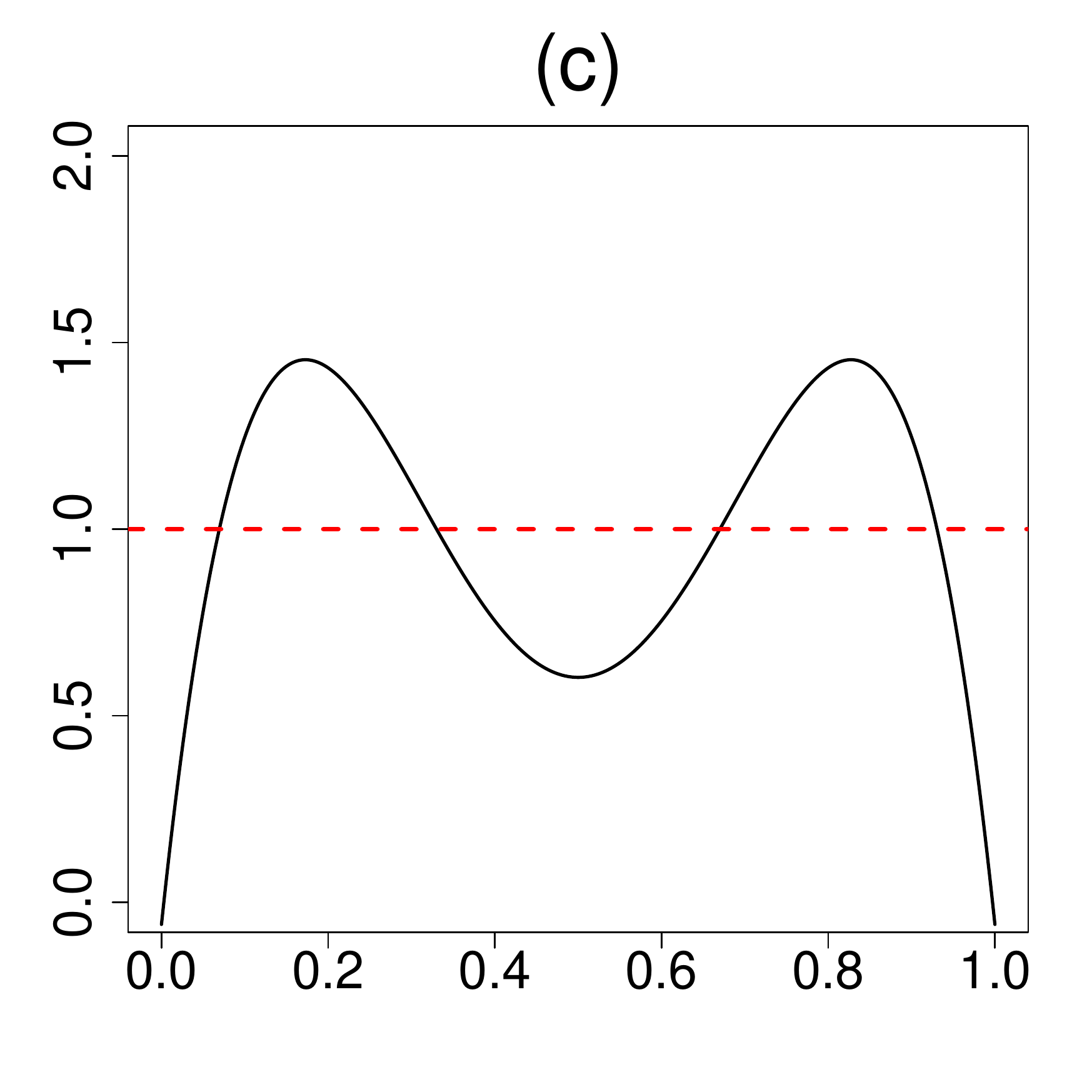}
\end{subfigure}
\caption{Graphical diagnostic tool: U-functions for (a) rat tumor data; (b) terbinafine and  ulcer data; (c) rolling tacks data. The deviation from uniformity (red dotted line) indicates that the default prior contradicts the observed data. The flat shape of the U-function in panel (b) suggests ${\rm Beta}(1.24, 34.7)$  and $\cN(-1.17, 0.98)$ are consistent with the terbinafine and ulcer data, respectively.}
\label{fig:sec3_3CD}
\end{figure}
\vskip.4em
In contrast, Figs. \ref{fig:sec3_3CD}(a,c) provide a strong warning in using $g={\rm Beta}(\al,\be)$ for the rat tumor and the rolling tacks experiments. The smooth estimated U-functions expose the nature of the discrepancy that exists between $g$ and the observed data by having an ``extra'' mode. Clearly, the answer does not lie in choosing a different $(\al,\be)$ as this cannot rectify the missing bimodality. This brings us to an important point: the full Bayesian analysis, by assigning hyperprior distribution on $\al$ and $\be$, is not always a fail-safe strategy and should be practiced with caution (not in a blind mechanical way). The bottom line is uncertainty in the prior probability model $\neq$ uncertainty in $\al, \be$. A foolproof prior uncertainty model, thus, has to allow ignorance in terms of the \textit{functional shape} around $g$. The foregoing discussion motivates the following entropy-like measure of uncertainty.
\begin{defn}
The $q\LP$ statistic for uncertainty quantification is defined as follows:
\beq \label{eq:qlp}
{\rm qLP}(G||\Pi)=\sum_j\big| \LP[j;G,\Pi]\big|^2.
\eeq
\end{defn}
The motivation  behind this definition comes from applying Parseval's identity in \eqref{eq:LPds}: $\int_0^1 d^2(u;G,\Pi)= 1 + {\rm qLP}(G||\Pi)$. Thus, the proposed measure captures the departure of the U-function from uniformity. The following result connects our $q\LP$ statistic with relative entropy.
\begin{thm} The $q\LP$ uncertainty quantification statistic satisfies the following relation:
\beq {\rm qLP}(G||\Pi) ~\approx~ 2 \times {\rm KL}(\Pi||G), \eeq
where   ${\rm KL}(\Pi||G)$ is the Kullback–Leibler (KL) divergence between the true prior $\pi$ and its parametric approximate $g$.
\end{thm}
\begin{proof}
Express KL-information divergence using U-functions by substituting $G(\te)=u$:
\beq \label{eq:KLp}{\rm KL}(\Pi||G)\,=\int \pi(\te) \log \frac{\pi(\te)}{g(\te)} \dd \te \,= \int_0^1 d(u;G,\Pi) \log d(u;G,\Pi) \dd u.\eeq
Complete the proof by approximating $d \log d$ in \eqref{eq:KLp} via Taylor series $(d-1) + \frac{1}{2} (d-1)^2$.
\end{proof}
We conclude this section with a few additional remarks:
\begin{itemize}
\item Our exploratory uncertainty diagnostic tool encourages ``interactive'' data analysis that is similar in spirit to Gelman et al.\cite{gelman1996PDC}. Subject-matter experts can use this tool to ``play'' with different hyperparameter choices in order to filter out the reasonable ones. This functionality might be especially valuable when multiple expert opinions are available.
\item When $\whd$ shows evidence of the prior-data conflict, the question remains: what to do next? It is not enough to check the adequacy without informing the user an explanation for the misfit or what is the ``deeper'' structure that is missing in the starting parametric prior. Fortunately, our $\DS(G,m)$ model suggests a simple, yet formal, guideline for upgrading: $\widehat{\pi}(\te)=g(\te;\hat \al, \hat \be) \times \whd[G(\te);G,\Pi]$, where the shape of $\whd(u;G,\Pi)$ captures the patterns which were not a priori anticipated. Hence our formalism \textit{simultaneously} addresses the problem of uncertainty quantification and the subsequent model synthesis.
\end{itemize}
\section{Estimation Method}\label{sec:EstMeth}
\subsection{Theory}\label{sec:EstTh}
In this Section, we lay out the key theoretical results that we use for designing our algorithm. Before deriving the general expressions under the LP-$\DS(G,m)$ model, it is helpful to start by recalling the results for the basic conjugate model, i.e., $\Theta \sim \DS(G,m=0)$ and $y_i|\te_i ~\overset{{\rm ind}}{\sim} ~ f(y_i|\te_i)$  for $i=1,\ldots,k$. Table \ref{tbl:familyconj} provides the marginal $f_{G}(y_i)=\int_{\te_i} f(y_i|\te_i) g(\te_i)\dd \te_i$ and the posterior distribution $\pi_{G}(\te_i|y_i)=\frac{f(y_i|\te_i)g(\te_i)}{f_G(y_i)}$ for four commonly encountered distributions, with the Bayes estimate of $h(\Te_i)$ being denoted as $\Ex_G\big [h(\Theta_i)|y_i\big] = \int_{\te_i} h(\te_i) \pi_{G}(\te_i|y_i) \dd \te_i$. The subscript `$G$' in these expressions underscores the fact that they are calculated for the conjugate $g$-model. 

\begin{table}[H]
\setlength{\tabcolsep}{8pt}
\def\arraystretch{1.2}
\centering
\caption{\label{tbl:familyconj} Details on the distributions, their conjugate priors, and the resulting marginal and posterior distributions for four familiar distributions (two discrete and two continuous): Binomial, Poisson, Normal, and Exponential. For the normal-normal posterior $\lambda_i = \sigma^2_i / (\sigma^2_i + \beta^2)$ and in the marginal of the Poisson-gamma $p = 1/(1+\beta)$. We use $\B(\al,\be)=\frac{\Gamma(\al) \Gamma(\be)}{\Gamma(\al+\be)}$ to denote the normalizing constant of beta distribution.}
{\small
\begin{tabular}{cccc}
\toprule
Family~~ & Conjugate $g$-prior & Marginal [$f_G(y_i)$] & Posterior [$\pi_G(\theta_i \mid y_i)$]\\ 
\midrule
${\rm Binomial}(n_i,\theta_i)$ &  ${\rm Beta}(\alpha, \beta)$ & $\binom{n_i}{y_i} \frac{\B(\alpha+y_i,\beta-y_i+n_i)}{\B(\alpha,\beta)}$  & ${\rm Beta}(\alpha+y_i,\beta-y_i+n_i)$\\[.2cm]
${\rm Poisson}(\theta_i)$ & ${\rm Gamma}(\alpha, \beta)$ & $\binom{y_i + \alpha -1}{y_i}p^\alpha(1-p)^{y_i}$ & $ {\rm Gamma}\big(\al + y_i, \frac{\beta}{1+\beta}\big)$\\[.2cm]
${\rm Normal}(\theta_i,\sigma^2_i)$ & $\rm{ Normal}(\alpha,\beta^2)$ & ${\rm Normal}(\alpha, \sigma^2_i + \beta^2)$   & ${\rm Normal}(\lambda_i \alpha + (1-\lambda_i)y_i, (1-\lambda_i)\sigma^2_i)$\\[.2cm]
${\rm Exp}(\la)$ & ${\rm Gamma}(\alpha, \beta)$ & $\frac{\al\be}{(1+ \be y)^{\al+1}}$ & ${\rm Gamma}\big(\alpha+1, \frac{\beta}{1+\beta y_i}\big)$\\
\bottomrule
\end{tabular}
}
\end{table}
Next, we seek to extend these parametric results to LP-nonparametric setup in a systematic way. Especially, without deriving analytical expressions for each case separately, we want to establish a more general representation theory that is valid for all of the above and, in fact, extends to any conjugate pairs, explicating the underlying unity of our formulation.
\vskip1em
\begin{thm} \label{thm:r1}
Consider the following model:
\beas
y_i|\theta_i &~\overset{{\rm ind}}{\sim}& ~ f(y_i|\te_i),~~~~(i=1,\ldots,k)\\
\Theta_i &~\overset{{\rm ind}}{\sim}& ~ \pi(\theta),
\eeas
where $\pi(\theta)$ is a member of ${\rm DS}(G,m)$ family \eqref{eq:LPds}, $G$ being the associated conjugate prior.  Under this framework, the following holds:

\begin{enumerate}[label = (\alph*)]

\item The marginal distribution of $y_i$ is given by
\beq \label{thm:marginal}
f_{\LP}(y_i) \,=\, f_G(y_i) \left(1+ \sum_j \LP[j;G,\Pi] ~\Ex_G[T_j(\Te_i;G)|y_i]  \right),
\eeq
where $\Ex_G[T_j(\Te_i;G)|y_i]= \int_{\te_i} \Leg_j \hspace{-.08em}\circ \,G(\te_i) \pi_G(\te_i|y_i) \dd \te_i$.

\item A closed-form expression for the posterior distribution of $\Theta_i$ given $y_i$ is
\beq \label{thm:post}
\pi_{\LP}(\te_i|y_i) = \dfrac{\pi_G(\te_i|y_i) \big(1+\sum_j \LP[j;G,\Pi] ~T_j(\te_i;G)\big)}{1+ \sum_j \LP[j;G,\Pi] ~\Ex_G[T_j(\Te_i;G)|y_i]}
\eeq

\item For any general random variable $h(\Theta_i)$, the Bayes conditional mean estimator can be expressed as follows:
\begin{equation} \label{eq:pm}
\Ex_{\LP}[h(\Theta_i)|y_i] = \frac{\Ex_G[h(\Theta_i)|y_i]+\sum_j \LP[j;G,\Pi]\,\Ex_G[h(\Theta_i)T_j(\Theta_i;G)| y_i]}{1+\sum_j\LP[j;G,\Pi]\, \Ex_G[T_j(\Te_i;G)|y_i]}
\end{equation}
\end{enumerate}
\end{thm} 

\begin{proof}
The marginal distribution for $\DS(G,m)$-nonparametric model can be represented as:
\[f_{\LP}(y_i)\,=\,\int f(y_i|\te_i) \times \big\{g(\te_i;\al,\be) \,d[G(\te_i);G,\Pi]\big \} \dd \te_i.\]
Expanding the U-function in the LP-bases \eqref{eq:LPds} yields
\beq \label{eq:pa}
f_{\LP}(y_i)\,=\,f_{G}(y_i)\, +\, \sum_j \LP[j;G,\Pi]\int T_j(\te_i;G) f(y_i|\te_i) g(\te_i;\al,\be) \dd \te_i.
\eeq
The next step is to recognize that 
\beq  \label{thm:key}
f(y_i|\te_i) \,g(\te_i;\al,\be)\,=\,f_{G}(y_i)\, \pi_G(\te_i|y_i).
\eeq
Substituting \eqref{thm:key} in the second term of \eqref{eq:pa} leads to 
\beq  \label{eq:pa2}
\sum_j \LP[j;G,\Pi]\int T_j(\te_i;G) f(y_i|\te_i) g(\te_i;\al,\be) \dd \te_i\,=\,f_{G}(y_i) \sum_j \LP[j;G,\Pi]\, \Ex_G[T_j (\Te_i;G)|y_i].
\eeq
Complete the proof of part (a) by replacing \eqref{eq:pa2} into \eqref{eq:pa}.
\vskip.4em
For part (b) of posterior distribution calculation we have
\beq \label{thm-pb}
\pi_{\LP}(\te_i|y_i) = \dfrac{f(y_i|\te_i)\,g(\te_i;\al,\be)}{f_{\LP}(y_i)} \Big\{1 + \sum_j \LP[j;G,\Pi] T_j(\te_j;G)\Big\}.
\eeq
Combine \eqref{thm:marginal} and \eqref{thm:key} to verify that
\beq \label{thm:pb2} \dfrac{f(y_i|\te_i)\,g(\te_i;\al,\be)}{f_{\LP}(y_i)} \,=\,\dfrac{\pi_G(\te_i|y_i)}{1+ \sum_j \LP[j;G,\Pi] ~\Ex_G[T_j(\Te_i;G)|y_i]}.\eeq
Finish the proof of part (b) by replacing \eqref{thm:pb2} into \eqref{thm-pb}.
\vskip.4em
Part (c) is straightforward as 
\[\Ex_{\LP}[h(\Theta_i)|y_i] = \int h(\te_i) \pi_{\LP}(\te_i|y_i) \dd \te_i,\]
which is same as 
\[ \dfrac{\int h(\te_i) \pi_G(\te_i|y_i)\big\{1 + \sum_j \LP[j;G,\Pi] T_j(\te_j;G)\big\}\dd \te_i}{1+ \sum_j \LP[j;G,\Pi] ~\Ex_G[T_j(\Te_i;G)|y_i]},\]
by \eqref{thm:post}. Hence, result \eqref{eq:pm} is immediate.
\end{proof}
Our LP-Bayes recipe \eqref{thm:marginal}-\eqref{eq:pm}, admits some interesting overall structure: The usual `parametric' answer multiplied by a correction factor involving $\LP[j;G,\Pi]$'s. This decoupling pays dividends for theoretical interpretation as well as computation.

\subsection{Algorithm}\label{sec:EstAlgo}
The critical parameters of our $\DS(G,m)$ model are the LP-Fourier coefficients, which, as is evident from \eqref{eq:lpt}, could be estimated simply by their empirical counterpart $\widehat{\LP}[j;G,\Pi] = k^{-1}\sum_{i=1}^k T_j(\te_i;G)$. But as we pointed out earlier, $\te_1,\ldots,\te_k$ are unobservable. How can we then estimate those parameters? While the $\te_i$'s are \textit{unseen}, it is interesting to note that they have left their footprints in the observables $y_1,\ldots,y_k$ with distribution $f(y_i)=\int f(y_i|\te_i) \pi(\te_i)\dd \te_i$. Following the spirit of the EM-algorithm, an obvious proxy for $T_j(\te_i;G)$ would be its posterior mean $\Ex_{\LP}[T_j(\Theta_i;G)|y_i]$, which also naturally arises in the expression \eqref{thm:marginal}. This leads to the following `ghost' LP-estimates:  
\beq \label{eq:Glp}
\tLP[j;G,\Pi]\,=k^{-1}\sum_{i=1}^k \Ex_{\LP}\big[T_j(\Te_i;G)|y_i\big],
\eeq
satisfying $\Ex\{\tLP[j;G,\Pi]\}=\hLP[j;G,\Pi]\,\, (j=1\ldots, m)$, by virtue of the law of iterated expectations. These estimates can then be refined via iterations. 
\vskip1em
\makebox[\textwidth]{\textbf{Type-II Method of Moments: Estimation of LP-Coefficients in $\DS(G,m)$}}
\rule{\textwidth}{.8pt}
\vskip.1em
\texttt{Step 0.} Input: Data $(y_1,\ldots,y_k)$ and $m$. Choice of $\al$ and $\be$: based on expert's knowledge, otherwise, we use MLE empirical estimate as our default starting choice.
\vskip.7em
\texttt{Step 1.} Initialize: $\LP^{(0)}[j;G,\Pi]=0$ for $j=1,\ldots,m$. For iteration $\ell > 0$, perform steps (2-3) until convergence: $ \sum_{j=1}^m \big | \tLP^{(\ell)}[j;G,\Pi] - \tLP^{(\ell-1)}[j;G,\Pi] \big|^2 \le \epsilon$.
\vskip.7em

\texttt{Step 2.}  Compute $\Ex_{\{\ell-1\}}[T_j(\Te_i;G)|y_i]$ by plugging $\big\{\hspace{-.2em}\tLP^{(\ell-1)}[j;G,\Pi]\big\}_{j=1}^m$ into \eqref{eq:pm}, where $h(\te_i)=\Leg_j \hspace{-.08em}\circ \,G(\te_i)$. 
\vskip.7em
\texttt{Step 3.} Determine the `ghost' LP-estimates: \[\tLP^{(\ell)}[j;G,\Pi] ~= ~k^{-1}\sum\nolimits_{i=1}^k \Ex_{\{\ell-1\}}[T_j(\Te_i;G)|y_i]~~~~(j=1,\ldots,m).\] 
\vskip.1em
\texttt{Step 4.} Return the final estimated LP-coefficients of $\DS(G,m)$ model together with $\whd(u;G,\Pi)$ and $\widehat{\pi}(\te)$.\\
\rule{\textwidth}{.8pt}
\vskip.5em
We conclude this section with a few remarks on the algorithm:
\begin{itemize}[itemsep=1pt,topsep=1.24pt]
\item Taking inspiration from I. J. Good's type II maximum likelihood nomenclature \cite{good1983book}, we call our algorithm \textit{Type-II} Method of Moments (MOM), whose computation is remarkably tractable and does not require \textit{any} numerical optimization routine.
\item To enhance the results, we smooth the output of MOM-II algorithm as follows:  determine significantly non-zero LP-coefficients via Schwartz's BIC-based smoothing. Arrange $\hLP[j;G,\Pi]$'s in a decreasing magnitude and choose $m$ that maximizes 
\begin{equation*}
\text{BIC}(m) \,=\, \sum_{j=1}^m |\hLP[j;G,\Pi]|^2 - \dfrac{m\log(k)}{k}.
\end{equation*}
See Supplementary Appendix D for more details. Furthermore, Supplementary Appendix I discusses how MOM-II Bayes algorithm can be adapted to yield LP-maximum entropy prior density estimate \cite{lpmode2017}.
\end{itemize}
\subsection{Results}\label{sec:EstResults}
In addition to the rat tumor data (cf. Section \ref{sec:Ufunc}), here we introduce and analyze three additional datasets: two binomial and one Poissonian example.
\begin{itemize}[itemsep=1pt,topsep=1.24pt]
\item The surgical node data \cite{efron2014bayes} involves number of malignant lymph nodes removed during intestinal surgery.  Each of the $k = 844$ patients underwent surgery for cancer, during which surgeons removed surrounding lymph nodes for testing.  Each patient has a pair of data ($n_i, y_i$), where $n_i$ represents the total nodes removed from patient $i$ and $y_i \sim {\rm Bin}(n_i,\te_i)$ are the number of malignant nodes among them. 
\item The Navy shipyard data \cite{martz1974empirical} consists of $k=5$ samples of the number of defects $y_i$ found in $n_i=5$ lots of welding material. 
\item The insurance data \cite{efron2016computer}, shown in Table \ref{tbl:InsMicro}, provides a single year of claims data for an automobile insurance company in Europe. The counts $y_i \sim {\rm Poisson}(\te_i)$ represent the total number of people who had $i$ claims in a single year.
\end{itemize}
Figure \ref{fig:DC_prior_all} displays the estimated LP-$\DS(G,m)$ priors along with the default parametric (empirical Bayes) counterparts. The estimated LP-Fourier coefficients together with the choices of hyperparameters $(\al,\be)$ are summarized below: 
\vspace{-.4em}
\begin{enumerate}[label=(\alph*)]
\item Rat tumor data, $g$ is the beta distribution with MLE $\alpha = 2.30$, $\beta= 14.08$:
\begin{equation}\label{eq:rat_pi_hat}
\hat{\pi}(\theta) = g(\theta; \alpha,\beta)\big[1 - 0.50T_3(\theta;G) \big].
\end{equation}
\item Surgical node data, $g$ is the beta distribution with MLE $\alpha = 0.32$, $\beta= 1.00$:
\begin{equation}\label{eq:surg_pi_hat}
\hat{\pi}(\theta) = g(\theta; \alpha,\beta)\big[1 - 0.07T_3(\theta;G) - 0.11T_4(\theta;G) + 0.09T_5(\theta;G) + 0.13T_7(\theta;G) \big].
\end{equation}
\item Navy shipyard data, $g$ is the Jeffreys prior with $\alpha = 0.5 $, $\beta= 0.5$:
\begin{equation}\label{eq:shp_pi_hat}
\hat{\pi}(\theta) = g(\theta; \alpha,\beta)\big[1 - 0.67T_1(\theta;G) + 0.90T_2(\theta;G) \big].
\end{equation}
\item Insurance data, $g$ is the gamma distribution with MLE $\alpha = 0.70$ and $\beta = 0.31$:
\begin{equation}\label{eq:ins_pi_hat}
\hat{\pi}(\theta) = g(\theta; \alpha,\beta)\big[1 - 0.26T_2(\theta;G)\big].
\vspace{-.4em}
\end{equation}
\end{enumerate}
\begin{figure}[t]
\centering
\begin{subfigure}{.35\textwidth}
  \centering
  \includegraphics[width=\linewidth,trim=1cm 0cm 1cm 1cm]{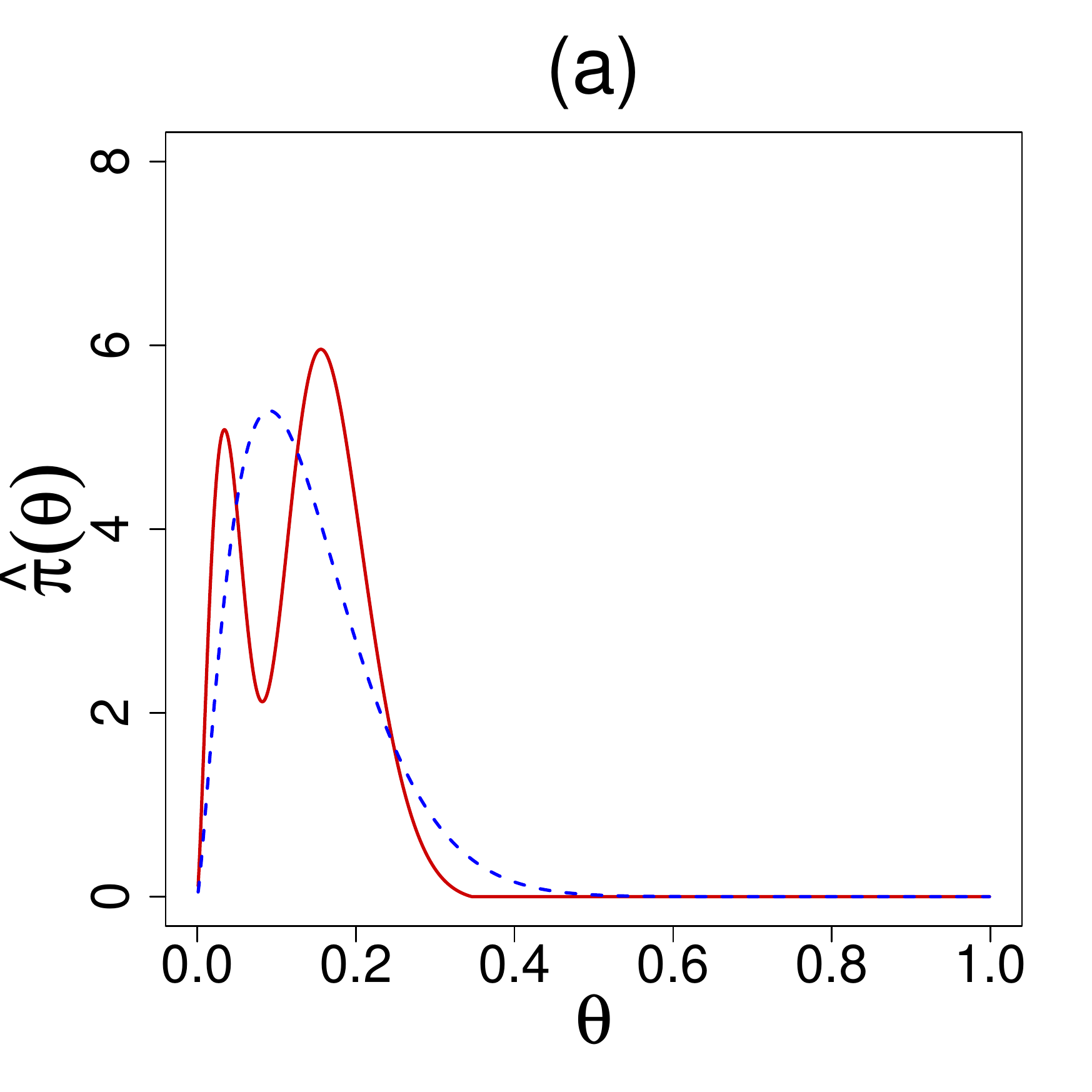}
\end{subfigure}\hspace{1.5cm}%
\begin{subfigure}{.35\textwidth}
  \centering
\includegraphics[width=\linewidth,trim=1cm 0cm 1cm 1cm]{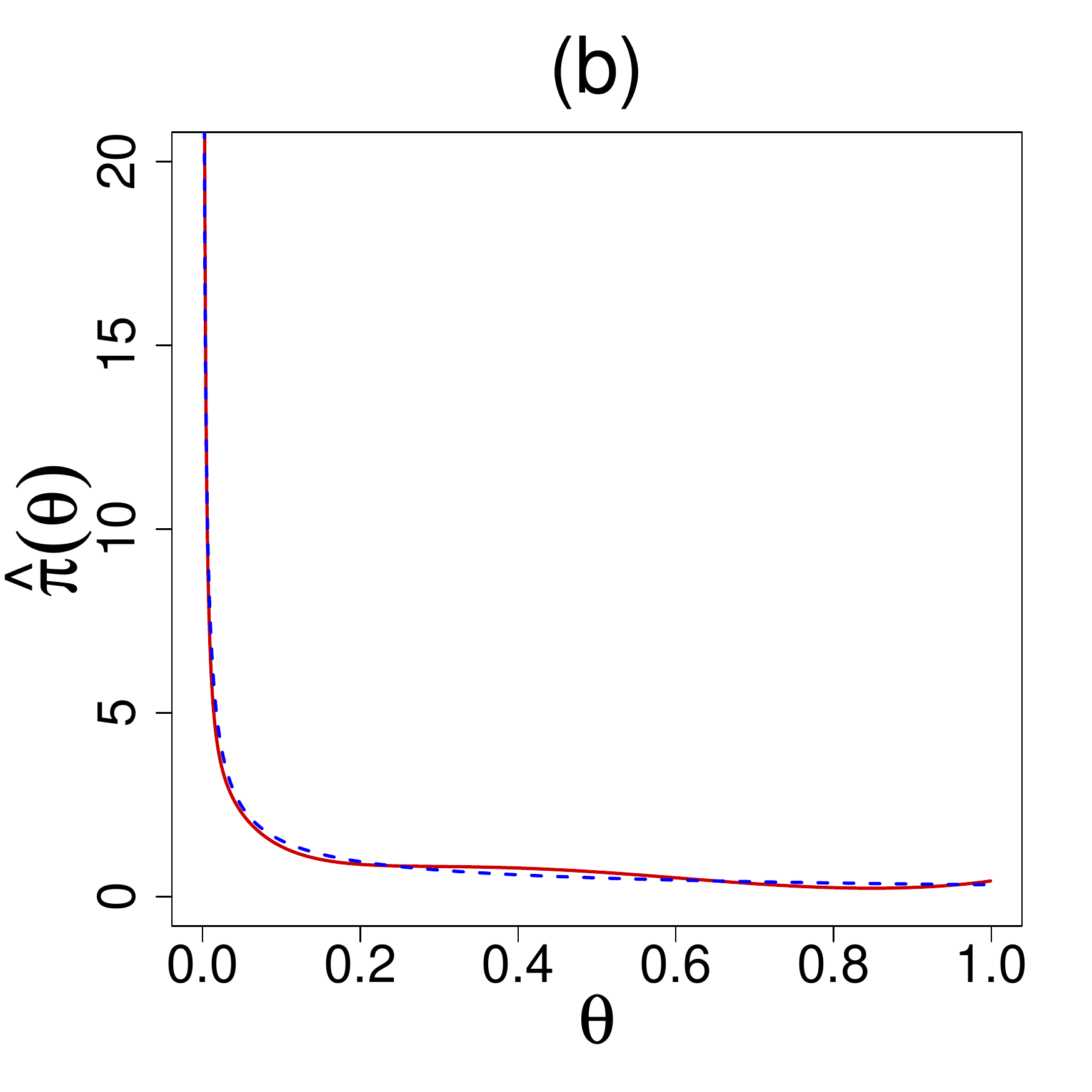}
\end{subfigure}\hspace{1mm}%
\begin{subfigure}{.35\textwidth}
  \centering
  \includegraphics[width=\linewidth,trim=1cm 1cm 1cm 0cm]{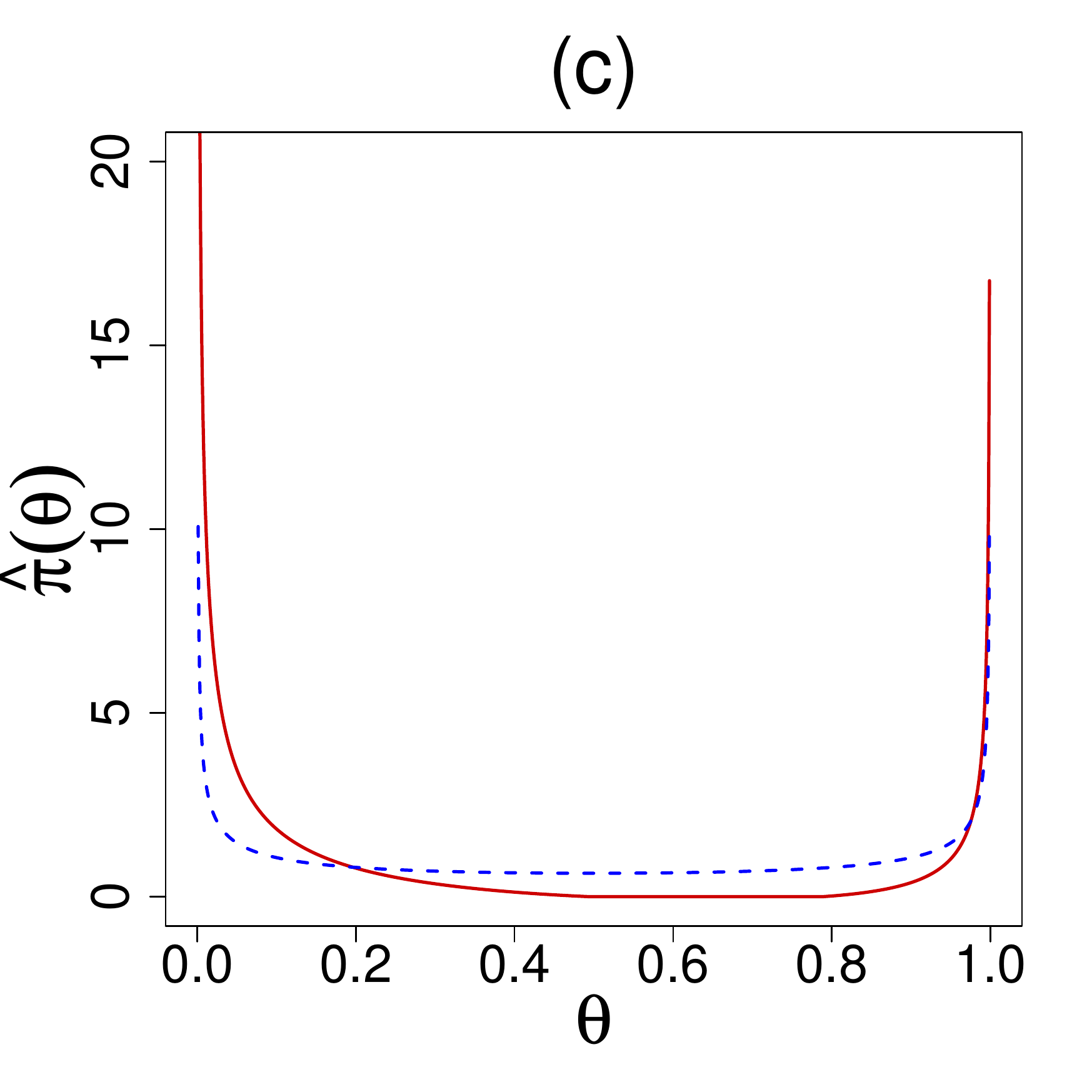}
\end{subfigure}\hspace{1.5cm}%
\begin{subfigure}{.35\textwidth}
  \centering
  \includegraphics[width=\linewidth,trim=1cm 1cm 1cm 0cm]{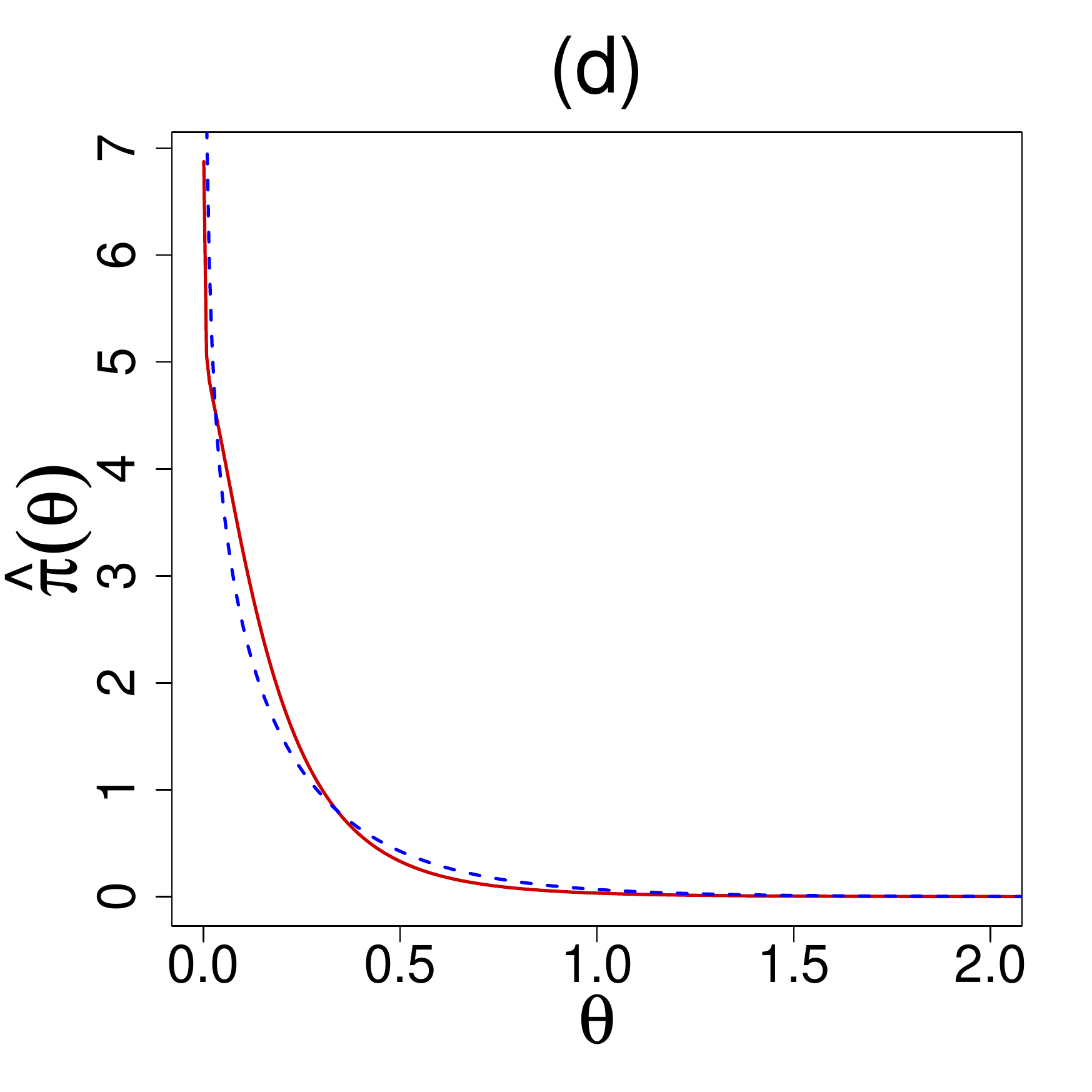}
\end{subfigure}
\vskip.6em
\caption{Comparisons of the $\DS(G,m)$ prior $\hat{\pi}(\theta)$ (solid red) with the respective parametric EB (PEB) priors $g(\theta; \alpha,\beta)$ (dashed blue) for the (a) rat tumor data, (b) surgical node data, (c) Navy shipyard data, and (d) insurance data.}
\label{fig:DC_prior_all}
\end{figure} 
The rat tumor data shows a prominent bimodal shape, which should not come as a surprise in light of Fig. \ref{fig:sec3_3CD}(a). For the surgical data, DS-prior puts excess mass around $0.4$, which concurs with the findings of Efron \cite[Sec 4.2]{efron2014bayes}. In the case of the Navy shipyard data, our analysis corrects the starting ``U'' shaped Jeffreys prior to make it asymmetric with an extended peak at $0$. This is quite justifiable looking at the proportions in the given data: $(0/5, 0/5, 0/5, 1/5, 5/5)$. Finally, for the insurance data, the starting gamma prior requires a second-order (dispersion parameter) correction to yield a bona-fide $\hat \pi$ \eqref{eq:ins_pi_hat}, which makes it slightly wider in the middle with sharper peak and tail.
\section{Inference}\label{sec:Inference}
\subsection{MacroInference}\label{sec:MacInf}
A single study hardly provides adequate evidence for a definitive conclusion due to the limited sample size. Thus, often the scientific interest lies in combining several \textit{related but (possibly) heterogeneous} studies to come up with an overall macro-level inference that is more accurate and precise than the individual studies. This type of inference is a routine exercise in clinical trials and public policy research.
\vskip1em

{\bf Terbinafine data analysis}. For the terbinafine data, the aim is to combine $k=41$ treatment arms with varying event rates and produce a pooled proportion of patients who withdrew from the study because of the adverse effects of oral anti-fungal agents. Recall that our U-function diagnostic in Fig. \ref{fig:sec3_3CD}(b) indicated the parametric beta-binomial model with MLE estimates $\al=1.24$ and $\be=34.7$ as a justifiable choice for this data. Thus the adverse event probabilities across $k=41$ studies can be summarized by the prior mean $\frac{\al}{\al+\be}=.034$. We apply parametric bootstrap using $\DS(G,m)$-sampler (see Supplementary Appendix C) with $m=0$ to compute the standard error (SE): $0.034 \pm 0.006$, highlighted in the Fig. \ref{fig:sec5_mode_int}(b). If one assumes a \textit{single} binomial distribution for all the groups (i.e., under homogeneity), then the `naive' average $\sum_{i=1}^k y_i/\sum_{i=1}^k n_i$ would lead to an overoptimistic biased estimate $0.037 \pm 0.0034$. In this example, heterogeneity arises due to overdispersion among the exchangeable studies. But there could be other ways too. An example is given in the following case study.

\vskip1em
{\bf Rat tumor and rolling tacks data analysis}. Can we always extract a ``single'' overall number to aptly describe $k$ parallel studies? Not true, in general. In order to appreciate this, let us look at Figs. \ref{fig:sec5_mode_int} (a,c), which depict the estimated $\DS$-prior for the rat tumor and rolling tacks data. We highlight two key observations: 
\vskip.75em

~~~~~1. \textit{Mixed population}. The bimodality indicates the existence of two distinct groups of $\te_i$'s. We call this ``\textit{structured heterogeneity},'' which is in between two extremes: homogeneity and complete heterogeneity (where there is no similarity between the $\te_i$'s whatsoever). The presence of two clusters for the rolling tacks data was previously detected by Jun Liu \cite{liu1996nonparametric}. The author further noted,
``Clearly, this feature is unexpected and cannot be revealed by a
regular parametric hierarchical analysis using the Beta-binomial priors.'' One plausible explanation for this two-group structure was attributed to the fact that the tack data were produced by two persons with some systematic difference in their flipping. On the other hand, the bimodal shape of the rat example was not previously anticipated \cite{tarone1982use, dempster1983combining,gelman2013bayesian}. 
\begin{figure}[t]
\begin{subfigure}{.32\textwidth}
  \centering
  \includegraphics[width=\linewidth,trim=1cm .5cm 0cm 1cm]{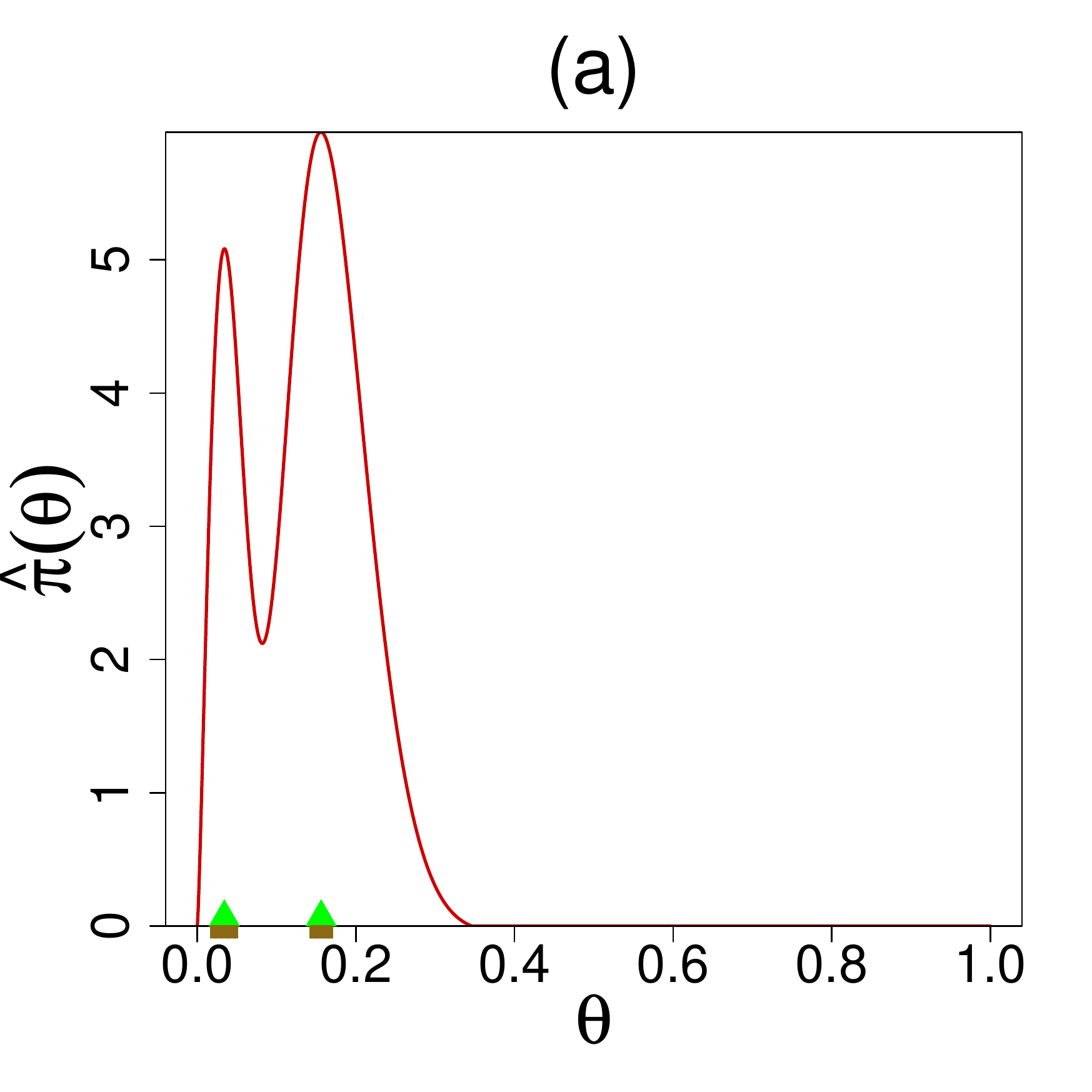}
\end{subfigure}\hspace{1.5mm}%
\begin{subfigure}{.32\textwidth}
  \centering
\includegraphics[width=\linewidth,trim=.5cm .5cm .5cm 1cm]{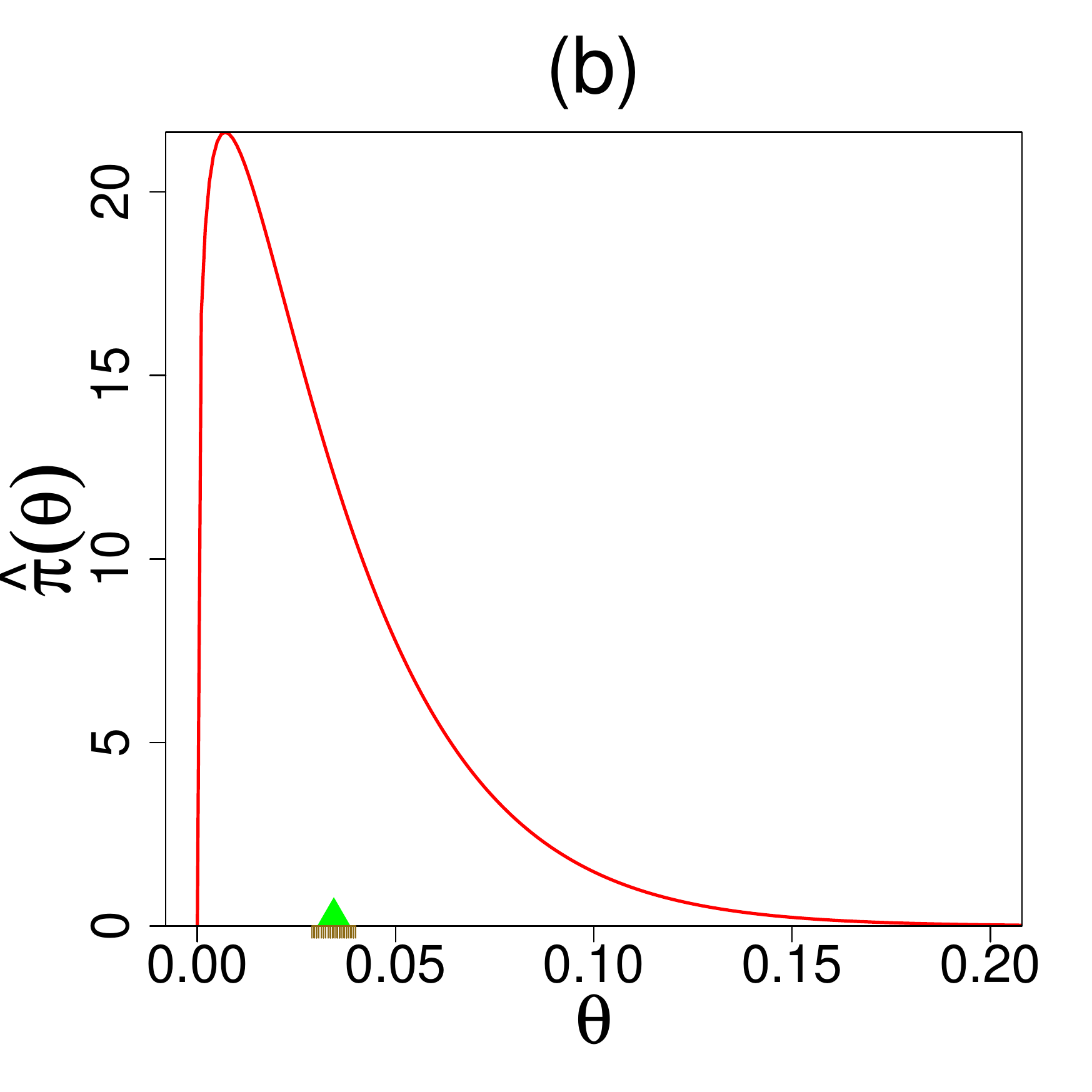}
\end{subfigure}\hspace{1.5mm}%
\begin{subfigure}{.32\textwidth}
  \centering
  \includegraphics[width=\linewidth,trim=0cm .5cm 1cm 1cm]{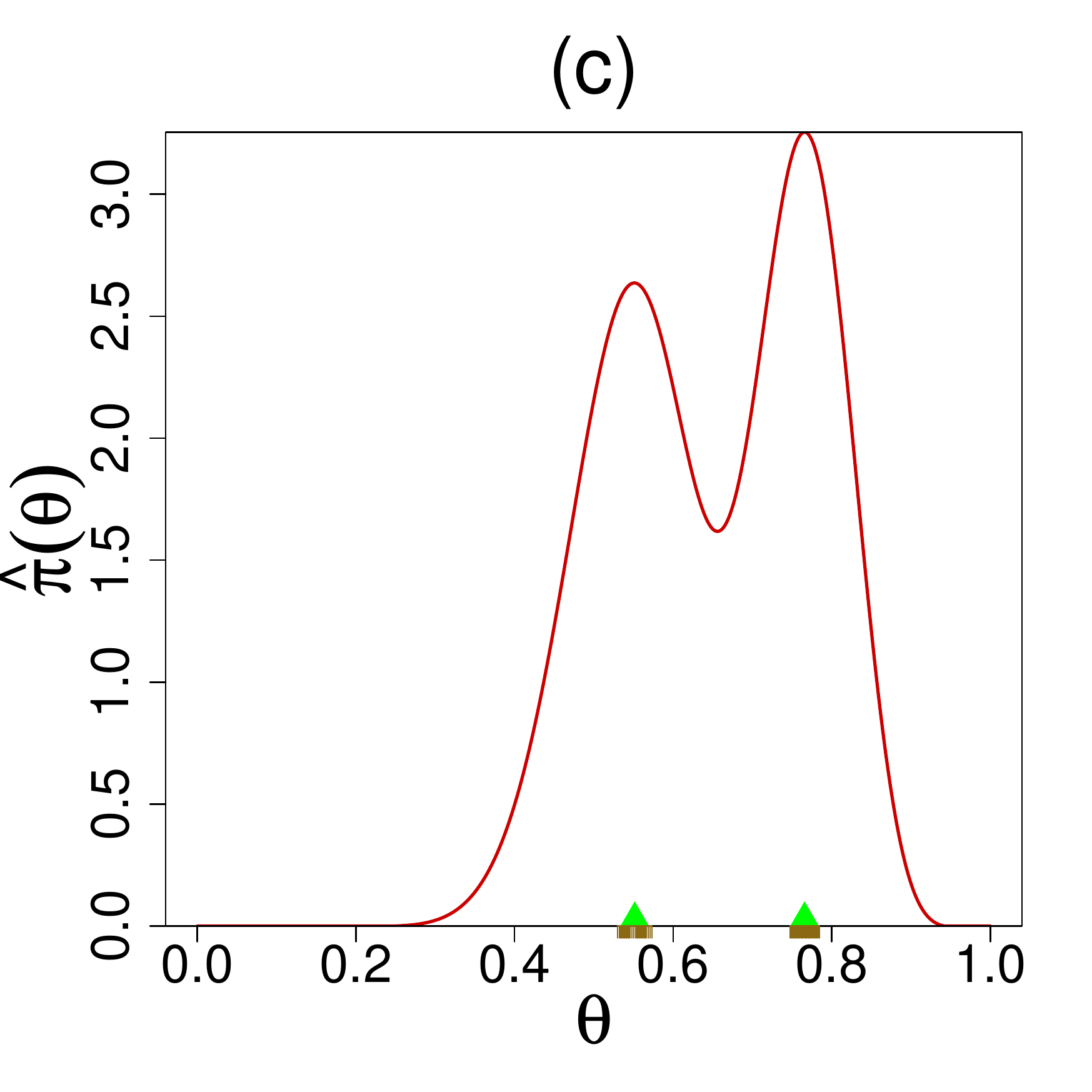}
\end{subfigure}
\caption{Estimated macro-inference summary along with standard errors (using smooth bootstrap) are shown. Panel (a) displays the rat tumor data modes located at $0.034\ (\pm 0.016)$ and $0.156\ (\pm 0.016)$. Panel (b) shows the estimated unimodal prior of the terbinafine data has a mean at $0.034\ (\pm 0.006)$. Panel (c) presents the modes of the rolling tacks data at $0.55\ (\pm 0.022)$ and $ 0.77\ (\pm 0.018)$.}
\label{fig:sec5_mode_int} 
\end{figure}
The resulting two groups of rat tumor experiments are enumerated in the Table \ref{tbl:rat_cluster}. Although we do not have the necessary biomedical background to scientifically justify this new discovery, we are aware that potentially numerous factors (e.g., experimental design, underlying conditions, selection of specific groups of female rats) may contribute to creating this systemic variation.
\begin{table}[th]
\centering
\caption{\label{tbl:rat_cluster} Two group partitions of the rat tumor studies based on K-means clustering on the posterior mode predictions (see Section \ref{sec:micro} and Fig. \ref{fig:rat_pospred3}(c)).}
\vskip.65em
\begin{tabular}{ll}
  \toprule
 Group & ~~~~~~~~~~~~~~~~~~~~~~~~~~~~~~~~~~~~~~~~~~~~~~Studies \\
  \midrule
 \multirow{ 2}{*}{~~~1}  & (0,20), (0,20), (0,20), (0,20), (0,20), (0,20), (0,20), (0,19), (0,19), (0,19), (0,19) \\[.25em] & (0,18), (0,18), (0,17), (1,20), (1,20), (1,20), (1,20), (1,19), (1,19), (1,18), (1,18)\\ [2em]
\multirow{ 5}{*}{~~~2} & (3,27), (2,25), (2,24), (2,23), (2,20), (2,20), (2,20), (2,20), (2,20), (2,20), (1,10)\\[.25em] & (5,49), (2,19), (5,46), (2,17), (7,49), (7,47), (3,20), (3,20), (2,13), (9,48), (10,50)\\[.25em] & (4,20), (4,20), (4,20), (4,20), (4,20), (4,20), (4,20), (10,48), (4,19), (4,19), (4,19)\\[.25em] & (5,22), (11,46), (12,49), (5,20), (5,20), (6,23), (5,19), (6,22), (6,20), (6,20), (6,20)\\[.25em] & (16,52), (15,46), (15,47), (9,24) \\
\bottomrule
\end{tabular}
\end{table}
\vskip.5em
~~~~~2. \textit{From single mean to multiple modes}. An attempt to combine the two subpopulations using a single prior mean (as carried out for the terbinafine example) would result in overestimating one group and underestimating another. We prefer \textit{modes} of $\widehat{\pi}(\te)$, along with their SEs, as a good representative summary, which can be easily computed by the nonparametric smooth bootstrap via $\DS(G,m)$ sampler.

\vskip1em
Learning from big heterogeneous studies is one of the most important yet unsettled matters of modern macroinference \cite{cox1988,efron1996empirical}. Our key insight is the realization that the `science of combining' critically depends on the \textit{shape} of the estimated prior. One interesting and commonly encountered case is multimodal structure of the learned prior. In such situations, instead of the prior-mean summary, we recommend group-specific modes. Our algorithm is also capable of finding data-driven clusters of the partially exchangeable studies in a fully automated manner.

\begin{figure}[t]
\vskip1em
\centering
\begin{subfigure}{.44\textwidth}
  \centering
  \includegraphics[width=\linewidth,trim=1cm .5cm 0cm 1cm]{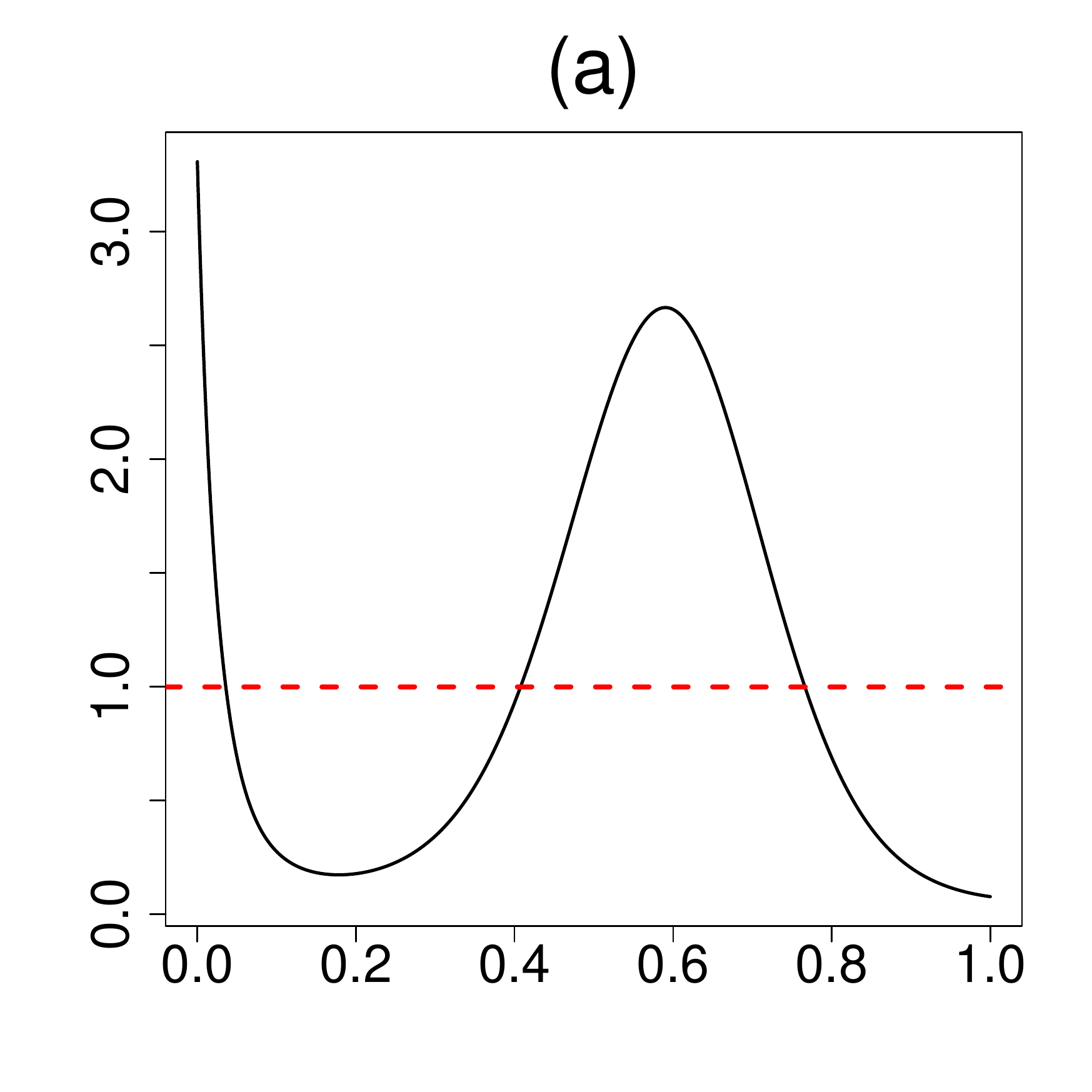}
\end{subfigure}\hspace{8mm}%
\begin{subfigure}{.44\textwidth}
  \centering
\includegraphics[width=\linewidth,trim=.4cm .5cm .5cm 1cm]{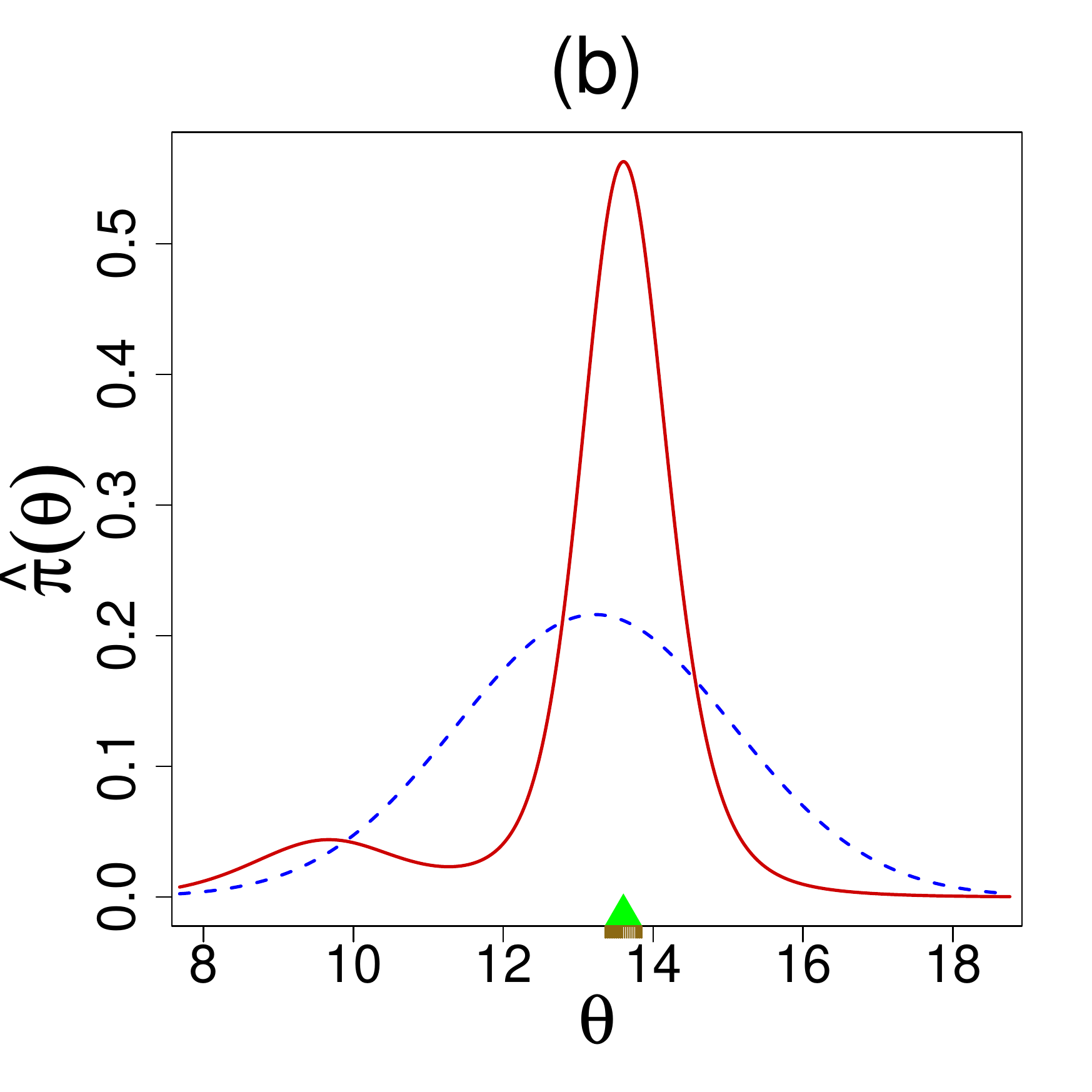}
\end{subfigure}
\caption{Panel (a) shows the U-function, while panel (b) compares the DS-prior $\hat{\pi}(\theta)$ (solid red) with the PEB prior $g(\theta; \alpha,\beta)$ (dashed blue) for the arsenic data. Based on the estimated macro-inference summary along with standard errors (using smooth bootstrap), the best consensus value is the mode $13.6\ (\pm 0.242)$.}
\label{fig:ars_ufunc_ds}
\vspace{-.5em}
\end{figure}

\vskip.5em
\subsection{Learning From Uncertain Data}\label{sec:InfLearningArs}
An important problem of measurement science that routinely appears in metrology, chemistry, physics, biology, and engineering can be stated as follows: measurements are made by $k$ different laboratories in the form of $y_1,\ldots,y_k$ along with their estimated standard errors $s_1,\ldots,s_k$. Given this uncertain data, a fundamental problem of interest is inference concerning: (i) estimation of the consensus value of the measurand, and (ii)  evaluation of the associated uncertainty. The data in Table \ref{tbl:ars_data} are an example of such an inter-laboratory study involving $k=28$ measurements for the level of arsenic in oyster tissue. The study was part of the National Oceanic and Atmospheric Administration’s National Status and Trends Program Ninth Round Intercomparison Exercise \cite{willie1995ninth}. 
\begin{table}[H]
\setlength{\tabcolsep}{5pt}
\centering
\caption{\label{tbl:ars_data} Measurements (sorted) along with their uncertainty from different laboratories in arsenic data.} 
\begin{tabular}{lcccccccccc}
  \toprule
Laboratory & 1 & 2 & 3 & 4 & 5 & $\cdots$ & 25 & 26 & 27 & 28 \\
  \midrule
Measurement ($y_i$) & 9.78 & 10.18 & 10.35 & 11.60 & 12.01 & $\cdots$ & 14.70 & 15.00 & 15.10 & 15.50 \\
Uncertainty ($s_i$) & 0.30 & 0.46 & 0.07 & 0.78 & 2.62 & $\cdots$ & 0.30 & 1.00 & 0.20 & 1.60 \\ 
   \bottomrule
\end{tabular}
\end{table}
{\bf Arsenic data analysis}. We start with the DS-measurement model: $Y_i|\Te_i=\te_i \sim \cN(\te_i,s_i^2)$ and $\Theta_i \sim \DS(G,m)$ $(i=1,\ldots, 28)$ with $G$ being $\cN(\mu,\tau^2)$. The shape of the estimated U-function in Fig. \ref{fig:ars_ufunc_ds}(a) indicates that the pre-selected  prior $\cN(\hat \mu=13.22, \hat \tau^2=1.85^2)$ is clearly unacceptable for arsenic data, thereby disqualifying the classical Gaussian random effects model \cite{rukhin1998}. The DS-corrected $\widehat \pi$ shows some interesting asymmetric pattern with two-bumps. The left-mode represents measurements from three laboratories that are unlike the majority. The result of our macro-inference is shown in Fig. \ref{fig:ars_ufunc_ds}(b), which delivers the  consensus value $13.6 \pm 0.24$. This is clearly far more resistant to fairly extreme low measurements and surprisingly, also more accurate when compared to the parametric EB estimate $13.22 \pm 0.26$. Most importantly, our scheme provides an automated solution to the fundamental problem of \textit{which (as well as how)} measurements from the participating laboratories should be combined to form a best consensus value. Possolo \cite{possolo2013five} fits a Bayesian hierarchical model with prior as Student’s $t_\nu$, where the degrees of freedom was also treated as a random variable over some arbitrary range $\{3,\ldots, 118\}$. Although a heavy-tailed Student's t-distribution is a good choice to `robustify' the analysis, it fails to capture the inherent asymmetry and the finer modal structure on the left. Distinguishing long-tail from bimodality is an important problem of applied statistics by itself.  
\vskip.2em
To summarize, there are several attractive features of our general approach: (i) it adapts to the structure of the data, yet (ii) allows the use of expert opinion to go from knowledge-based prior to statistical prior; (iii) if multiple expert opinions are available, one can also use the U-diagnostic for reconciliation--exploratory uncertainty assessment; (iv) it avoids the questionable exercise of detecting and discarding apparently unusual measurements \cite{possolo2009lab}, and finally (v)  our theory is still applicable for very small number of parallel cases (cf. Fig. \ref{fig:DC_prior_all}(c)), a situation which is not uncommon in inter-laboratory studies.
\subsection{MicroInference}\label{sec:micro}
The objective of microinference is to estimate a specific microlevel $\te_i$ given $y_i$. Consider the rat tumor example where, along with earlier $k=70$ studies, we have an additional current experimental data, that shows $y_{71}=4$ out of $n_{71}=14$ rats developed tumors. How can we estimate the probability of a tumor for this new clinical study? There could be at least three ways to answer this question:
\begin{itemize}[itemsep=1.24pt,topsep=1.24pt]
\item Frequentist MLE estimate: An obvious estimate would be the sample proportion $\wtte_i:$ $y_{71}/n_{71}=0.286$. This operates in an isolated manner, completely ignoring the additional historical information of $k=70$ studies.
\item Parametric empirical Bayes estimate: It is reasonable to expect that the historical data from earlier studies may be related to the current $71$st study, thus borrowing information can result in improved estimator of $\te_{71}$. Bayes posterior mean estimate $\widecheck{\te}_i=\Ex_G[\Te_i|y_i]$ operationalizes this heuristic, which in the Binomial case takes the following form:
\beq \label{eq:ss}
\wcte_i=\dfrac{n_i}{\al+\be+n_i} \, \wtte_i\,+\,\dfrac{\al+\be}{\al+\be+n_i} \Ex_G[\Te].\eeq
This is famously known as Stein's shrinkage formula \cite{stein1956,efron1975}, as it pulls the sample proportions toward the \textit{overall} mean of the prior $\frac{\al}{\al+\be}$. For smaller ($n_i$) studies, shrinkage intensity is higher, which allows them to learn from other experiments. 
\item Nonparametric Elastic-Bayes estimate: Is it a wise strategy to shrink all $\wtte_i$'s toward the grand mean $0.14$? Interestingly, this shrinking point is near the valley between the twin-peaks of the rat tumor prior density estimate (verify from Fig. \ref{fig:sec5_mode_int}(a)) and therefore may not represent a preferred location. Then, \textit{where to shrink?} Ideally, we want to learn only from the \textit{relevant} subset of the full dataset--\textit{selective shrinkage}, e.g., for rat data, it would be the group 2 of Table \ref{tbl:rat_cluster}. This brings us to the question: how to rectify the parametric empirical Bayes estimate $\wcte_i $? The formula \eqref{eq:pm} gives us the required (nonlinear) adjusting factor:
\beq \label{eq:LPstein}
\hte_i \,= \,\frac{ \wcte_i +\sum_j \widehat{\LP}[j;G,\Pi]\,\Ex_G[\Theta_i T_j(\Theta_i;G)| y_i]}{1+\sum_j\widehat{\LP}[j;G,\Pi]\, \Ex_G[T_j(\Te_i;G)|y_i]},
\eeq
dictating the magnitude and direction of shrinkage in a completely data-driven manner via LP-Fourier coefficients. Note that when $d \equiv 1$, i.e., all the $\LP[j;G,\Pi]$ are zero, \eqref{eq:LPstein} reproduces the parametric $\wcte_i$. Due to its flexibility and adaptability, we call this the Elastic-Bayes estimate. This can be considered as a nonparametric class of shrinkage estimators that starts with the classical Stein's formula and rectifies it by looking at the data. 
\end{itemize}
\begin{figure}[t]
\begin{subfigure}{.33\textwidth}
\vskip.65em
  \centering
  \includegraphics[width=\linewidth,trim=1cm 0cm 0cm 1cm]{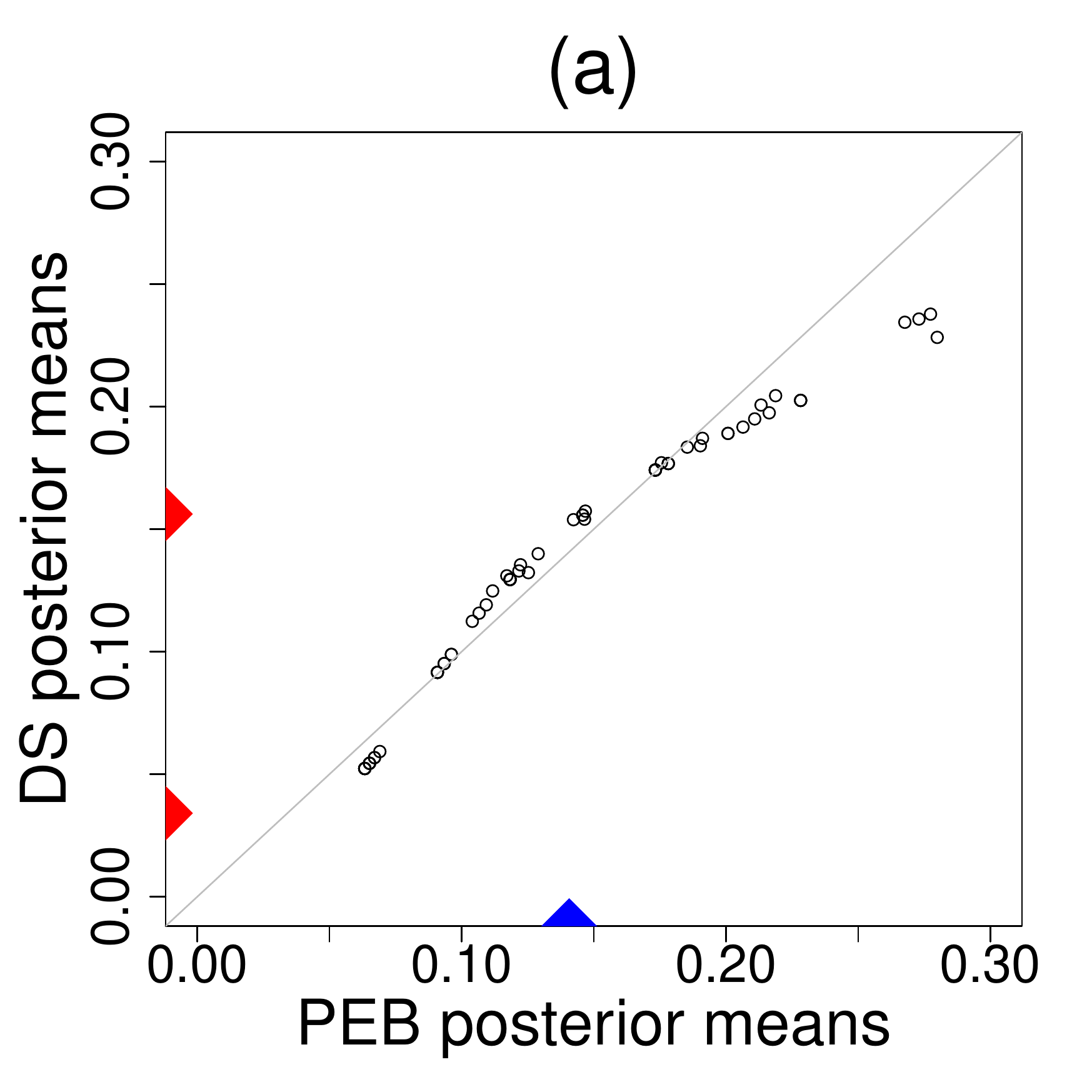}
\end{subfigure}\hspace{2mm}%
\begin{subfigure}{.33\textwidth}
  \centering
\includegraphics[width=\linewidth,trim=.5cm 0cm .5cm 1cm]{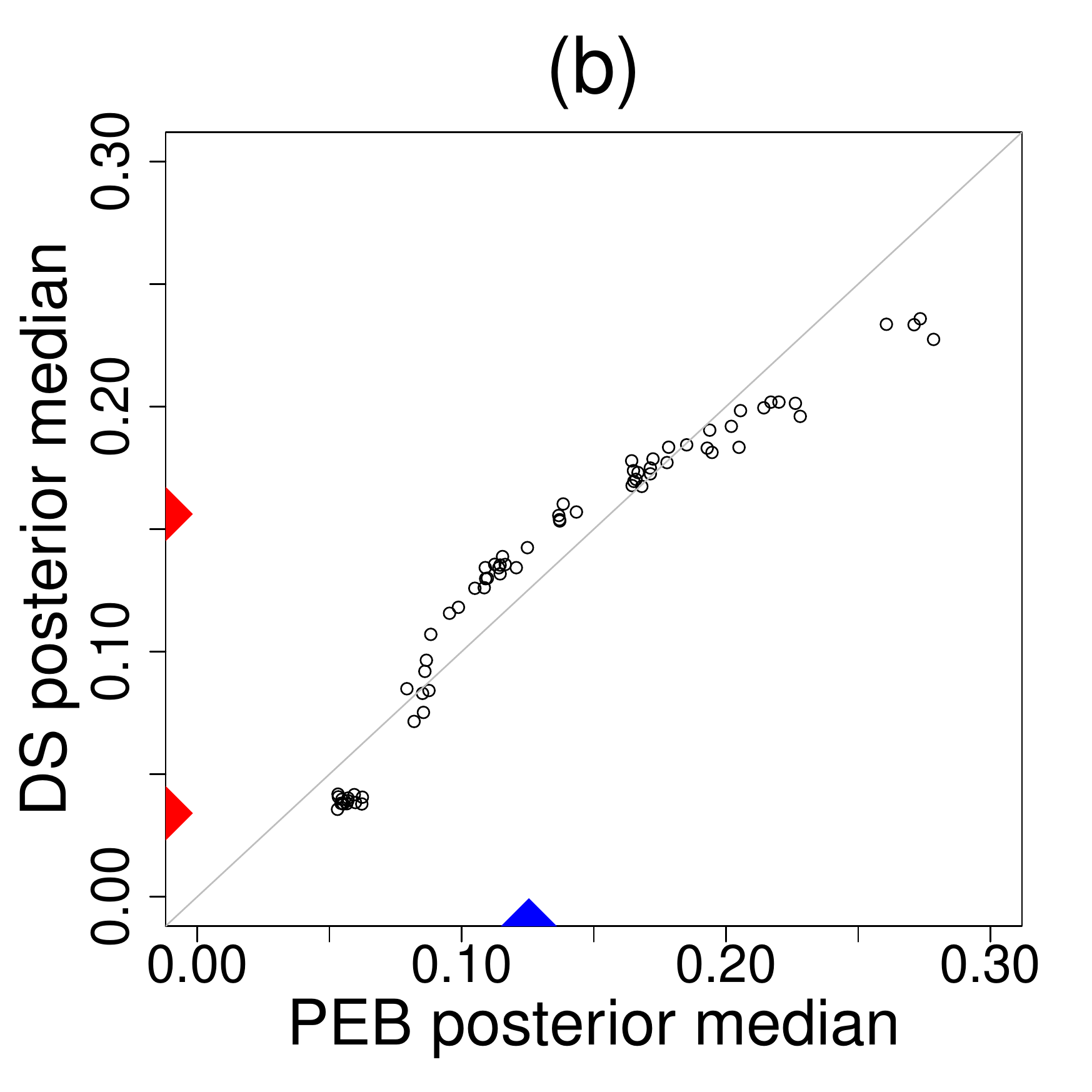}
\end{subfigure}\hspace{2mm}%
\begin{subfigure}{.33\textwidth}
  \centering
  \includegraphics[width=\linewidth,trim=0cm 0cm 1cm 1cm]{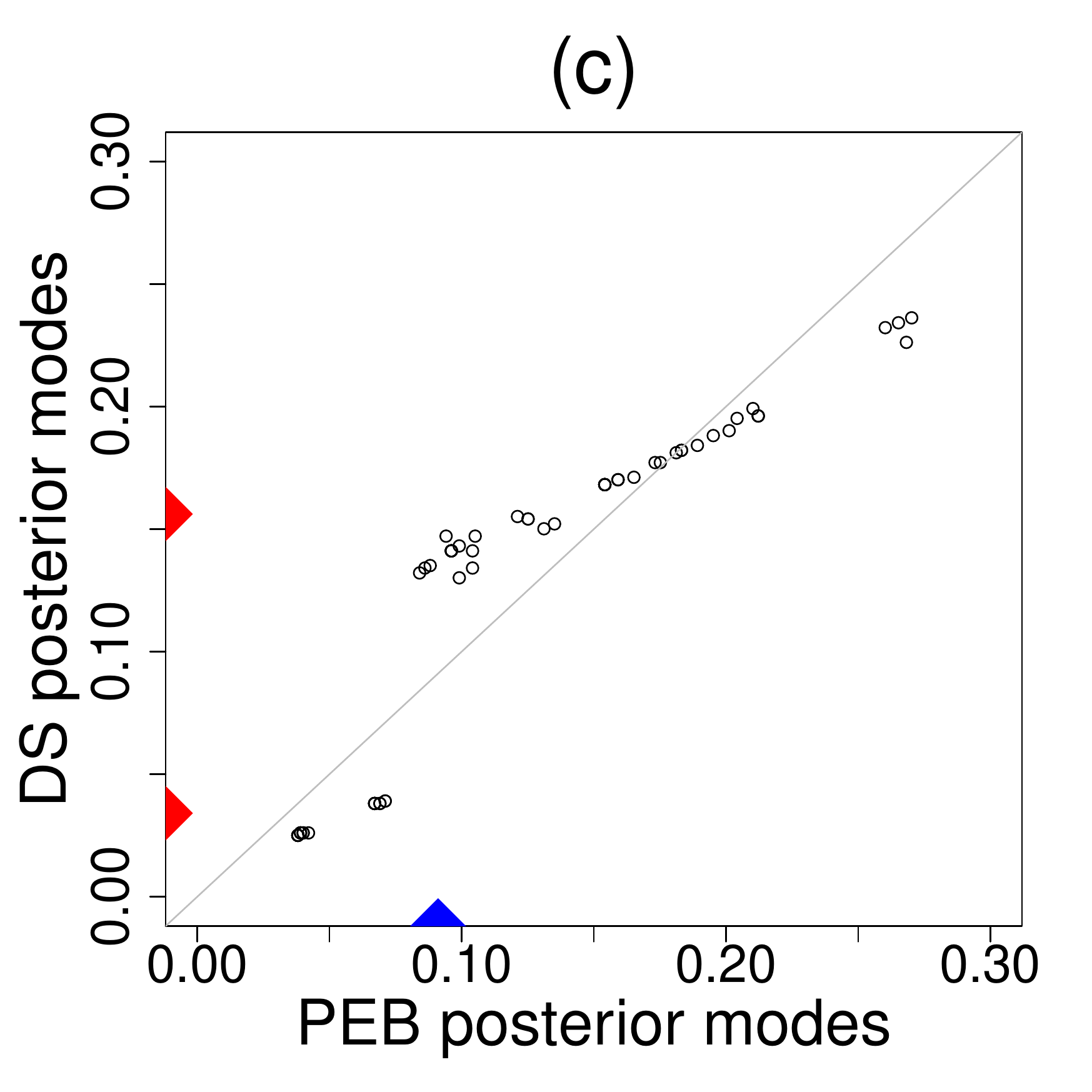}
\end{subfigure}
\caption{Comparisons of DS Elastic-Bayes and PEB posterior predictions of the rat tumor data: (a) posterior means, (b) posterior medians, and (c) posterior modes.  The vertical red triangles indicate the location of the modes on the DS prior; the blue triangles respectively denote the mean, median, and mode of the parametric ${\rm Beta}(\hat \al=2.3,\hat \be=14.08)$.}
  \label{fig:rat_pospred3}
\end{figure}
{\bf Rat tumor example}. Figure \ref{fig:rat_pospred3} compares Stein's empirical Bayes estimate with our Elastic-Bayes estimate for the all $k=70$ tumor rates. Posterior mean, median, and mode of $\te_j$'s are shown side by side in three plots. The departure from the $45^{\circ}$ reference line is a consequence of ``adaptive shrinkage.'' Elastic-Bayes automatically shrinks the empirical $\wtte_i$ towards the representative modes ($0.034$ and $0.156$), whereas the Stein's PEB estimate uses the grand mean ($\approx 0.14$) as the shrinking target for \textit{all} the tumor rates. This is particularly prominent in Fig. \ref{fig:rat_pospred3} (c) for maximum a posteriori (MAP) estimates. As before, for heterogeneous population, we prescribe posterior mode as the final prediction. 
\vskip.25em
{\bf The Pharma-example}.  Our DS Elastic-Bayes estimate is especially powerful in the presence of prior-data conflict. To illustrate this point, we report a small simulation study.  The goal is to compare MSE for frequentist MLE, parametric empirical Bayes, and nonparametric Elastic-Bayes estimates for a new study $y_{\text{new}}$ in various levels of prior-data conflict.  To capture the prior-data conflict, we consider the following model for $\pi(\theta)$ and $y_{\text{new}}$:
\begin{align*}
\pi(\theta) &= \eta \text{Beta}(5,45) + (1 - \eta) \text{Beta}(30,70)\\
y_{\text{new}} &\sim \text{Bin}(50, 0.3).
\end{align*}
The parameter $\eta$ varies from 0 to 0.50 in increments of 0.05; as $\eta$ increases we introduce more heterogeneity into the true prior distribution and exacerbate the prior-data conflict between $\pi(\theta)$ and $y_{\text{new}}$; see Fig. \ref{fig:SEC4_MSE_Ratio}(a).  We simulated $k=100$ $\theta_i$ from $\pi(\theta)$, with which we generate $y_i | \theta_i \sim \text{Bin}(60,\theta_i)$.  Using the Type-II MoM algorithm on the simulated data set, we found $\hat{\pi}$. After generating $y_{\text{new}}$, we then determined the frequentist MLE, parametric EB (PEB), and the nonparametric elastic Bayes estimates of the mode.  For each value of $\eta$, we repeated this process $250$ times and found the mean squared error (MSE) for each estimate.  To better illustrate the impact of prior-data conflicts, we used ratio of PEB MSE to frequentist MSE and PEB MSE to DS MSE.  The results are shown in Fig. \ref{fig:SEC4_MSE_Ratio} (b).

\begin{figure}[t]
\centering
\begin{subfigure}{.41\textwidth}
  \centering
  \includegraphics[width=\linewidth,trim=1cm .5cm 0cm 1cm]{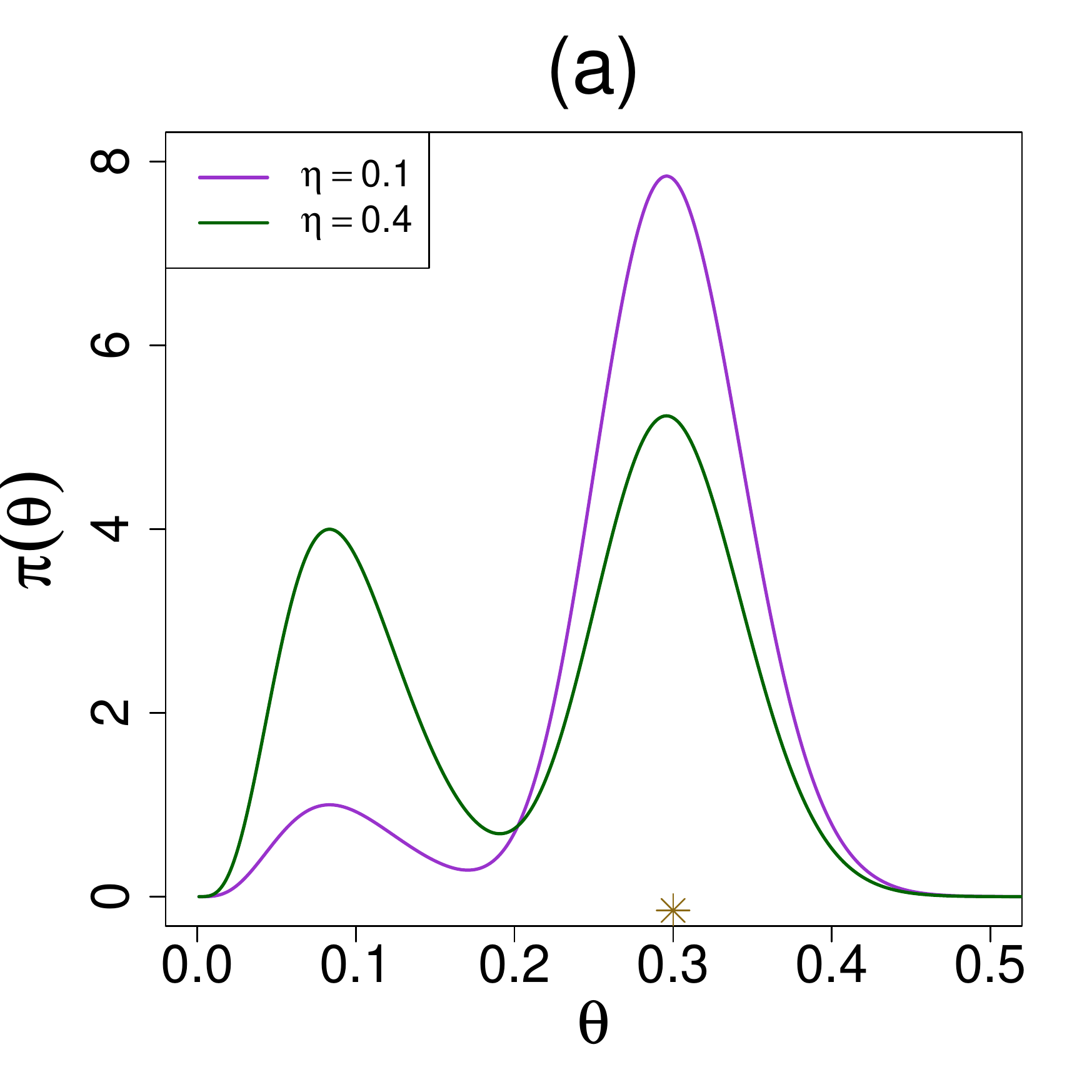}
\end{subfigure}\hspace{10mm}%
\begin{subfigure}{.41\textwidth}
  \centering
\includegraphics[width=\linewidth,trim=.5cm .5cm .5cm 1cm]{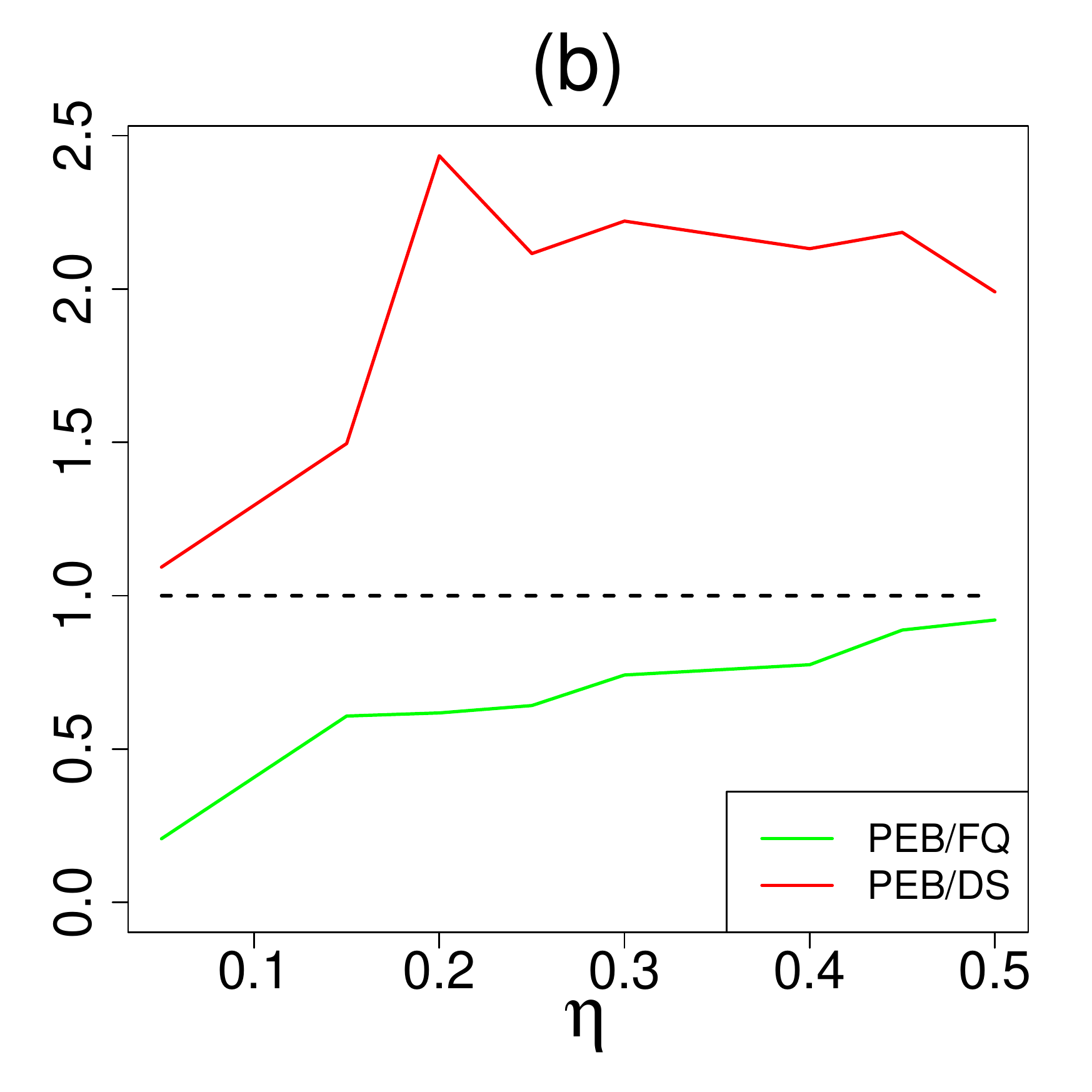}
\end{subfigure}\hspace{1.5mm}%
  \caption{Panel (a) illustrates the prior-data conflict for $\eta = 0.1$ versus $\eta = 0.4$; `*' denotes $0.3$, the true mean of $y_{\text{new}}$. Panel (b) shows the MSE ratios for PEB to Frequentist MLE (PEB/FQ; green) and PEB to DS (PEB/DS; red) with respect to $\eta$.  Notice that as more prior-data conflict is introduced, DS outperforms PEB while frequentist MLE performance improves.}
\label{fig:SEC4_MSE_Ratio}
\end{figure}
The Elastic-Bayes estimate outperforms the Stein's estimate for all $\eta$. More importantly the efficiency of our estimate continues to increase with the heterogeneity. This is happening because elastic Bayes performs \textit{selective} shrinkage of sample proportion towards the appropriate mode (near $0.3$) and thus gains ``strength'' by combining information from `similar' studies even when the contamination in the study population increases.  An interesting observation is the performance of the frequentist MLE estimate; as the data becomes more heterogeneous, the frequentist MLE shows improvement with respect to the  Stein's PEB estimate.  Our simulation depicts a scenario that is very common in historic-controlled clinical trials, where the heterogeneity arises due to changing conditions. Additional comparisons with other empirical Bayes procedures can be found in Supplementary Appendix G.

\begin{figure}[t]
\begin{subfigure}{.32\textwidth}
  \centering
  \includegraphics[width=\linewidth,trim=1cm 0cm 0cm 1cm]{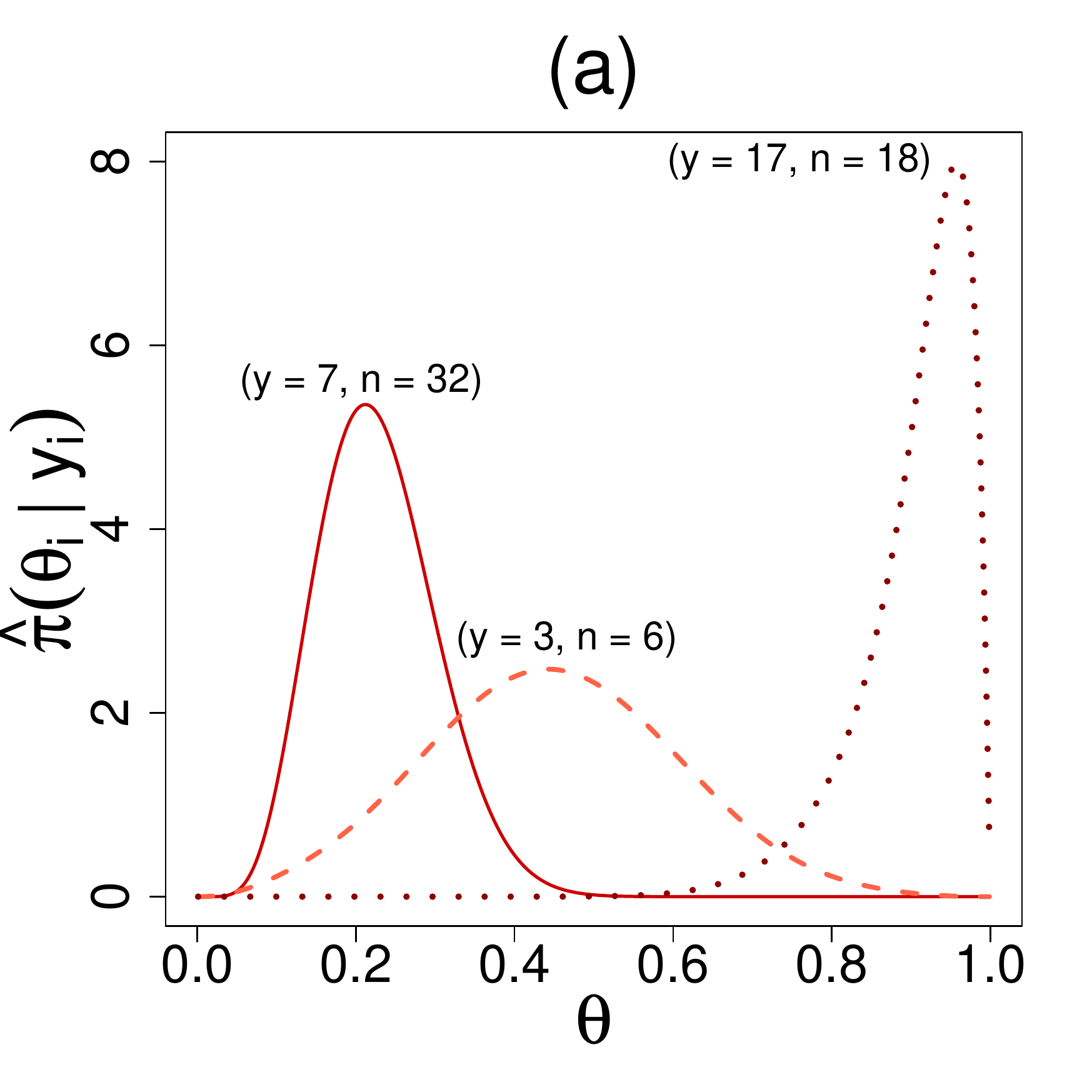}
\end{subfigure}  \hspace{.25mm}
\begin{subfigure}{.32\textwidth}
  \centering
\includegraphics[width=\linewidth,trim=.5cm 0cm .5cm 1cm]{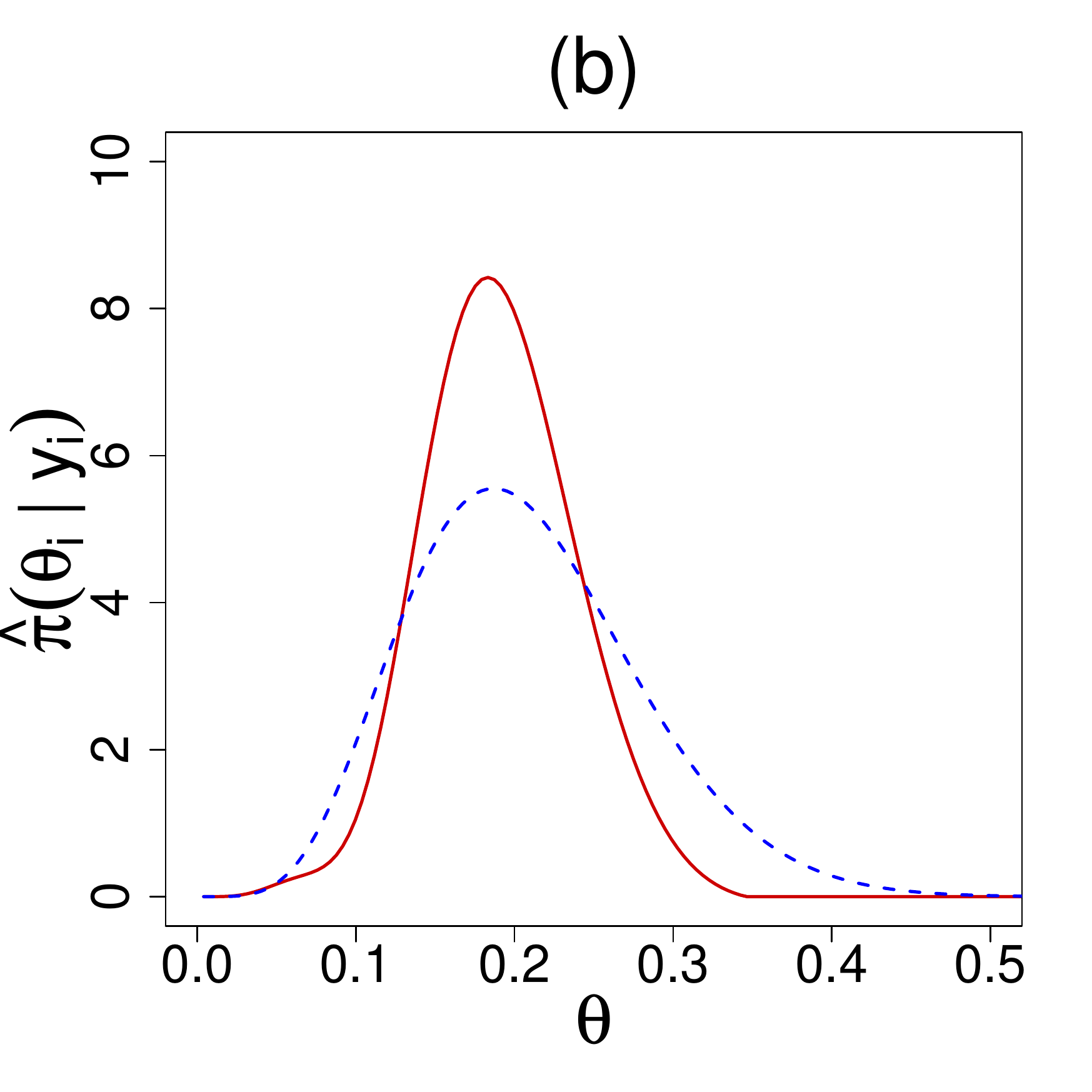}
\end{subfigure}  \hspace{.5mm}
\begin{subfigure}{.32\textwidth}
  \centering
  \includegraphics[width=\linewidth,trim=.25cm 0cm .75cm 1cm]{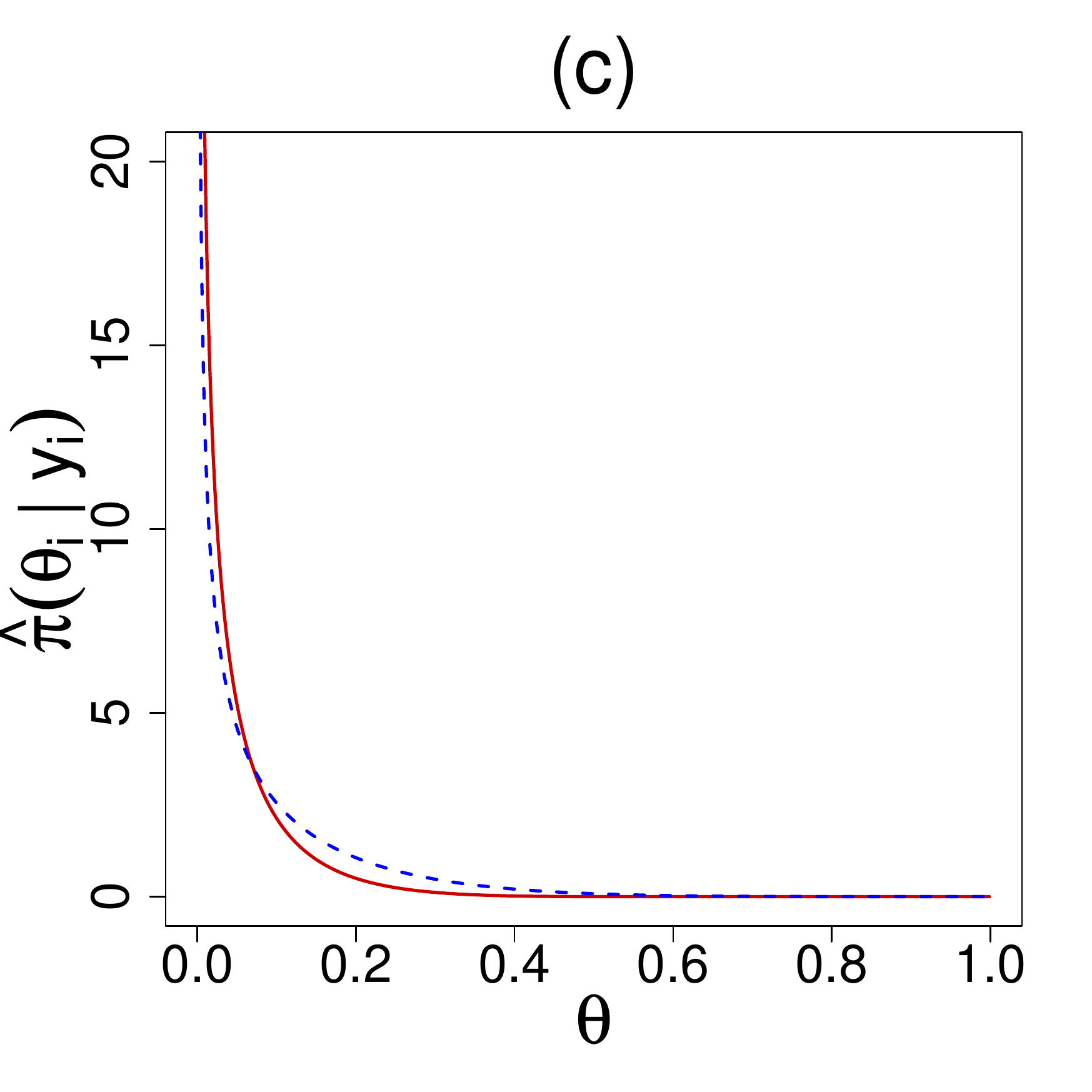}
\end{subfigure}
\begin{subfigure}{.32\textwidth}
  \centering
  \includegraphics[width=\linewidth,trim=1cm 1cm 0cm 0cm]{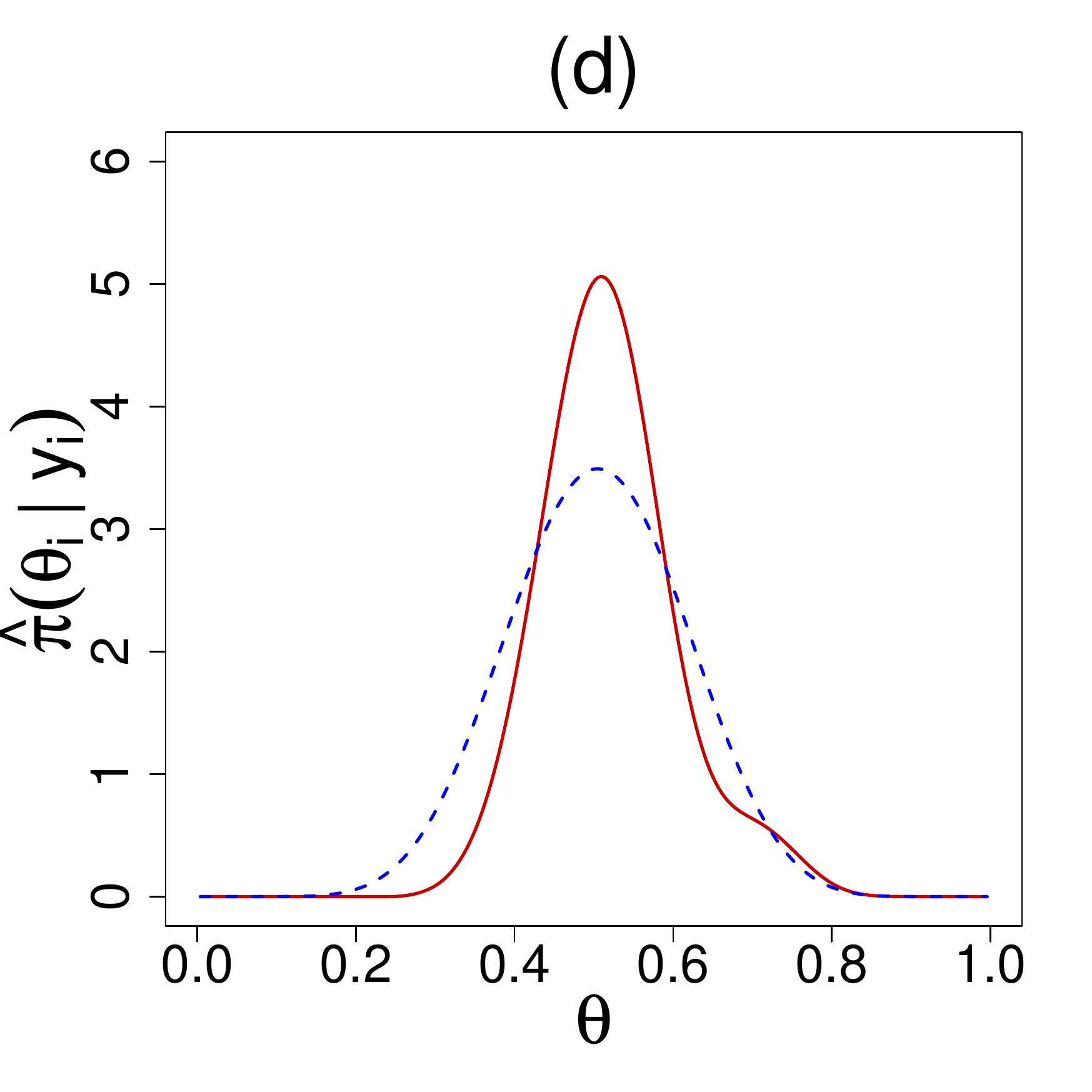}
\end{subfigure} \hspace{.5mm}
\begin{subfigure}{.32\textwidth}
  \centering
  \includegraphics[width=\linewidth,trim=.5cm 1cm .5cm 0cm]{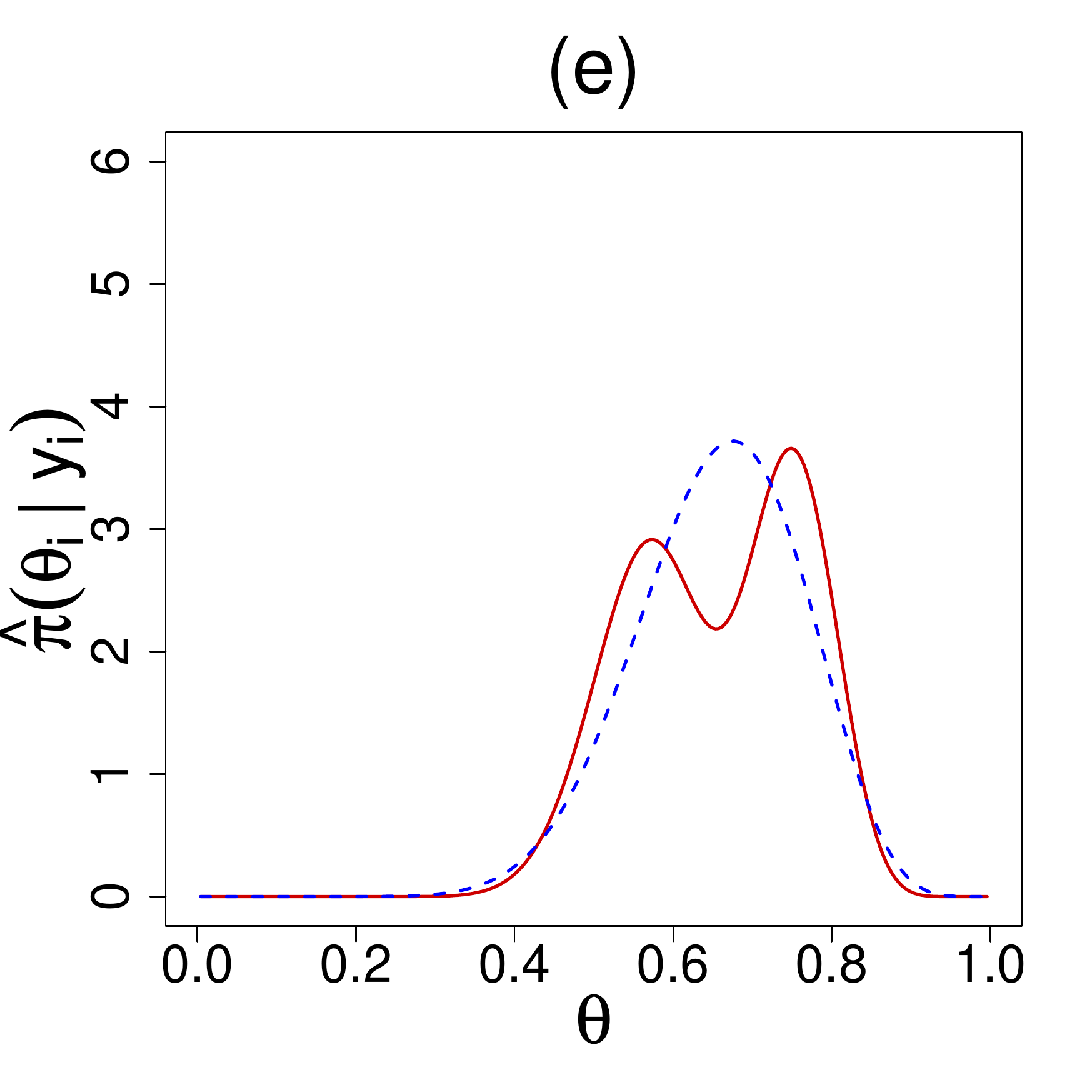}
\end{subfigure} \hspace{.9mm}
\begin{subfigure}{.32\textwidth}
  \centering
  \includegraphics[width=\linewidth,trim=0cm 1cm 1cm 0cm]{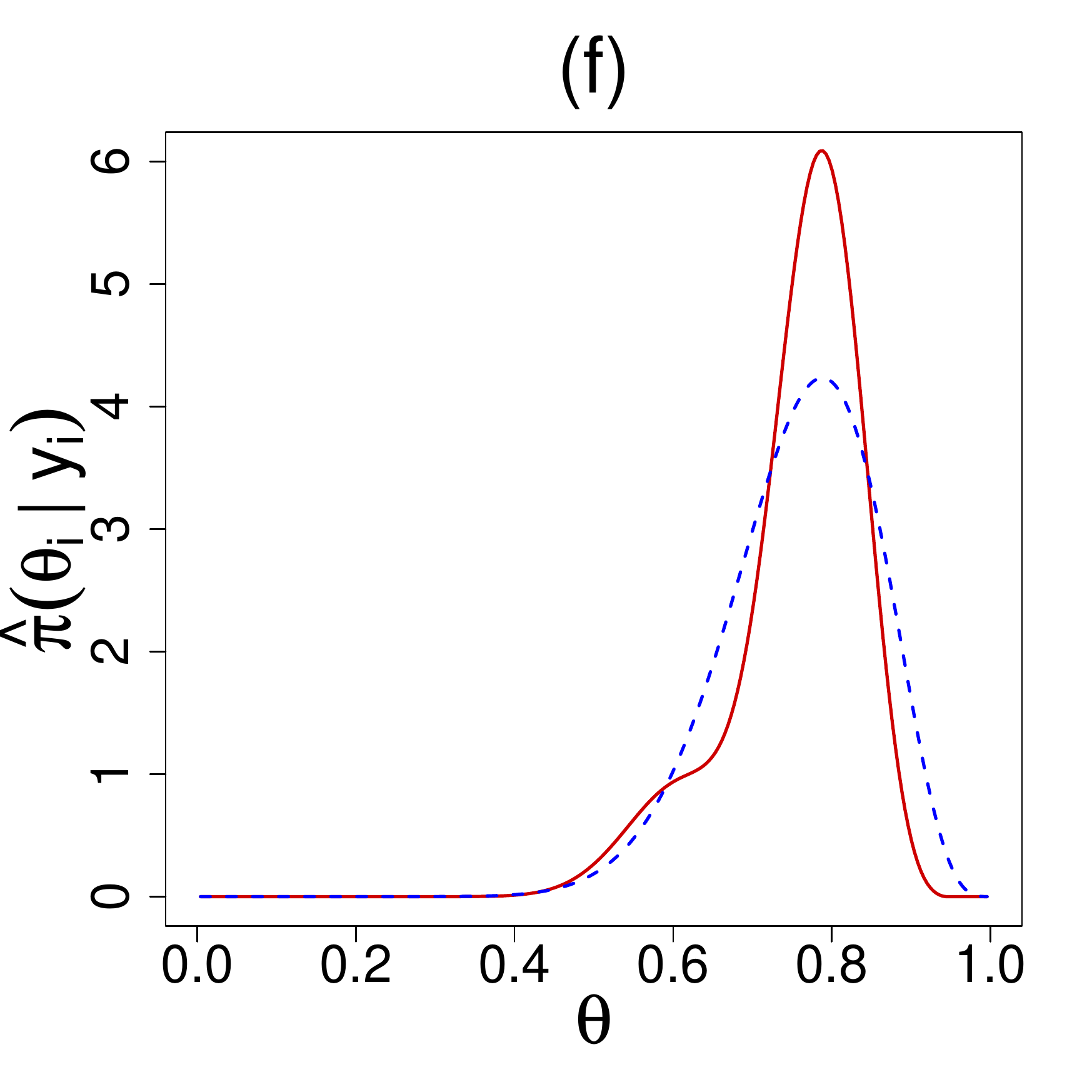}
\end{subfigure}

\caption{Panel (a) shows DS posterior plots of three observations from the surgical node data: $(y=7,n=32)$, $(y=3,n=$ $6)$, and $(y=17,n=18)$. For panels (b) through (f), red denotes the DS posterior and blue dashed is the PEB posterior. Panel (b) is $\hat{\pi}(\theta_{71}|y_{71}=4)$ for the rat tumor data.  Panel (c) displays $\hat{\pi}(\theta_{6}|y_{6}=0)$ for the Navy shipyard data.  The second row shows the posterior distributions of (d) $y_i = 3$, (e) $y_i = 6$, and (f) $y_i = 8$ from the rolling tacks data.}
\label{fig:SEC6_Post_Plot}
\end{figure} 

\vskip.25em
{\bf Three additional real examples.} Figure \ref{fig:SEC6_Post_Plot} shows the posterior plots for specific studies in four of our data sets: surgical node, rat tumor, Navy shipyard, and rolling tacks.  In studies like the surgical node data, personalized predictions are typically valuable.  Figure \ref{fig:SEC6_Post_Plot}(a) shows posterior distributions for three selected patients, which are indistinguishable from Efron's deconvolution answer \cite[Fig. 4]{cox2017statistical}; the patient with $n_i=32$ and $y_i=7$ shows almost certainly $\te_i > 0.5$, i.e., he or she is highly prone to positive lymph nodes, and thus should be referred to follow-up therapy. With regard to the rat tumor data, Fig. \ref{fig:SEC6_Post_Plot}(b) depicts the DS-posterior distribution of $\theta_{71}$ along with its parametric counterpart $\pi_G(\te_{71}|y_{71},n_{71})$. Interestingly, the DS nonparametric posterior shows less variability;  this possibly has to do with the selective learning ability of our method, which learns from similar studies (e.g. group 2), rather than the whole heterogeneous mix of studies.  We see similar phenomena in the rolling tacks data, where panel (d): $y_i = 3$, is more reflective of the first mode and panel (f): $y_i = 8$, of the second.  Panel (e) shows the bimodal posterior for $y_i = 6$ case.  Finally, the Navy shipyard data (Fig. \ref{fig:SEC6_Post_Plot} (c)) exhibits another advantage of DS priors: it works equally well for small $k$. The DS-posterior mean estimate for $y_6=0$ is $0.0471$, which is consistent with the findings of Sivaganesan and Berger \cite[p. 117]{sivaganesan1993robust}.

\subsection{Poisson Smoothing: The Two Cultures}\label{sec:PoiSmo}
We consider the problem of estimating a vector of Poisson intensity parameters $\te=(\te_1,\ldots,\te_k)$ from a sample of $Y_i|\te_i \sim \mbox{Poisson}(\te_i)$, where the Bayes estimate is given by: 
\vskip.1em
\beq \label{eq:PM} \Ex[\Theta|Y=y]\,=\,\dfrac{\int_0^\infty \te \big[ e^{-\te} \te^y/y! \big] \pi(\te) \dd \te}{\int_0^\infty \big[ e^{-\te} \te^y/y! \big] \pi(\te) \dd \te};~~~y=0,1,2,\ldots.\eeq
Two primary approaches for estimating \eqref{eq:PM}:
\begin{itemize}[itemsep=1.24pt,topsep=1.24pt]
\item Parametric Culture \cite{fisher1943,maritz1969}: If one assumes $\pi(\te)$ to be the parametric conjugate Gamma distribution $g(\te;\al,\be)=\frac{1}{\be^\al \Gamma (\al)}\te^{\al-1}$ $e^{-\te/\be}$, then it is straightforward to show that Stein's estimate takes the following analytical form $\wcte_i= \frac{y_i+\al}{\be^{-1}+1}$, weighted average of the MLE $y_i$ and the prior mean $\al \beta$.
\item Nonparametric Culture \cite{robbins1956,efron2003robbins,koenker2016}:
This was born out of Herbert Robbins' ingenious observation that \eqref{eq:PM} can alternatively be written in terms of marginal distribution $(y+1) \frac{f(y+1)}{f(y)}$, and thus can be estimated non-parametrically by substituting empirical frequencies. This remarkable ``prior-free'' representation, however, does not hold in general for other distributions. As a result, there is a need to develop methods that can bite the bullet and estimate the prior $\pi$ from the data. Two such promising methods are Bayes deconvolution \cite{efron2003robbins} and the Kiefer-Wolfowitz non-parametric MLE (NPMLE) \cite{kiefer1956,koenker2016}. Efron's technique can be viewed as \textit{smooth} nonparametric approach, whereas NPMLE generates a discrete (atomic) probability measure. For more discussion, see Supplementary Appendix A2.
\end{itemize}
\vskip.2em
{\bf The Third Culture}. Each EB modeling culture has its own strengths and shortcomings. For example, PEB methods are extremely efficient when the true prior is Gamma.  On the other hand, the NEB methods possess extraordinary robustness in the face of a misspecified prior yet they are inefficient when in fact $\pi \equiv {\rm Gamma}(\al,\be)$. Noticing this trade-off, Robbins raised the following intriguing question \cite{robbins1980}: \textit{how can this efficiency-robustness dilemma be resolved in a logical manner?} To address this issue, we must design a data analysis protocol that offers a mechanism to answer the following \textit{intermediate} modeling questions (before jumping to estimate $\widehat \pi$): Can we assess whether or not a Gamma-prior is adequate in light of the sample-information? In the event of a prior-data conflict, how can we estimate the `missing shape' in a completely data-driven manner? All of these questions are at the heart of our `Bayes \textit{via} goodness-of-fit' formulation, whose goal is to develop a third culture of generalized empirical Bayes (gEB) modeling by uniting the parametric and non-parametric philosophies. Compute the DS Elastic-Bayes estimate by substituting $\wcte_i= \frac{y_i+\al}{\be^{-1}+1}$ in the Eq. \eqref{eq:LPstein}, which reduces to the PEB answer when $d(u;G,\Pi) \equiv 1$ (i.e, the true prior is a Gamma) and modifies non-parametrically, only when needed; thereby turning Robbins' vision into action (see Supplementary Appendices A and G for more discussions on this point).
\vskip.5em
\begin{table}[ht]
\setlength{\tabcolsep}{9pt}
\renewcommand*{\arraystretch}{1.24}
\centering
\caption{\label{tbl:InsMicro} For the insurance data set, estimates for the number of claims expected in the following year by an individual who made $y$ claims during the present year, $\hat{\Ex}(\theta|Y=y)$, by five different methods.}
{\footnotesize
\begin{tabular}{lcccccccc}
  \toprule
Claims $y$  & 0 & 1 & 2 & 3 & 4 & 5 & 6 & 7  \\ 
  \midrule
Counts & 7840 & 1317 & 239 & 42 & 14 & 4 & 4 & 1 \\ [.2em]
Gamma PEB& 0.164 & 0.398 & 0.633 & 0.87 & 1.10 & 1.34 & 1.57 & 1.80\\ [.2em]
Robbins' EB& 0.168 & 0.363 & 0.527 & 1.33 & 1.43 & 6.00 & 1.75 & {\bf ---}\\ [.2em]
Deconvolve & 0.164 & 0.377 & 0.642 & 1.14 & 2.13 & 3.45 & 4.47 & 5.08\\ [.2em]
NPMLE &  0.168 & 0.362 & 0.534 & 1.24 & 2.21 & 2.53 & 2.58 & 2.58\\ [.2em]
DS Elastic-Bayes &  0.156  & 0.322  &  0.517 & 0.744  & 1.02 &  1.56 & 3.01 & 5.24\\ [.2em]
\bottomrule
\end{tabular}
}
\end{table}
\vskip.5em
{\bf The insurance data}. Table \ref{tbl:InsMicro} reports the Bayes estimates $\Ex[\te|Y=y]$ for the insurance data. We compare five methods: parametric Gamma, classical Robbins' EB, Efron's Deconvolve, Koenker's NPMLE, and our procedure. The raw-nonparametric Robbins' estimator is clearly erratic at the tail due to data-sparsity. The PEB estimate overcomes this limitation and produces a stable estimate; but \textit{is it dependable?} Should we stop here and report this as our final result? Our exploratory U-diagnostic tells that (consult Sec \ref{sec:EstResults}) the PEB estimate needs a second-order correction to resolve the discrepancy between the Gamma prior and data. The improved LP-Stein estimates are shown in the last row of Table \ref{tbl:InsMicro}.
\vskip.5em
\begin{figure}[ht]
\centering
\vskip.65em
\begin{subfigure}{.44\textwidth}
  \centering
  \includegraphics[width=\linewidth,trim=1cm .5cm 0cm 1cm]{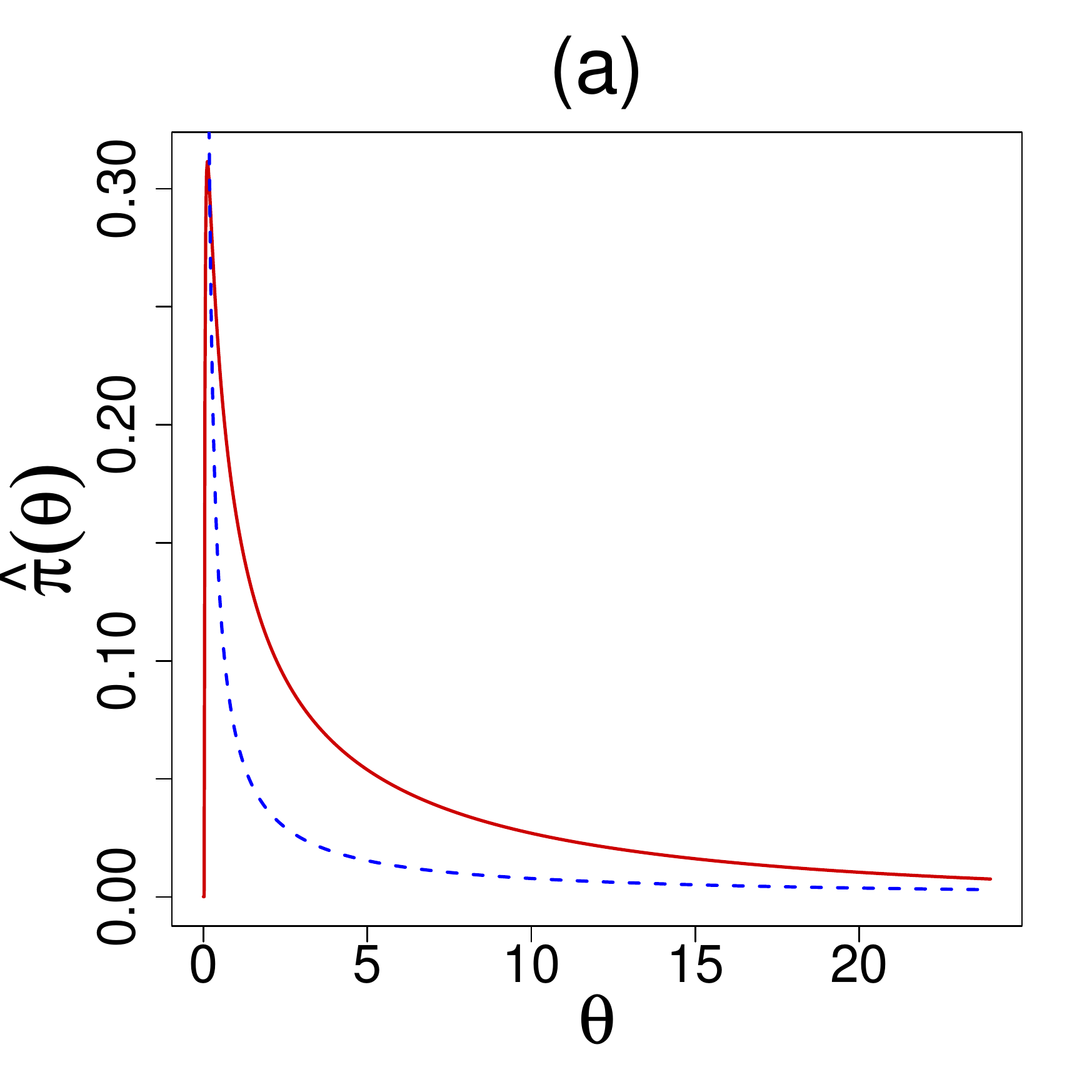}
\end{subfigure}\hspace{6mm}%
\begin{subfigure}{.44\textwidth}
  \centering
\includegraphics[width=\linewidth,trim=.4cm .5cm .5cm 1cm]{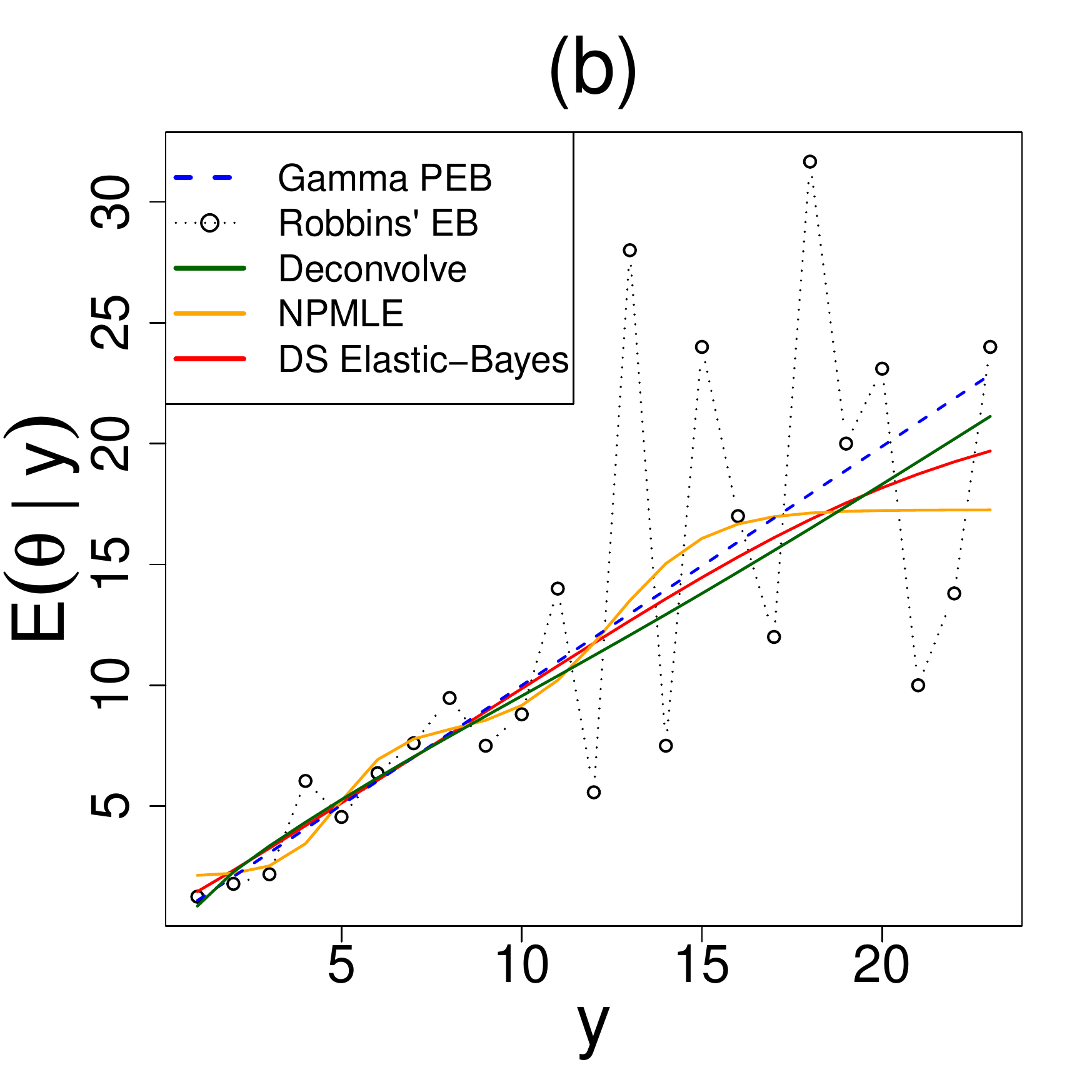}
\end{subfigure}
\caption{Panel (a) displays the estimated $\DS(G,m=4)$ prior (solid red) with the PEB Gamma prior $g(\theta; \alpha,\beta)$ (dashed blue) for the butterfly data; these results indicate that Fisher's Gamma-prior guess required some correction. Panel (b) shows estimates for the number of butterfly species caught in the following year $\hat{\Ex}(\theta \mid x)$ by the Gamma PEB, Robbins' formula, Bayesian deconvolution, NPMLE, and our Elastic-Bayes estimate.}
\label{fig:bfly4_ufunc_ds}
\vspace{-.5em}
\end{figure}

{\bf The butterfly data}. The next example is Corbet's Butterfly data \cite{fisher1943}-- one of the earliest examples of empirical Bayes. Alexander Corbet, a British naturalist, spent two years in Malaysia trapping butterflies in the 1940s. The data consist of the number of species trapped exactly $y$ times in those two years for $y=1,\ldots, 24$. Figure \ref{fig:bfly4_ufunc_ds}(b) plots different Bayes estimates. The Robbins’ procedure suffers from similar `jumpiness.' The blue dotted line represents the linear PEB estimate with $\al=0.104$ and $\be=89.79$ (same as of Efron and Hastie \cite[Eq. 6.24]{efron2016computer}) estimated from the zero-truncated negative binomial marginals. Our DS-estimate is almost sandwiched between the PEB and Deconvolve answer. The NPMLE method (the orange curve) yields some strange looking sinusoidal pattern, probably due to overfitting. In conclusion, we must say that the triumph of our procedure as compared to the other Bayes estimators lies in its automatic adaptability that Robbins alluded in his 1980 article \cite{robbins1980}.

\section{Discussions}\label{sec:Disc}
We laid out a new mechanics of data modeling that effectively consolidates Bayes and frequentist, parametric and nonparametric, subjective and objective, quantile and information-theoretic philosophies. However, at a practical level, the main attractions of our ``Bayes \textit{via} goodness-of-fit'' framework lie in its (i) ability to quantify and protect against prior-data conflict using exploratory graphical diagnostics; (ii) theoretical simplicity that lends itself to analytic closed-form solutions, avoiding computationally intensive techniques such as MCMC or variational methods. 

We have developed the concepts and principles progressively through a range of examples, spanning application areas such as clinical trials, metrology, insurance, medicine, and ecology,  highlighting the core of our approach that gracefully combines Bayesian way of thinking (parameter probability where prior knowledge can be encoded) with a frequentist way of computing via goodness-of-fit (evaluation and synthesis of the prior distribution). If our efforts can help to make Bayesian modeling more attractive and transparent for practicing statisticians (especially non-Bayesians) by even a tiny fraction, we will consider it a success.
\section*{Data availability} All datasets and the computing codes are available via free and open source \texttt{R}-software package \texttt{BayesGOF}. The online link: \mbox{https://CRAN.R-project.org/package=BayesGOF}
\putbib[ref-Doug]
\end{bibunit}
\vskip1em
\section*{Additional information}

{\bf Supplementary information}: It includes (i) connection with other major Bayesian modeling cultures, (ii) Details of \texttt{BayesGOF} R-Software together with additional numerical illustrations, (iii) important extensions to examples with covariates and (iv) Maximum-entropy $\DS(G,m)$ modeling.
\vskip1em
{\bf \noindent Competing Interests}: The authors declare no competing interests.

\newpage
\renewcommand{\baselinestretch}{1.24}
\setlength{\parskip}{1.4ex}
\begin{center}
{\Large {\bf Supplementary Material for ``Bayesian Modeling via Goodness-of-fit''}}\\[.15in] %
Subhadeep Mukhopadhyay$^*$, Douglas Fletcher\\  
Temple University, Department of Statistical Science \\ Philadelphia, Pennsylvania, 19122, U.S.A. \\[1em]
$^*$ To whom correspondence should be addressed; E-mail: deep@temple.edu\\[2.5em]
\end{center}

\setcounter{figure}{8}   
\setcounter{table}{4}
This supplementary document contains nine Appendices, organized as follows:
\begin{itemize}[noitemsep,topsep=1.24pt]
\item Appendix A: Connections with other Bayesian modeling cultures.
\item Appendix B: More insights into the LP-basis functions.
\item Appendix C: The DS$(G,m)$ sampler.
\item Appendix D: Other practical considerations.
\item Appendix E: Software.
\item Appendix F: Data Catalogue.
\item Appendix G: The Robbins' puzzle.
\item Appendix H: Example with covariates.
\item Appendix I: Maximum-Entropy enhancement.
\end{itemize}

\begin{bibunit}[naturemag-doi]
\vskip1em
\begin{center}
{\large A. CONNECTIONS WITH OTHER BAYESIAN MODELING CULTURES}
\end{center}

In this section, we explore the relationship of our approach with other existing Bayesian data modeling cultures from philosophical and computational perspective. We will show that our formulation can be interpreted from surprisingly diverse perspectives. 

\subsection*{A1. Robust Bayesian Methods} 
Our view of going from a unique prior assumption to a class of priors for robust Bayesian modeling was shaped by the Jim Berger's outstanding article \cite{Berger1994robust}. In the same spirit of the $\epsilon$-contamination class \cite{berger1986robust}, our U-function $d(u;G,\Pi)$ can be thought of as an automatic robustifier for standard (conjugate) priors. Thus, our approach may attain similar goals in a more computationally friendly way. Finally, we completely agree with Berger \cite{Berger1994robust} that `The major objection of non-Bayesians to Bayesian analysis is uncertainty in the prior, so eliminating this concern can make Bayesian methods considerably more appealing.'

\subsection*{A2. Empirical Bayes Methods}
Empirical Bayes approaches use data to determine the prior. While parametric empirical Bayes [PEB] \cite{morris1983parametric} fixes the hyperparameters based on the data, nonparametric empirical Bayes [NEB] \cite{efron2003robbins} makes no assumptions on the prior's form and develops it based solely on the data. In particular, Brad Efron \cite{efron1996empirical,efron2014bayes} advocate a \textit{smooth} nonparametric exponential family model: $\log \pi(\te)=\sum_{j=0}^m \be_j \te^j$ for the prior distribution where ${\bm \be}=(\be_0,\ldots,\be_m)$ is estimated by maximizing the marginal log-likelihood function. 
\vskip1em

{\bf Example 1}. The dotted line in Figure \ref{fig:app_connections}(a) denotes the non-parametrically estimated Efron's $\widehat{\pi}$ based on two-dimensional sufficient vector $S = (\theta, \theta^2)$ for the ulcer data \cite{efron1996empirical}. At a first glance, it appears strikingly close to the conjugate normal prior $\cN(-1.17, 0.98)$, marked as the bold red line. Perhaps the reader may be curious to know whether `$\pi(\te) \equiv$ PEB Normal' here? This is indeed the case, as already shown in Figure \ref{fig:sec3_3CD}(b) of the main paper. Our generalized empirical Bayes (gEB) framework automatically reduces to PEB when the data is consistent with the assumed parametric prior and modifies it non-parametrically otherwise. The output of the combined inference from $k=40$ clinical trials is shown as a green triangle $-1.17 \pm 0.197$, which is quite close\footnote[2]{The slight gain in accuracy for our method lies in the style of estimation that proceeds \textit{via} goodness-of-fit. Constructing prior by validating its credibility (using frequentist criterion) may also strengthen the Bayesian objectivity that Brad Efron \cite{efron1986isn} alluded to his article ``Why isn't everyone a Bayesian?''} to the Efron's nonparametric answer \cite{efron1996empirical} $-1.22 \pm 0.26$. The negative macro-estimate of the log-odds ratio parameters suggests that the new surgical treatment for stomach ulcers is overall more effective than the existing one.
\vskip.25em
Another attractive NEB technique is based on non-parametric maximum likelihood estimate (NPMLE): maximize the log-likelihood $\sum_{i=1}^k \log \big\{ \int f(y_i|\te)\dd \Pi(\te) \big\}$ over the set of all $\pi(\te)$ on $\cR$, which is known to be a notoriously difficult problem. Thanks to Gu and Koenker \cite{koenker2016}, an approximate NPMLE can be estimated via convex optimization technique (interior point method) instead of classical EM (Expectation-Maximization) algorithm \cite{laird1978}, thereby making it a computationally feasible alternative.
\begin{figure}[t]
\begin{subfigure}{.32\textwidth}
  \centering
  \includegraphics[width=\linewidth,trim=1cm .5cm 0cm 1cm]{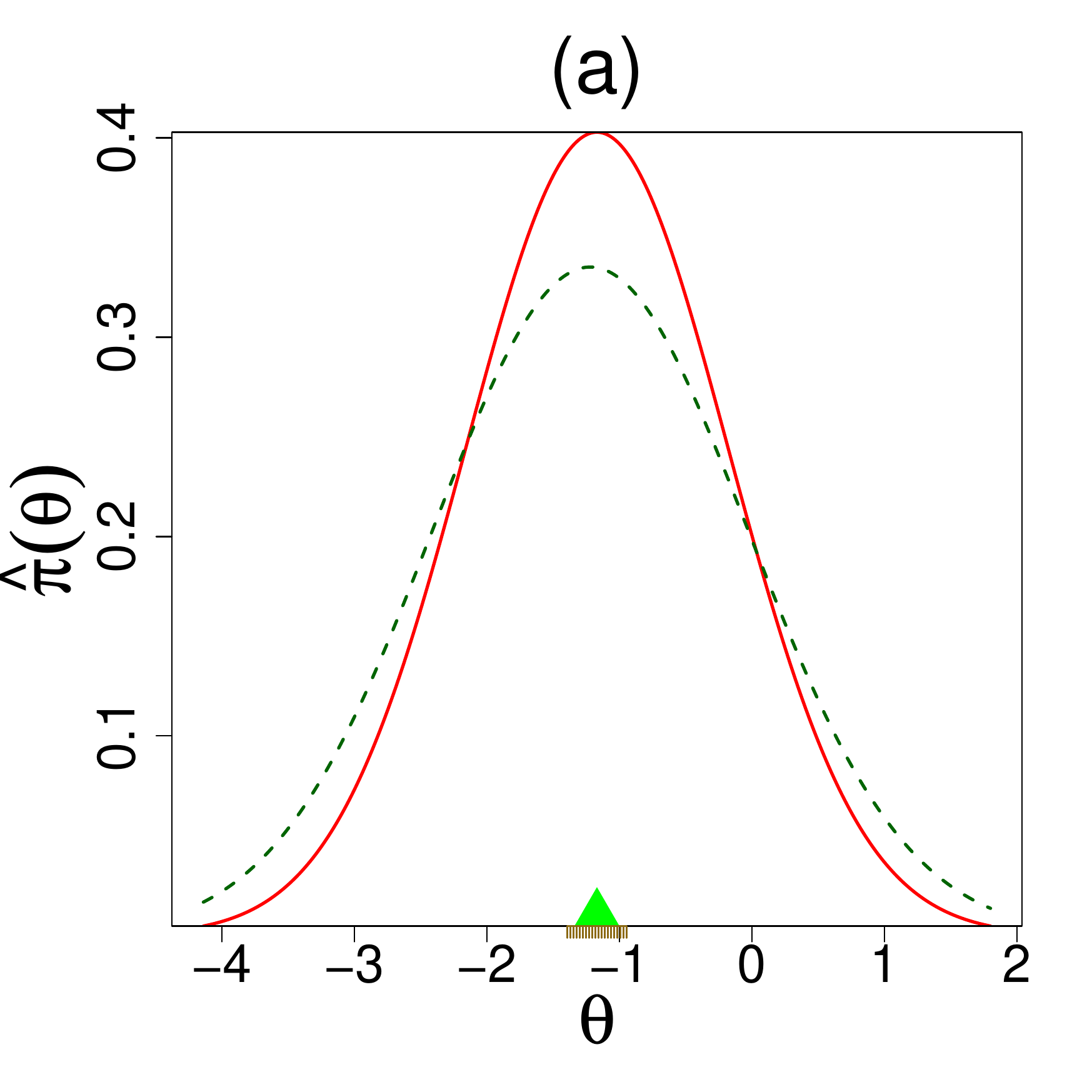}
\end{subfigure}\hspace{1.5mm}%
\begin{subfigure}{.32\textwidth}
  \centering
  \includegraphics[width=\linewidth,trim=1cm .5cm 0cm 1cm]{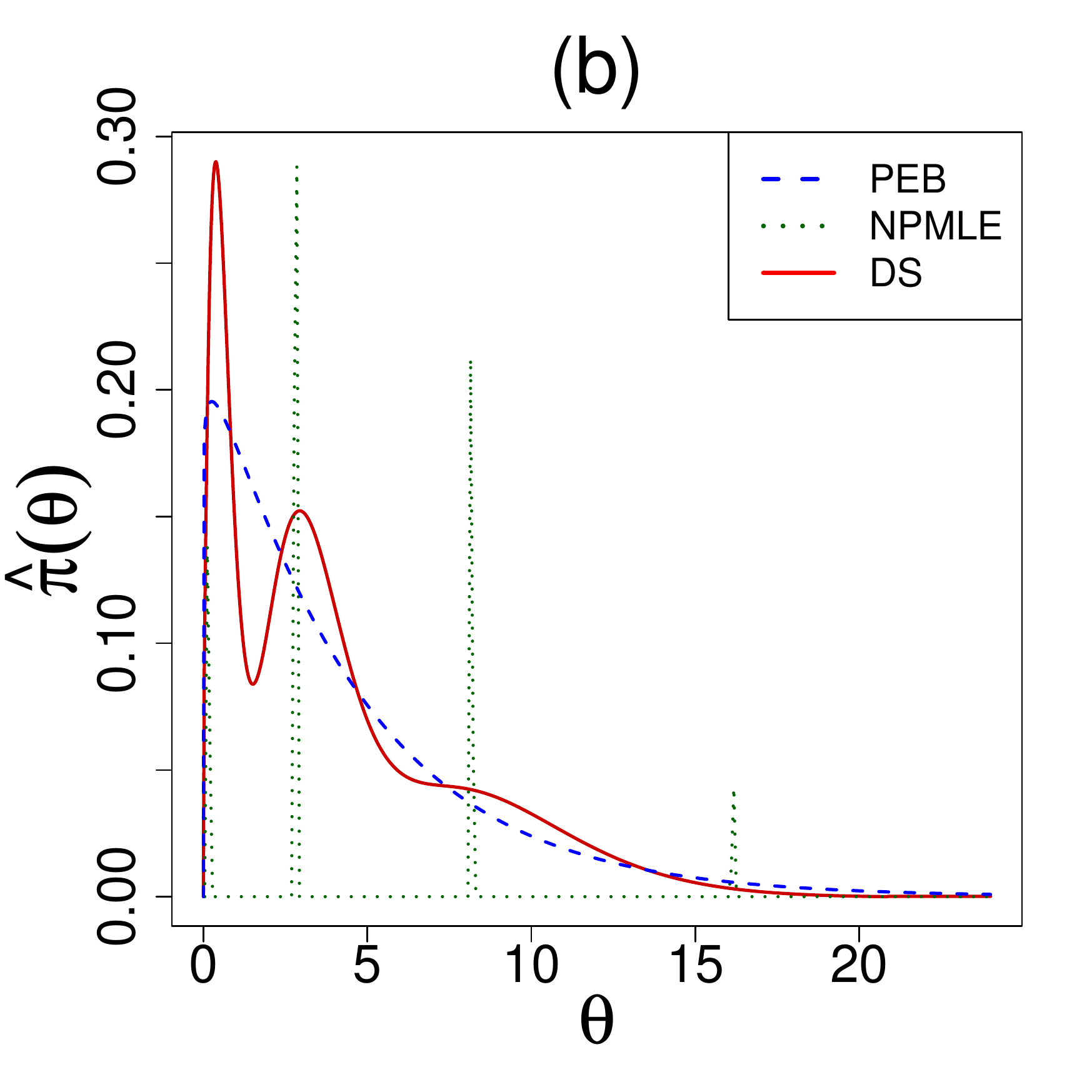}
\end{subfigure}\hspace{1.5mm}%
\begin{subfigure}{.32\textwidth}
  \centering
  \includegraphics[width=\linewidth,trim=0cm .5cm 1cm 1cm]{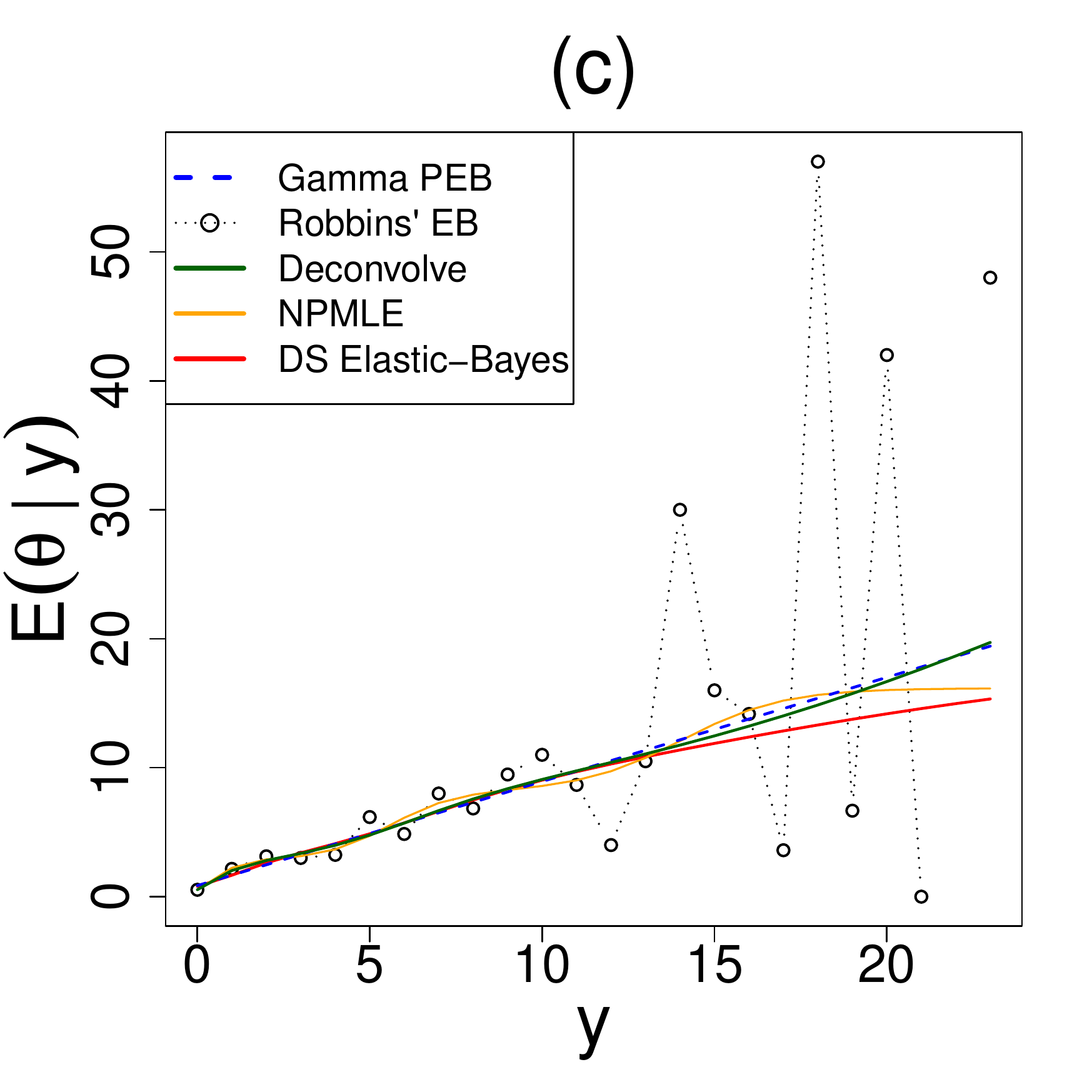}
\end{subfigure}
\caption{Comparisons of $\DS(G,m)$ (red) with other empirical Bayes modeling cultures (green): (a) The DS-estimated prior is compared with Efron's exponential prior model \cite{efron1996empirical}; (b) The DS distribution for the child illness data compared to NPMLE (the dotted line); (c) Estimates for the number of illnesses in the following year $\hat{\Ex}(\theta \mid x)$ by Gamma PEB, Robbins' formula,  Bayesian deconvolution, NPMLE, and our elastic-Bayes estimate.}
\label{fig:app_connections}
\end{figure}
\vskip1em 
{\bf Example 2}. NPMLE imposes no structural constraint and produces an estimated prior as discrete measure supported on at most $k$ points within the data range. Figure \ref{fig:app_connections}(b) shows its application to the child illness data \cite{wang2007fast}, which comes from a study that followed $k=602$ pre-school children in north-east Thailand from June 1982 through September 1985.  Researchers recorded the number of times ($y$) a child became sick during every 2-week period. Using the DS-Bayes method, we have $\hat{\pi}(\theta)$ where $g(\theta)$ is a gamma distribution with $\hat{\alpha} = 1.06$ and $\hat{\beta} = 4.19$ as
\begin{equation}\label{eq:ill_pi_hat}
\hat{\pi}(\theta) = {\rm Gamma}(\te;\alpha,\beta)\big[1 - 0.13T_3(\theta;G) - 0.28T_6(\theta;G) \big].
\end{equation} 
Our method produces a smooth, grid-free $\widehat \pi$ that accurately captures the overall shape. Figure \ref{fig:app_connections}(c) plots the Bayes estimates $\Ex[\Te_i|Y=y]$ for all competing methods. For Efron's \texttt{Deconvolve} we have used c0 = 2 and  pDegree = 25, which seems to produce a reasonable prior density estimate for this example. A careful look at the plot reveals an `oscillating' NPMLE Bayes estimates (orange curve), which many not be particularly desirable.
\begin{table}[t]
\centering
\def\arraystretch{.25}
\caption{\label{tbl:computation_time} Run-time comparisons between DS-Bayes and two other BNP methods: Dirichlet prior (DP), and Bernstein-Dirichlet (BDP) model. All methods were run using an Intel\textregistered Core\textsuperscript{TM} i5-7200 CPU @ 2.50GHz. \texttt{DPpackage} uses \texttt{C++} complier to speed-up, while ours is a prototype version implemented in \texttt{R}.}
\scalebox{.85}{
\begin{tabular}{lcccccc}
  \toprule
\multirow{2}{*}{Data Set}  &  \multirow{2}{*}{\# Studies ($k$)} & DS & DP & Ratio  & BDP & Ratio \\
& & Time & Time  & DP to DS & Time  & BDP to DS \\ 
  \midrule
 Rat Tumor & 70 & 1.83 & 10.42 & 5.69 & 3457.75 & 1889.5 \\[1em]
 Surgical Node & 844 & 30.95 & 189.3 & 6.12 & 45292.15 & 1463.4\\[1em]
 Terbinafine & 41 & 1.7 & 5.46 & 3.2 & 1883.18 & 1107.8\\[1em]
 Rolling Tacks & 320 & 8.27 & 59.16 & 7.15 & 16569.78 & 2003.6\\[1em]
 Arsenic & 28 & 0.47 & 13.09 & 27.8 & 433.29 & 254.9\\
\bottomrule
\end{tabular}
}
\end{table}
\subsection*{A3. Dirichlet-Process-based Approaches}
Bayesian nonparametric [BNP] technique assigns prior
distribution on infinite-dimensional spaces of probability models. The majority of work on Bayesian nonparametrics utilizes a Dirichlet process prior \cite{ferguson1973}. The computational cost of BNP is severe and produces prior on a set of discrete probability measures that demands an additional layer of smoothing. Figure \ref{fig:DP_rattack_comp} contrasts Dirichlet-process based Beta-Binomial models \cite{liu1996nonparametric} with our DS-Bayes model. There are few remarks warranted here:
\begin{itemize}[itemsep=1.24pt,topsep=1.24pt]
  \setlength{\itemsep}{1pt}
\item BNP method requires careful tuning of several hyper-priors values, which from our experience can be quite sensitive (see Figure \ref{fig:DP_rattack_comp}). Without practical guidance, this ``fishing expedition'' can potentially overwhelm one who seeks to confidently use it in practice. On the contrary, our method finds practically the same answer without adjusting multiple hyper-prior values. 
\item The posterior inferences of BNP are highly complex and require computationally expensive MCMC. In contrast, the beauty of our approach is that it provides compact analytical expressions that make the computation much more amicable.
\item The flexibility of BNP comes with the heavy task of estimating a massive number of parameters--``massively parametric Bayes.'' Contrast this with $\DS(G,m)$ model, which provides a reduced-dimensional characterization of the prior distribution with a closed form solution that is computationally efficient (see Table \ref{tbl:computation_time}) and produces smooth estimates in one-shot. For additional comments see the `Critical Appraisal' section.
\end{itemize}

\begin{figure}[t]
\centering
\begin{subfigure}{.35\textwidth}
  \centering
  \includegraphics[width=\linewidth,trim=1cm 0cm 1cm 1cm]{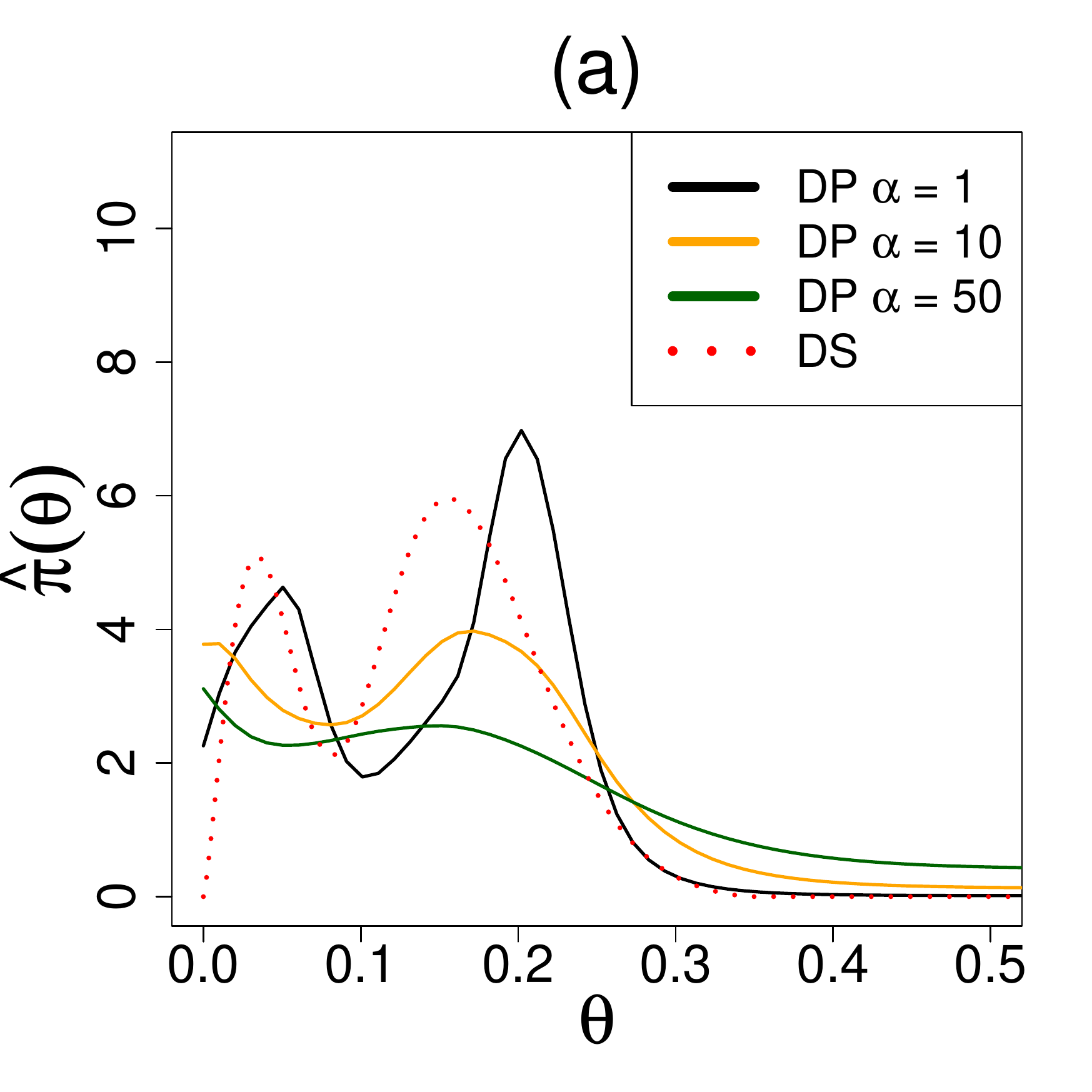}
\end{subfigure}\hspace{1.75cm}%
\begin{subfigure}{.35\textwidth}
  \centering
\includegraphics[width=\linewidth,trim=1cm 0cm 1cm 1cm]{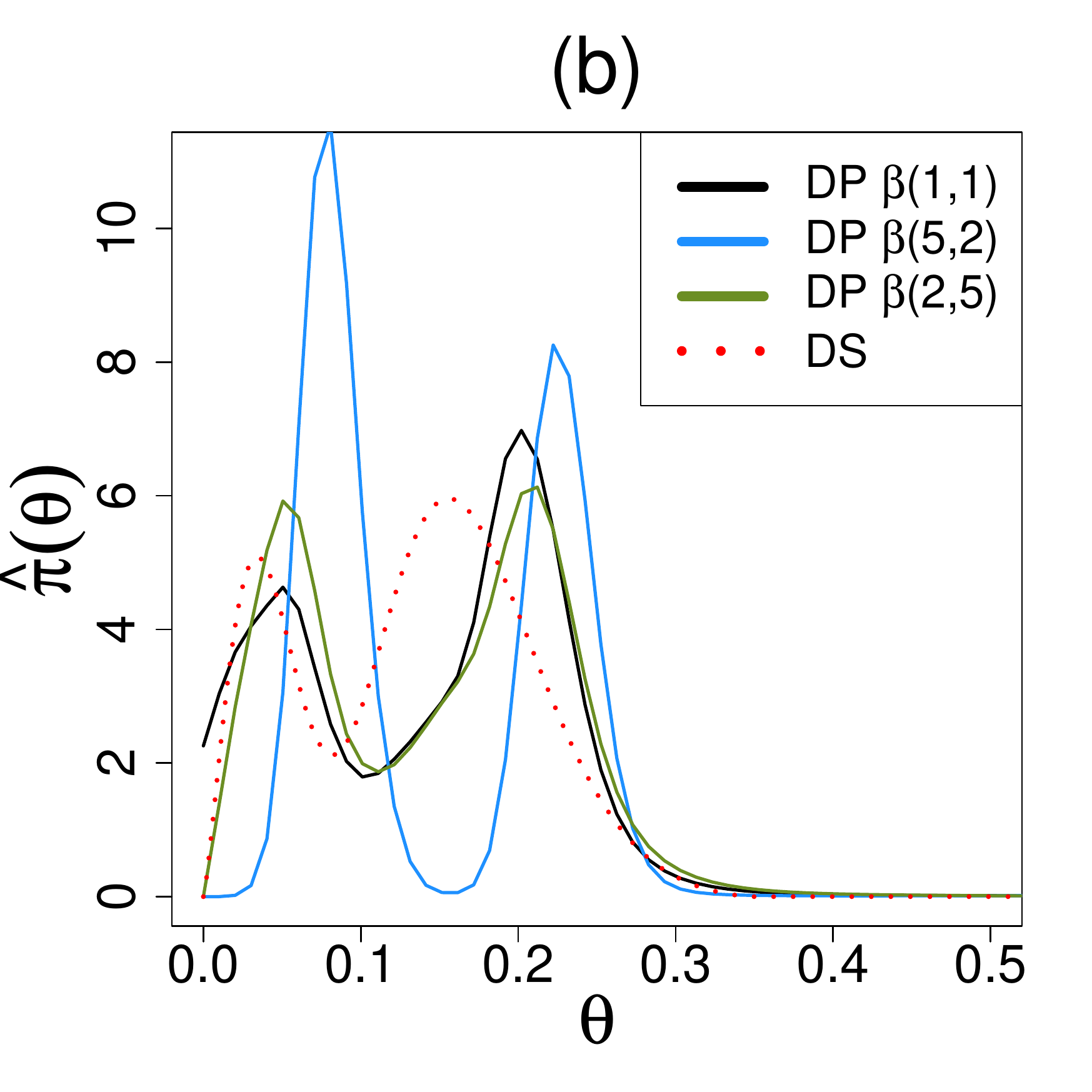}
\end{subfigure}\hspace{1mm}%
\begin{subfigure}{.35\textwidth}
  \centering
  \includegraphics[width=\linewidth,trim=1cm 1cm 1cm 0cm]{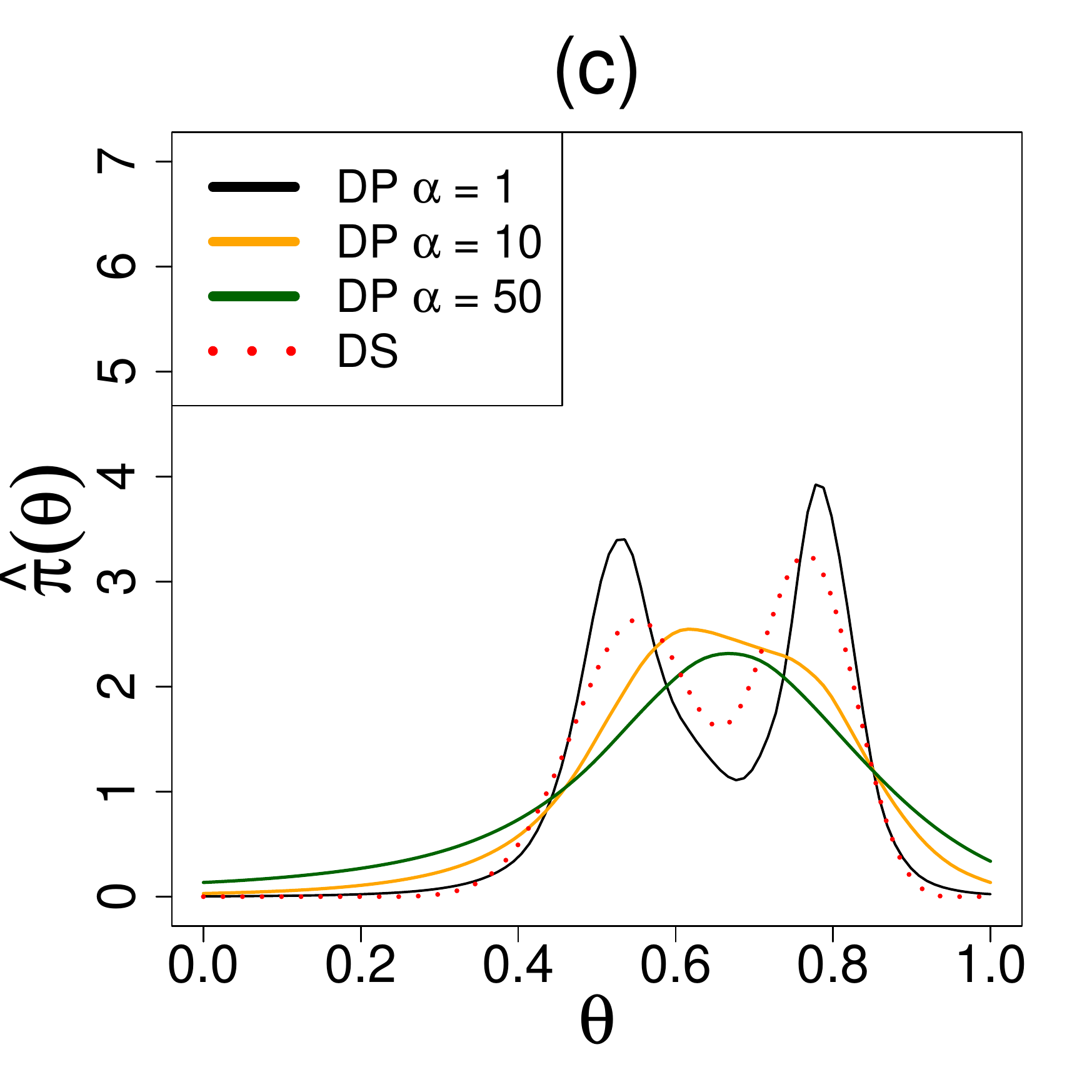}
\end{subfigure}\hspace{1.75cm}%
\begin{subfigure}{.35\textwidth}
  \centering
  \includegraphics[width=\linewidth,trim=1cm 1cm 1cm 0cm]{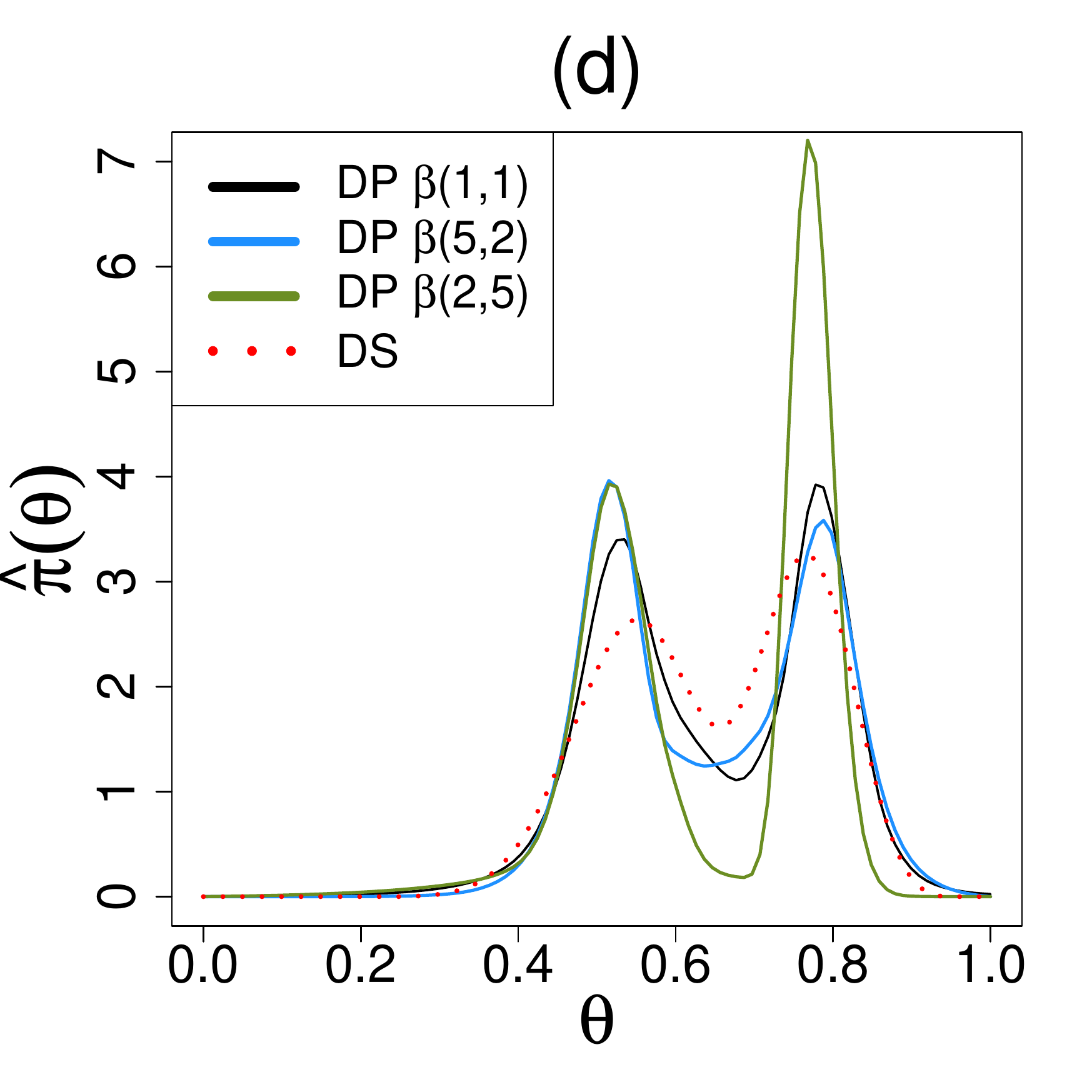}
\end{subfigure}
\vskip.4em
\caption{Illustrations of the different settings for BNP modeling with a Dirichlet process prior. Panel (a) displays results for the rat tumor data using uniform base prior while varying $\alpha$.  Panel (b), also for the rat tumor data, fixes $\alpha = 1$ and varies the base prior between uniform, ${\rm Beta}(5,2)$ and ${\rm Beta}(2,5)$.  Panels (c) and (d) use the same settings as (a) and (b), but applied to the rolling tacks data.}
\label{fig:DP_rattack_comp}
\end{figure} 
\subsection*{A4. Weakly Informative Priors}
A weakly informative prior [WIP] is a proper prior that intentionally provides less information than available prior knowledge.  This lies somewhere between a fully subjective and a fully objective prior \cite{gelman2008weakly, gelman2013bayesian}.

One can also view our approach from a WIP-angle where $d(u;G,\Pi)$ acts as a ``spreading/weakening function'' of the subjective prior $g(\te)$, which we \textit{learn from the data}. In the $\DS(G,m)$ language: $m$ is the radius of spread; the larger the $m$, the greater possibility you allow for changing the shape (the process of weakening) of the presumed scientific prior distribution $g(\te)$. These analogies suggest that our concepts and notations might provide a systematic way to formulate the WIP philosophy by addressing the debates around ``WIP is a subjective prior with ad hoc large but bounded support.'' This reformulation can also bring some tangible computational gain.
\vskip1em
\subsection*{A Critical Appraisal} We close this section by highlighting some of the unique aspects and practical advantages of our technique: 
\vskip1em
\begin{itemize}[topsep=3pt]
  \setlength{\itemsep}{4pt}
\item  \textit{Clarifying the Motivation}: Let's start by reminding ourselves  that the core motivation behind the `Bayes \textit{via} goodness-of-fit' is more than just another recipe for estimating the prior from data. To understand the mysterious prior in a transparent and definitive way, it is critical to ask: How can we provide automatic protection from unqualified specifications of prior distribution? How do we assess the prior-uncertainty using exploratory graphical tools? How can we prescribe a revised statistical-prior starting from the user-specified scientific-prior?  As it stands, these fundamental questions are usually left unanswered in traditional Bayes framework and create a major obstacle for non-Bayesian practitioners to confidently use Bayesian tools. Consequently, there is a need to address these issues in a formal manner to bring much-needed transparency. This paper has taken some solid steps toward this goal with a methodology that is readily usable for wide-range of applied problems. We believe that our technology can become an integral part of applied Bayesian modeling.
\item \textit{Theoretical Novelty}: Our proposed theory, which is general enough to include almost all commonly-used models, yields analytic closed-form solutions for posterior modeling. This is noteworthy for the simple reason that none of the nonparametric methods mentioned above can stand by this claim.
\item \textit{Theoretical Simplicity}: The whole `Bayes \textit{via} Goodness-of-fit' framework can be developed starting from a few basic principles, without requiring any exotic theoretical treatment. This could add invaluable transparency to the theory and practice of (empirical) Bayesian statistics. 
\item \textit{Exploratory Side}: Our approach brings a distinct exploratory flavor into the empirical-Bayes modeling. It encourages interactive data analysis rather than blindly `turning the crank.' Through numerous examples, we demonstrated how this mode of operation often leads to more insights into the data that are typically infeasible under a business-as-usual Bayesian modus operandi.
\item \textit{Computational Side}: Simplicity of implementation and computational ease are the two hallmarks of our method. No expensive MCMC or even sophisticated optimization routines are required! We made a sincere effort to design a practical Bayesian data analysis tool that is both simpler to comprehend and easy to implement.
\item \textit{A Third Empirical Bayes Culture}. Our empirical Bayes approach is neither parametric nor nonparametric. As argued in Section \ref{sec:PoiSmo} (of the main paper), our algorithmic approach blends conventional PEB and Robbins-style full-fledged NEB. Our goal is to combine the best of both worlds, in the sense that the prior reduces to PEB (ulcer data example) when in fact the default parametric $g$ is appropriate, while in the event of prior-data conflict (rat tumor or child illness data), it automatically produces reliable nonparametric procedures.  And in this whole story, the U-function $d(u;G,\Pi)$ acts as the ``connector'' between these two extreme philosophies. Overall, we are hopeful that our Generalized EB (gEB) modeling framework might expedites the development of a new \textit{genre} of `unified' Bayesian algorithms \cite{berger2000bayesian} by leveraging the rich interplay between two extreme EB philosophies.

\end{itemize}
\vskip2em
\begin{center}
{\large B.  MORE INSIGHTS INTO THE LP-BASIS FUNCTIONS}
\end{center}
Here we will show the shapes of the LP-polynomials, focusing only the Binomial case. It works similarly for other families. 
\vskip.55em
The $\{T_j(\theta; G_{\al,\be})\}_{j\geq 1}$ denotes the class of orthonormal polynomials of the beta distribution with parameters $\alpha$ and $\beta$.  Let $T_j(\theta;G_{\al,\be}) = \Leg_j\{G_{\al,\be}(\te)\}$ and $G_{\al,\be}(\te) = \frac{1}{\B(\alpha,\beta)}\int_0^{\theta} \phi^{\alpha-1}(1-\phi)^{\beta-1} d\phi$.  Figure \ref{fig:App_scr_func} displays the shapes of top four LP polynomials for three different sets of parameters. We generate these polynomials with the following R code: 
\vskip1em
\begin{tcolorbox}[colback=gray!5!white,colframe=black]
\begin{verbatim}
LP.basis.beta <- function(y, g.par, m){
#######################################
##  g.par: parameters for the beta distribution
#######################################
	require(orthopolynom)
	u <- pbeta(y, g.par[1], g.par[2]) # computes G(y)
	poly <-  slegendre.polynomials(m,normalized=TRUE) 
	TY <- matrix(NA,length(u),m)
	for(j in 1:m) TY[,j] <- predict(poly[[j+1]],u)
	return(TY)}
\end{verbatim}
\end{tcolorbox}

\begin{figure}[t]
\begin{subfigure}{.32\textwidth}
  \centering
  \includegraphics[width=\linewidth,trim=1cm .5cm 0cm 1cm]{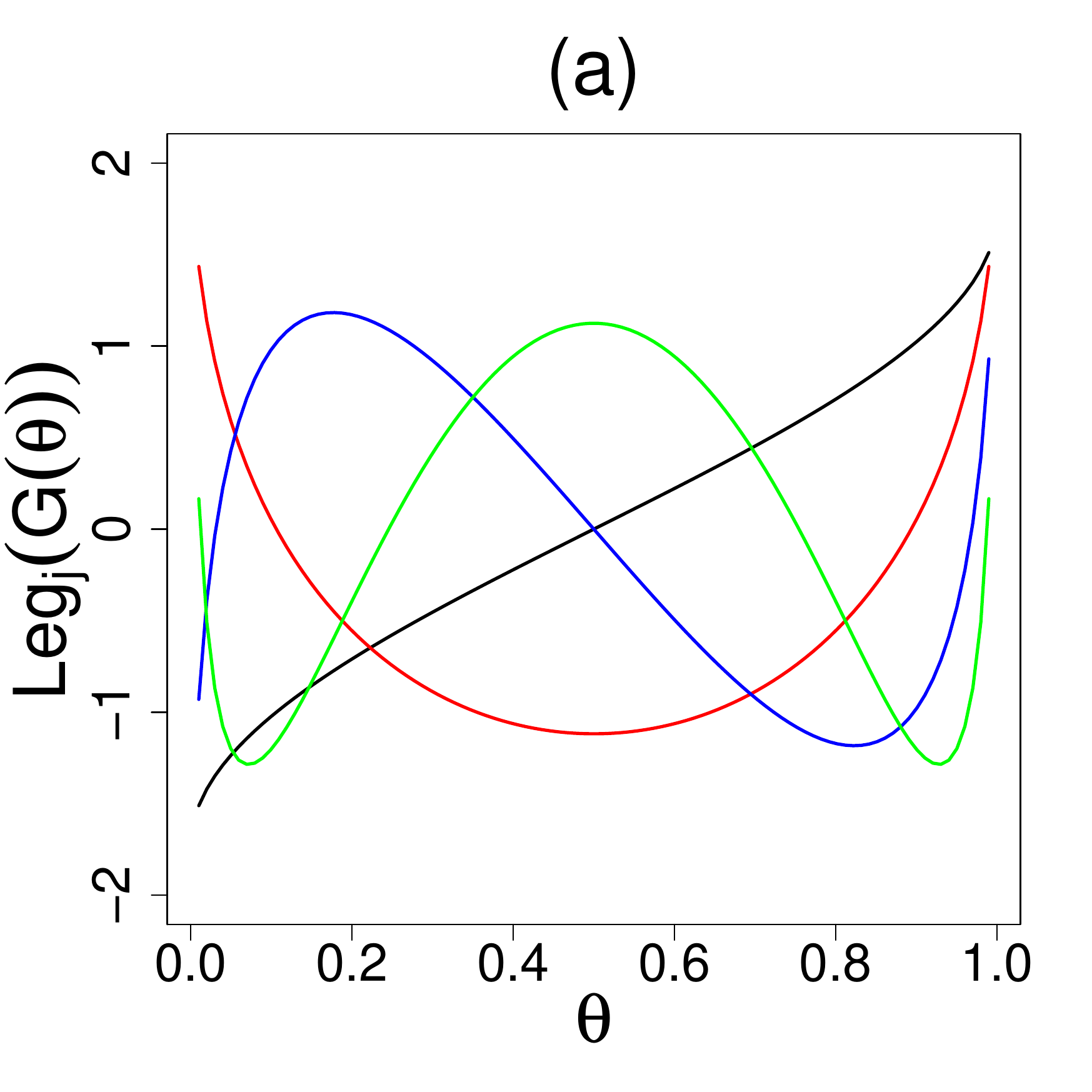}
\end{subfigure}\hspace{1.5mm}%
\begin{subfigure}{.32\textwidth}
  \centering
\includegraphics[width=\linewidth,trim=.5cm .5cm .5cm 1cm]{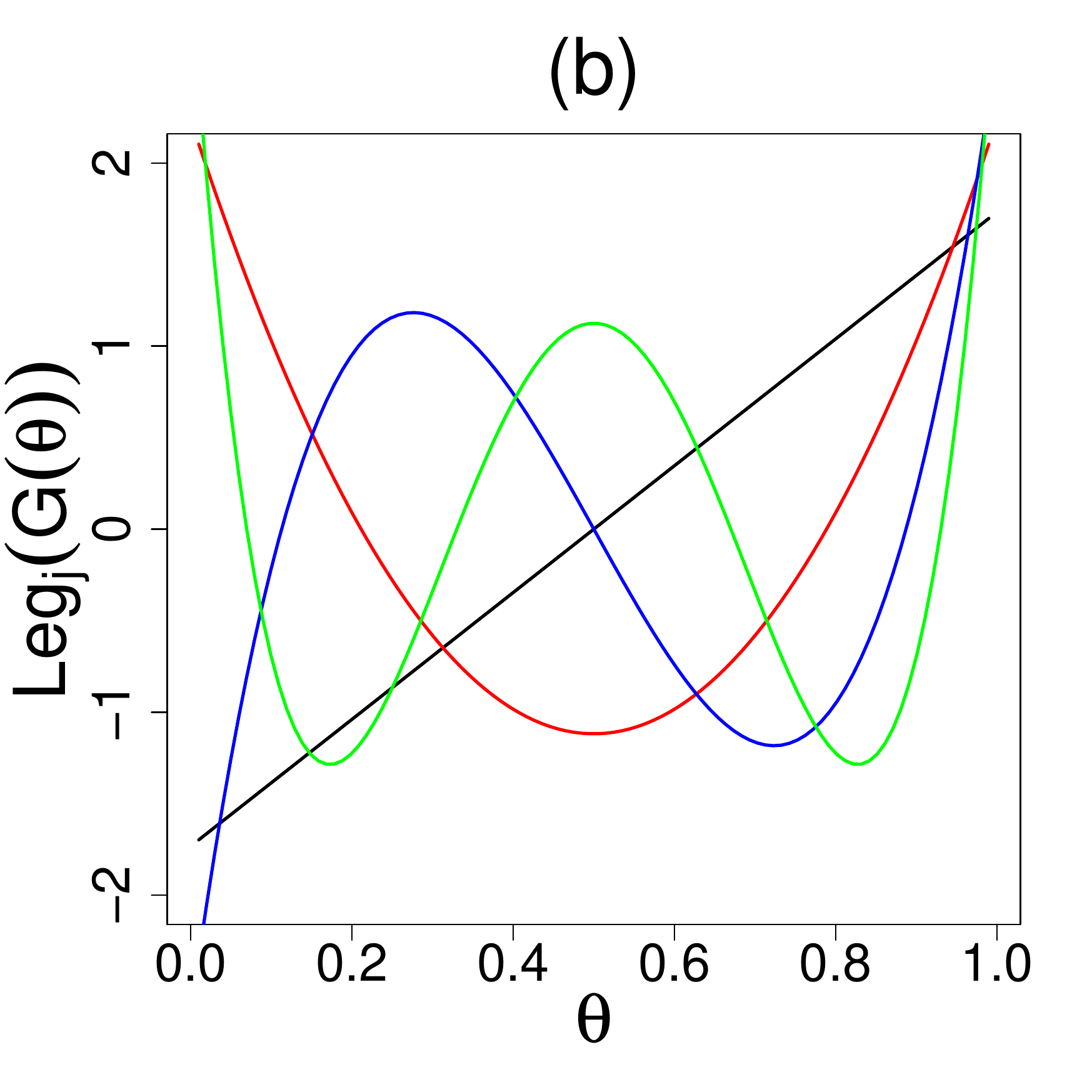}
\end{subfigure}\hspace{1.5mm}%
\begin{subfigure}{.32\textwidth}
  \centering
  \includegraphics[width=\linewidth,trim=0cm .5cm 1cm 1cm]{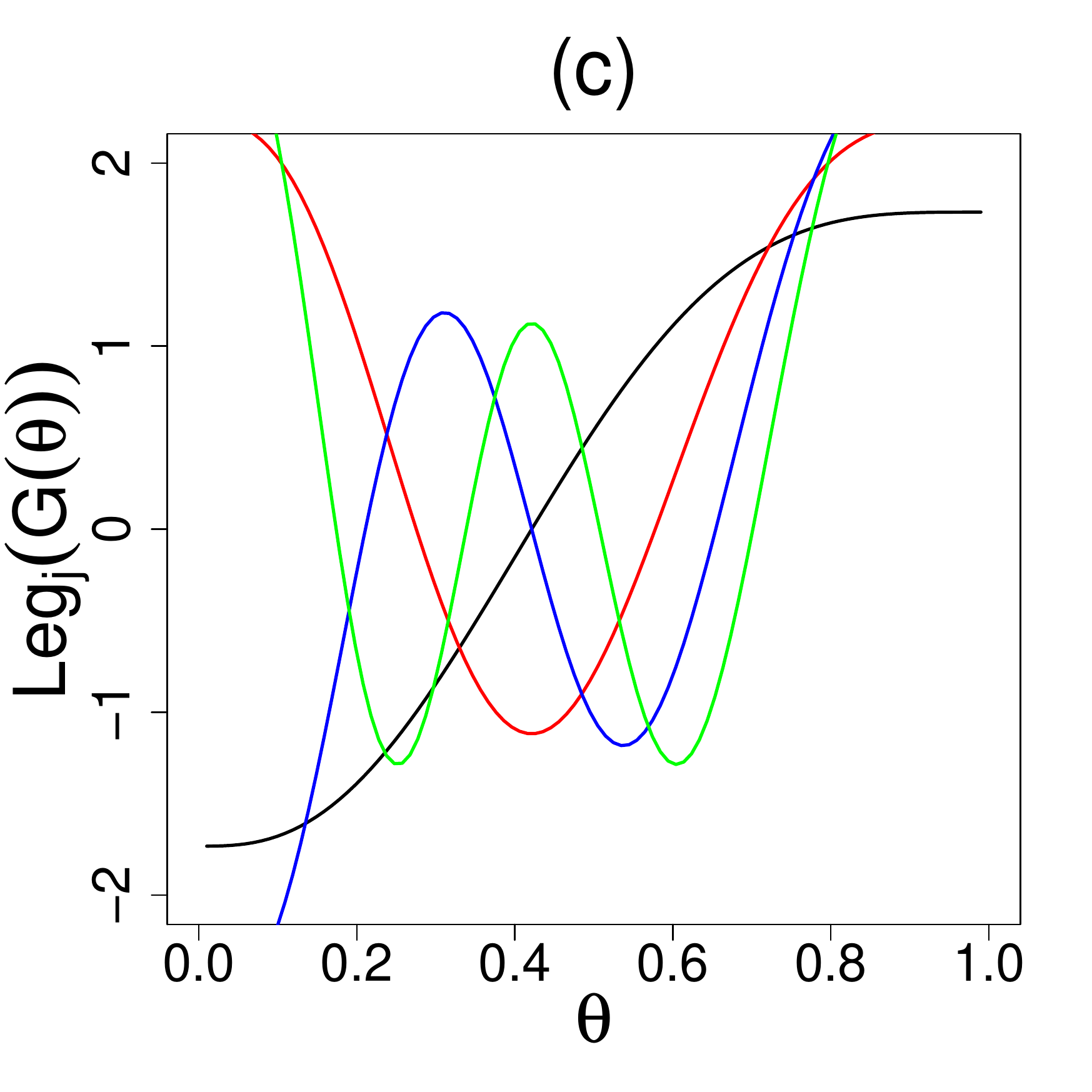}
\end{subfigure}
\vskip.3em
\begin{subfigure}{\textwidth}
  \centering
  \includegraphics[width=\linewidth,trim=0cm 5.2cm 0cm 4.7cm]{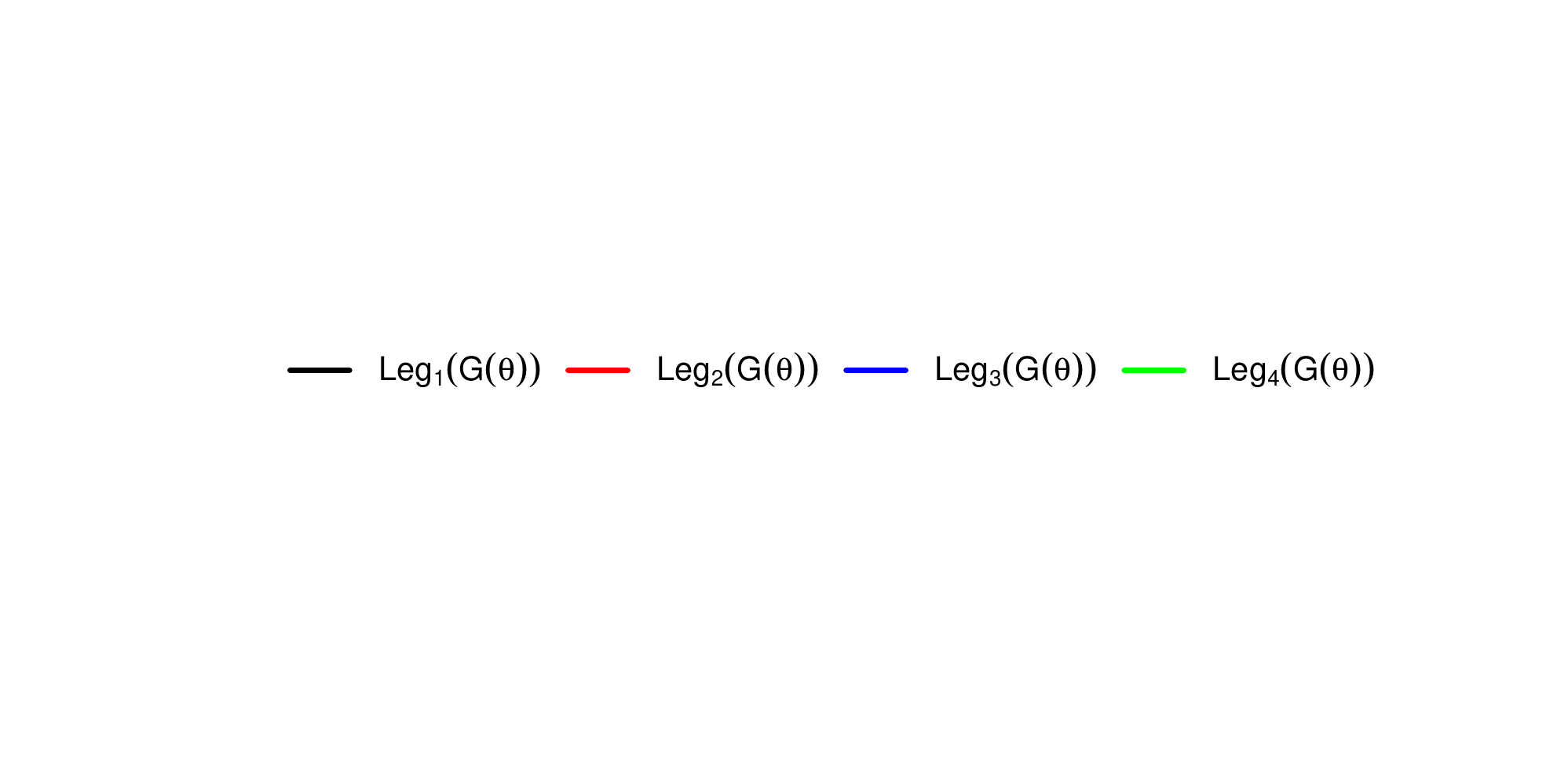}
\end{subfigure}
\vskip.4em
\caption{LP-polynomials $T_j(\theta;G_{\al,\be})$ for \texttt{family= "beta"} with the following $(\alpha, \beta)$ choices: (a) Jeffrey's prior ($\alpha = \beta = 0.5$), (b) uniform prior ($\alpha = \beta = 1$), and for (c) ${\rm Beta}(\alpha = 3,\beta = 4)$. Note that for $U[0,1]$ (the middle panel): $T_j \equiv \Leg_j$, as $G(\te)$ is simply $\te$ in this case.}
\label{fig:App_scr_func}
\end{figure}
\vspace{-.5em}
\begin{center}
{\large C. THE $\DS(G,m)$ SAMPLER}
\end{center}
The following algorithm generates samples from the $\DS(G,m)$ model via accept/reject scheme.
\vskip.5em
\makebox[\textwidth]{\textbf{$\DS(G,m)$ Sampling Algorithm}}
\rule{\textwidth}{.8pt}
\texttt{Step 1.} Generate $\Te$ from $g$; independent of $\Te$, generate $U$ from ${\rm Uniform}[0,1]$.
\vskip.4em
\texttt{Step 2.}  Accept and set $\Te^* = \Te$ if
\begin{equation*}
\widehat{d}[ G(\te); G,\Pi] > U \max_u\{\widehat{d}(u;G,\Pi) \};
\end{equation*}
otherwise, discard $\Te$ and return to Step 1.
\vskip.4em
\texttt{Step 3.} Repeat until simulated sample of size $k$, $\{\te_1^*,\te_2^*,\cdots,\te_k^*\}$.\\
\rule{\textwidth}{.8pt}
\vskip1em
Note that when $\widehat d \equiv 1$ then the $\DS(G,m)$ automatically samples from parametric $G$.  
\vskip2em
\begin{center}
{\large D. OTHER PRACTICAL CONSIDERATIONS}
\end{center}

In the event that \textit{no} prior knowledge is available, selecting the parametric conjugate prior $G$ with empirically estimated $\al,\be$ in conjunction with our Type-II Method of Moments algorithm (sec. \ref{sec:EstAlgo}) will provide a quick estimate of the oracle $\pi$. The algorithm finds the `best' approximating prior model given \texttt{m.max}: the maximum complexity that the subject-matter experts want to entertain. From our experience with $\DS(G,m)$ model, we found $\texttt{m.max}=8$ works satisfactorily well in practice (in fact in all our examples $8$ was our default choice), which encompasses the space of reasonable priors around $G$. Given this maximum radius, our method generates a deviance plot, where the  ``elbow'' shape (see Figure \ref{fig:appD_sim} (a)) denotes the most likely model dimension. This procedure is fully incorporated into our algorithm so that practitioners can use it in a completely automatic manner. 
\vskip.6em
\begin{figure}[t]
\begin{subfigure}{.32\textwidth}
  \centering
  \includegraphics[width=\linewidth,trim=1cm .5cm 0cm 1cm]{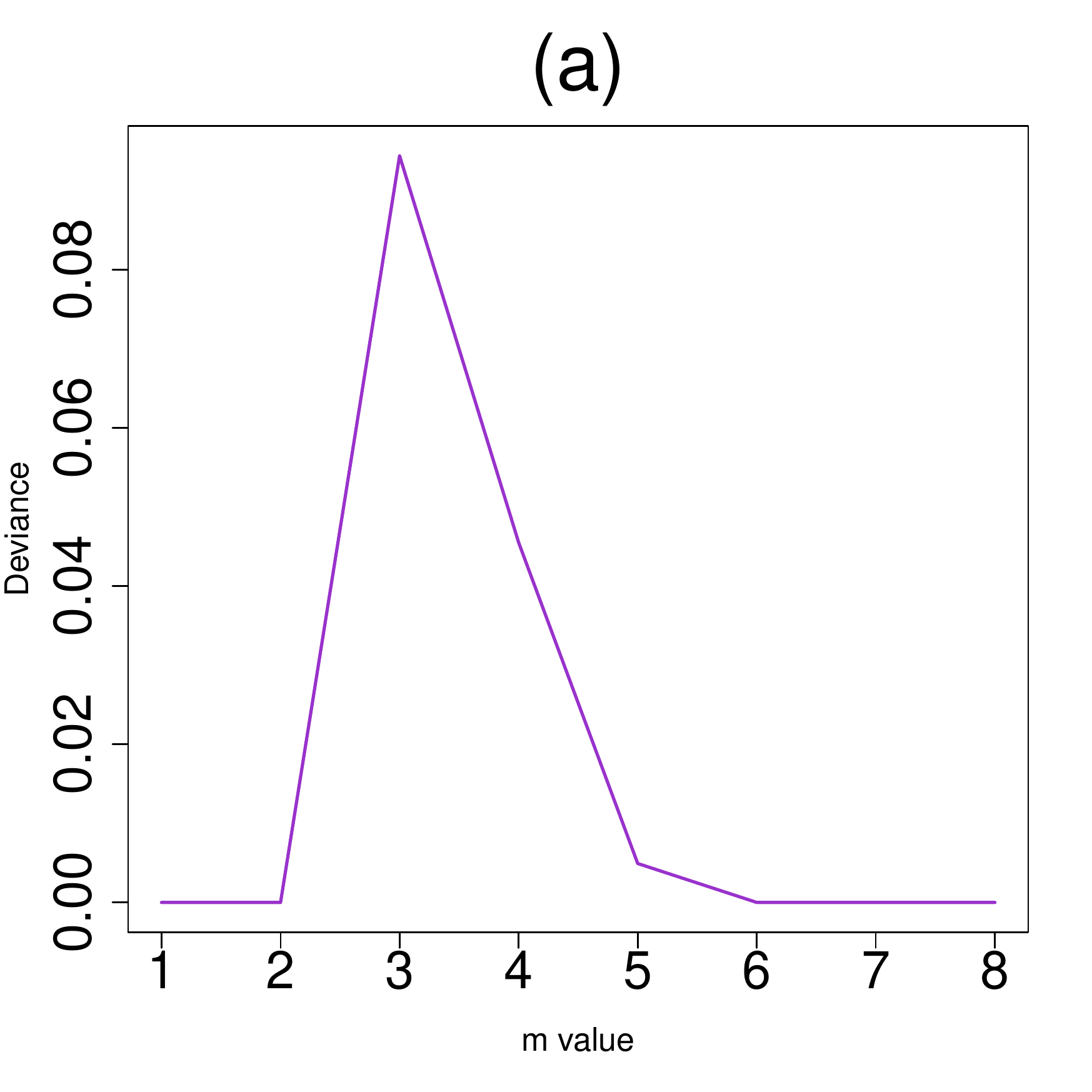}
\end{subfigure}\hspace{1.5mm}%
\begin{subfigure}{.32\textwidth}
  \centering
\includegraphics[width=\linewidth,trim=.5cm .5cm .5cm 1cm]{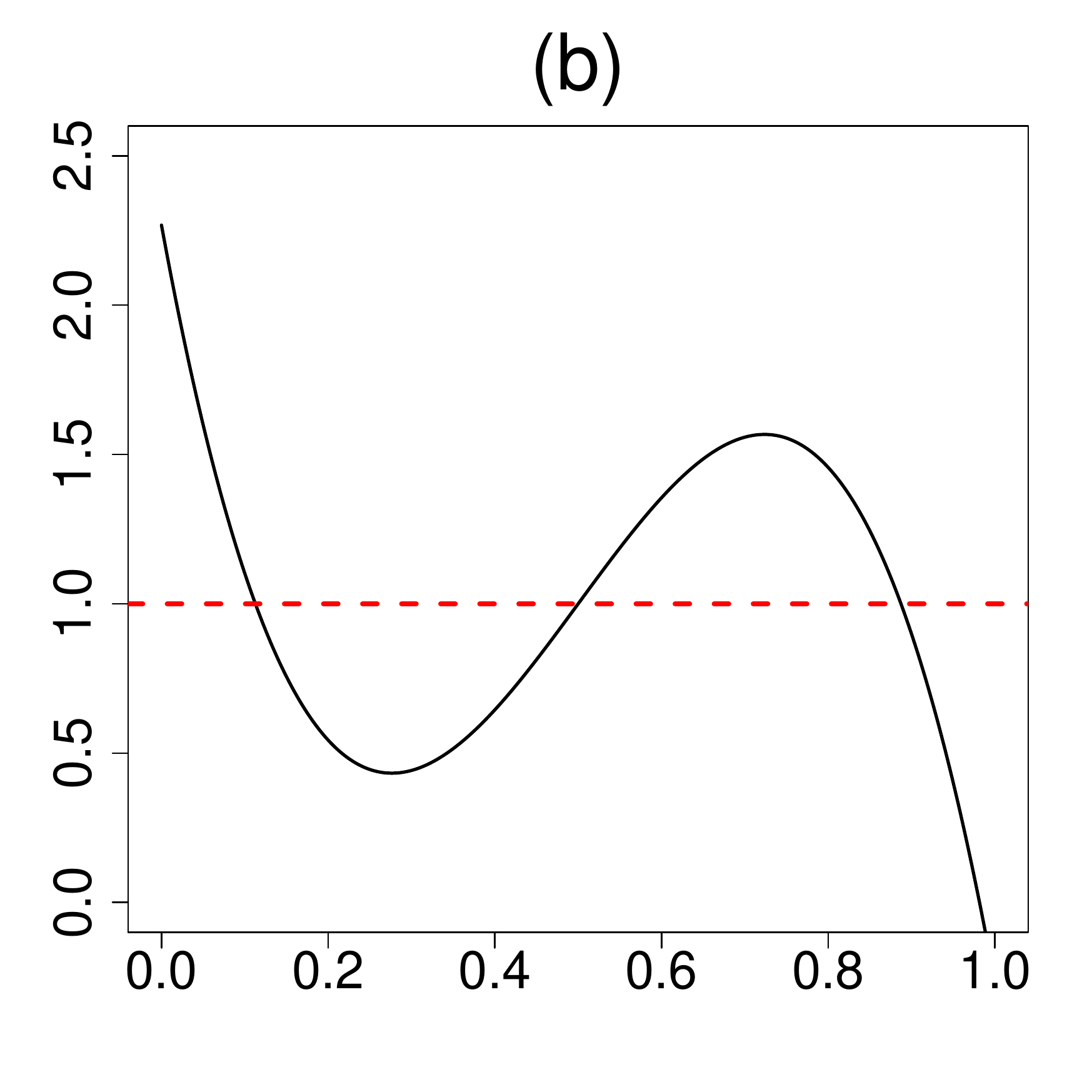}
\end{subfigure}\hspace{1.5mm}%
\begin{subfigure}{.32\textwidth}
  \centering
  \includegraphics[width=\linewidth,trim=0cm .5cm 1cm 1cm]{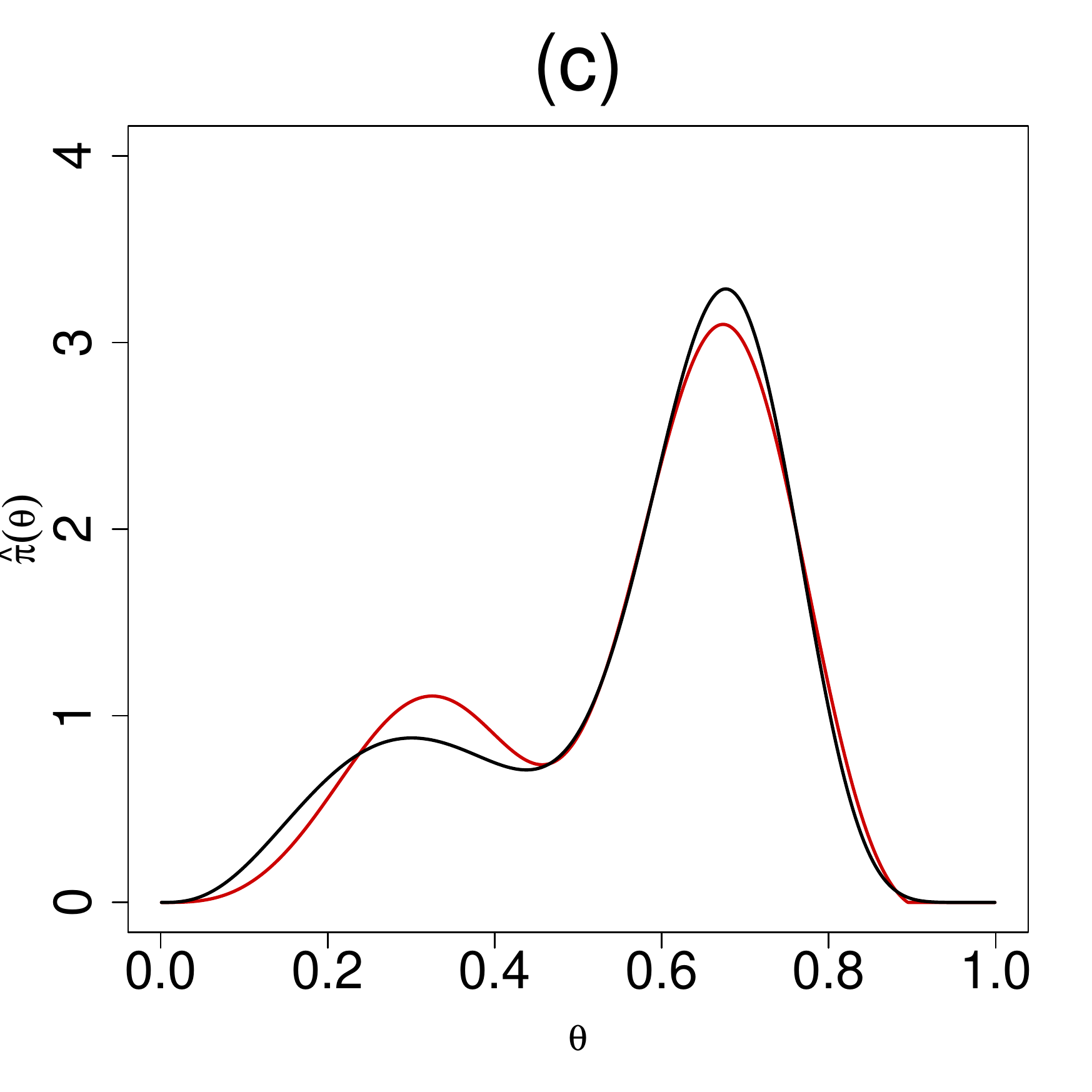}
\end{subfigure}
\caption{Analysis for simulated data based on Type-II Method of Moments algorithm. The first panel (a) finds the ``elbow'' in the ${\rm BIC}(m)$ deviance plot at $m=3$; (b) shows the U-function, while (c) plots the true $\pi(\theta)$ (black) along with the estimated DS prior (red) $\hat{\pi}(\theta) = g(\theta; \hat \alpha,\hat\beta)\big[1 - 0.48T_3(\theta;G) \big]$ with MLE $\hat \al=4.16$ and $\hat\be=3.04$.}
\label{fig:appD_sim}
\end{figure}

{\bf Illustration}. Consider the model: $y_i | \theta_i \sim \text{Binomial}(50,\te_i)$ with $i=1,\ldots,k=90$ and the true prior distribution $\pi(\te)=.3{\rm Beta}(4,6) +.7{\rm Beta}(20,10)$. Our goal is to see how well we can approximate the unknown $\pi$ without any prior knowledge of its shape. The following R code can be used to reproduce our findings reported in Figure \ref{fig:appD_sim}.
\vskip1em
\begin{tcolorbox}[colback=gray!5!white,colframe=black]
 \begin{verbatim}
set.seed(8701)
k <- 90
n.i <- 50
n.vec <- rep(n.i,k)
k1 <- ceiling(.7*k)
#Test Simulation: Mixed beta Distribution
theta.sim <- c(rbeta(k1,20,10), rbeta( (k-k1),4,8))
y.sim <- sapply(theta.sim, rbinom, size = n.i, n = 1)
sim.df <- data.frame(y = y.sim, N = n.vec)
##Run Type II MoM Algorithm
sim.start <- gMLE.bb(sim.df$y,sim.df$N)$estimate 
sim.LP.par <- DS.prior(sim.df, g.par = sim.start, family = "Binomial")
\end{verbatim}
\end{tcolorbox}
\vskip1em
The \texttt{sim.start} object holds the MLE estimate for the initial parameters for $G$.  From the \texttt{sim.LP.par} object, we generate diagnostic and analysis plots for appropriate $m$, U-function, and the $\DS(G,m)$ estimate, as shown in Figure \ref{fig:appD_sim}.




\begin{center}
{\large E. SOFTWARE} 
\end{center}
We provide an \texttt{R} package, \texttt{BayesGOF} \cite{bayesgofR} to perform all the tasks outlined in the paper.  We now summarize the main functions and their usage for the Rat binomial data example: 
\vskip.75em
\begin{tcolorbox}[colback=gray!5!white,colframe=black]
 \begin{verbatim}
 #Phase I: Modeling
library("BayesGOF")
data(rat)
rat.start <- gMLE.bb(rat$y, rat$n)$estimate
rat.ds <- DS.prior(rat, g.par = rat.start, family = "Binomial")
plot(rat.ds, plot.type = "Ufunc") # Figure 1(a)
plot(rat.ds, plot.type = "DSg") # Figure 2(a)
\end{verbatim}
\end{tcolorbox}
\vskip.4em
The package also provide functionalities for Macro and MicroInference:  
\vskip.9em
\begin{tcolorbox}[colback=gray!5!white,colframe=black]
 \begin{verbatim}
 #Phase II: Inference
rat.ds.macro <- DS.macro.inf(rat.ds, num.modes = 2, method = "mode")
plot(rat.ds.macro) # Figure 3(a)
rat.ds.pos <- DS.micro.inf(rat.ds, y.0 = 4, n.0 = 14)
plot(rat.ds.pos)  # Figure 5(b)
\end{verbatim}
\end{tcolorbox}
\vskip.5em
We hope this software will encourage applied data scientists to apply our method for their real problems. 
\vskip.8em
\begin{center}
{\large F. DATA CATALOGUE}
\end{center}
\begin{table}[H]
\vskip.5em
\setlength{\tabcolsep}{18pt}
\def\arraystretch{1.15}
\centering
{\footnotesize
\caption{\label{tbl:DataSummary} List of datasets by distribution family and sources. They are sorted first by family, then according to $k$: from large to small-scale studies.} \vskip.25em
\begin{tabular}{lccl}
  \toprule
Dataset  & \# Studies ($k$)&Family~~ & ~~~Sources\\
  \midrule
 Surgical Node & 844 & Binomial &~~ Efron (2016) \cite{efron2014bayes}\\
 Rolling Tacks & 320 & Binomial &~~ Beckett and Diaconis (1994) \cite{beckett1994spectral}\\
Rat Tumor & 70 & Binomial  &~~ Gelman et al. (2013, Ch. 5) \cite{gelman2013bayesian}\\
 Terbinafine &  41 &  Binomial &~~ Young-Xu and Chan (2008) \cite{young2008pooling} \\ 
 Naval Shipyard & 5 & Binomial &~~ Martz et al. (1974) \cite{martz1974empirical} \\ \hdashline
Galaxy & 324 & Gaussian &~~~De Blok et al.(2001) \cite{de2001high}\\
 Ulcer & 40 & Gaussian &~~~Sacks et al.(1990) \cite{sacks1990endoscopic}\\ 
 Arsenic & 28 & Gaussian & ~~~Willie and Berman (1995) \cite{willie1995ninth}\\\hdashline
Insurance & 9461 & Poisson & ~~~Efron and Hastie (2016) \cite{efron2016computer}\\
Child Illness & 602 & Poisson & ~~~Wang (2007) \cite{wang2007fast} \\
 Butterfly & 501 & Poisson & ~~~Fisher et al. (1943) \cite{fisher1943}\\
 Norberg & 72 & Poisson & ~~~Norberg(1989) \cite{norberg1989experience}\\
\bottomrule
\end{tabular}
}
\end{table}
\vskip1.5em

\begin{center} 
{\large G. THE ROBBINS' PUZZLE}
\end{center}
In Section \ref{sec:micro} of the main paper, we presented a simulated scenario ( Pharma-example) that demonstrated the power of the DS Elastic-Bayes estimate when there is significant prior-data conflict.  Here we include further comparisons with two recent methods:  Efron's Bayesian deconvolution (implemented in the \texttt{deconvolveR} package), and Koenker's NPMLE (implemented in the \texttt{REBayes} package). 
\begin{figure}[H]
\vskip.65em
  \centering
  \begin{subfigure}{.44\textwidth}
  \includegraphics[width=\linewidth,trim=1cm .5cm 0cm 1cm]{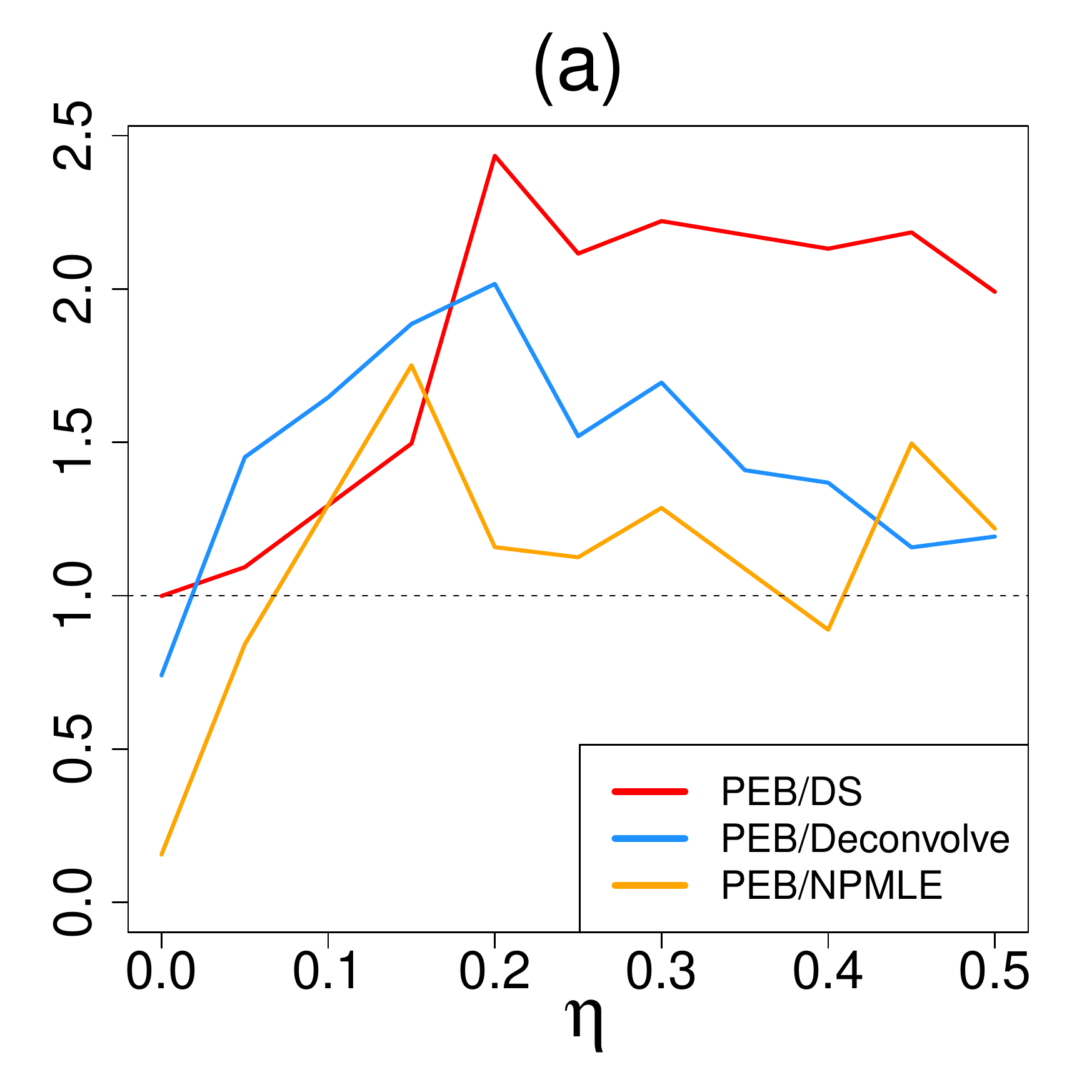}
  \end{subfigure}\hspace{12mm}%
  \begin{subfigure}{.44\textwidth}
    \includegraphics[width=\linewidth,trim=1cm .5cm 0cm 1cm]{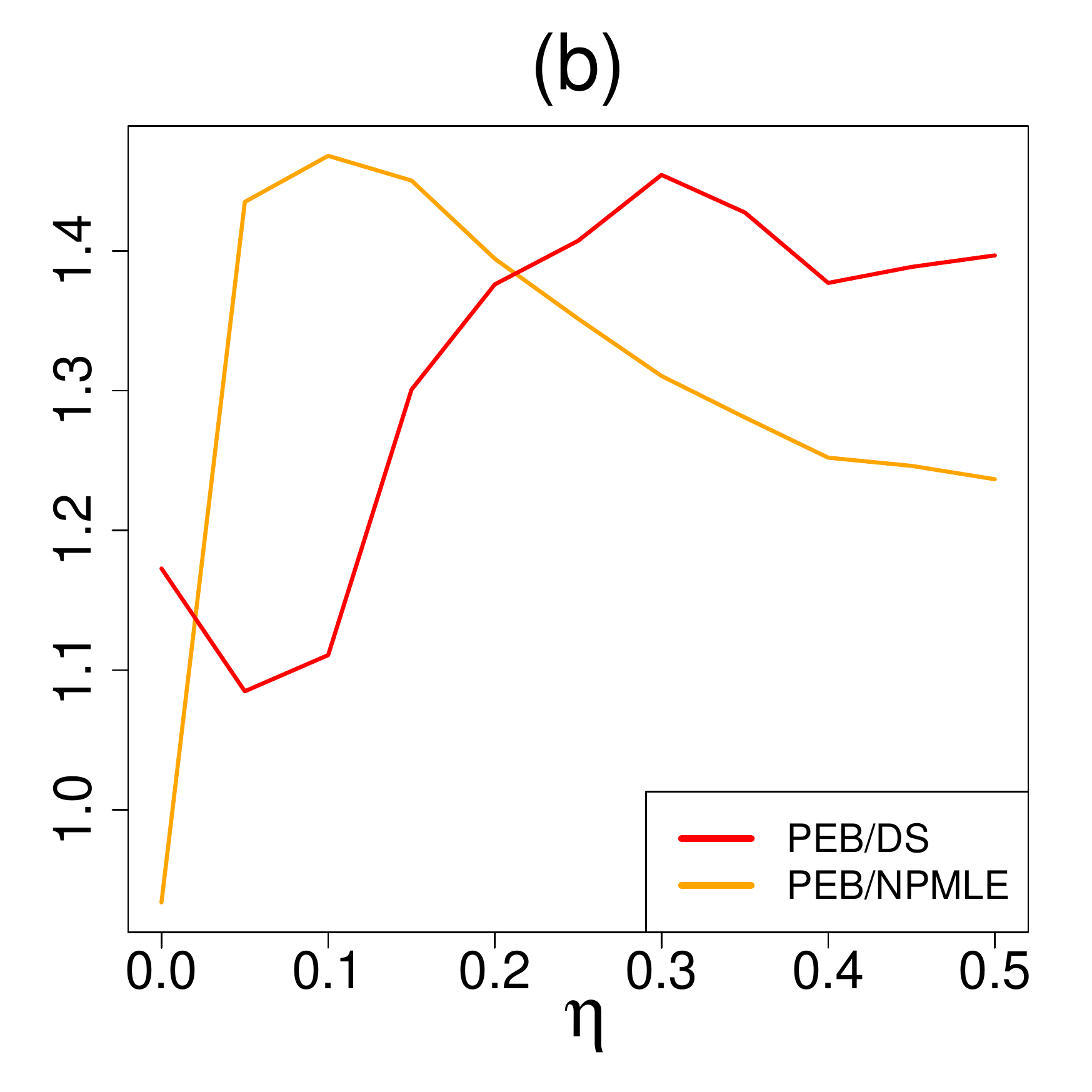}
  \end{subfigure}
  \vskip.35em
\caption{Results of two separate simulations comparing DS with other methods.  In (a), the MSE ratios for PEB to empirical Bayes deconvolution (PEB/Dec; blue), PEB to Kiefer-Wolfowitz NPMLE using REBayes \texttt{Bmix} (PEB/NPMLE; orange) and PEB to DS (PEB/DS; red) with respect to $\eta$. Panel (b) shows the ratio of empirical risks after applying both DS and NPMLE methods to Robbins' `compound decision' problem.}
\label{fig:appG_PDsim}
\end{figure}

{\bf Example 1}. 
Here we will operate under the exact settings presented in Section \ref{sec:micro}. Figure \ref{fig:appG_PDsim}(a) shows that as $\eta$ increases, DS tends to outperform the other two methods, although Deconvolve performs superbly for $\eta$ smaller than $0.15$. Two specially interesting extreme cases are $\eta = 0$ and $\eta = 0.5$. The first scenario describes the situation when the underlying parametric $beta$ distribution is the right choice for the prior where, as expected, the Stein's parametric shrinkage estimator dominates other nonparametric approaches. On the other hand, the $\eta = 0.5$ is a complicated situation where $\pi(\te)=\frac{1}{2} \text{Beta}(5,45) + \frac{1}{2}  \text{Beta}(30,70)$, and consequently, the parametric EB [PEB] is less efficient compared to the nonparametric ones. The most interesting and surprising result, however, comes from DS Elastic-Bayes, which acts like the Stein prediction formula when underlying parametric assumption is correct (the null $\eta = 0$ case) but adapts itself non-parametrically in a completely automated manner when the true $\pi(\te)$ deviates from the assumed $g$, thereby elegantly addressing the robustness-efficiency puzzle of Robbins \cite{robbins1980}. 

\vskip.5em
{\bf Example 2}. 
Next, we investigate the prediction problem originally introduced by Robbins \cite{robbins1985asymptotically} and discussed in Gu and Koenker \cite{koenker2016}. We observe $Y_i = \theta_i + \epsilon_i$, $i = 1 \cdots k$, where $\epsilon_i \overset{{\rm ind}}{\sim} {\rm Normal}(0,1)$, and $\theta_i =\pm 1$ with probability $\eta$ and $1-\eta$ respectively. Our goal is to estimate the $k$-vector $\theta \in \{-1,1\}^k$ under the loss $\mbox{L}(\hat{\theta},\theta) = k^{-1}\sum_{i=1}^k |\hat{\theta_i}-\theta_i|$. For comparison purpose, we computed the ratio of PEB empirical risk\footnote[2]{Mean loss is computed over $500$ replications.} to the the DS method (EB/DS) and to the NPMLE estimator (EB/KW) for $k = 1000$.  Figure \ref{fig:appG_PDsim}(b) shows a very interesting result:  Kiefer-Wolfowitz NPMLE method performs significantly better than the DS-elastic Bayes when $0< \eta < 0.2$. While for other values of $\eta$, including $\eta$ equals to zero point, our micro-estimation procedure demonstrates tremendous promise. This further validates the flexibility and adaptability of our technique even in the discrete settings. 

\begin{center} 
{\large H. EXAMPLE WITH COVARIATES}
\end{center}
The `Bayes via goodness-of-fit' methodology can easily accommodate additional covariates.  We demonstrate this capability using the following example.

{\bf The Norberg Example}. The Norberg insurance dataset \cite{norberg1989experience} consists of $k=72$ Norwegian occupational categories, where $y_i$ denotes the number of claims made against a policy.  Additionally, we have the total number of years each group was exposed to risk $E_i$; when normalized by a factor of 344, $E_i$ gives the expected number of claims during a contract period.  Similar to Norberg \cite{norberg1989experience}, we assume $Y_i \sim {\rm Poisson}(\theta_i E_i)$. Given the normalized $E_i$, we interpret $\theta_i$ as the occupational-specific rate of risk.  
\vskip.8em
DS-Bayes analysis yields the following estimated prior, where $g$ is the conjugate gamma prior with MLE $\alpha = 6.02$ and $\beta = 0.20$:
\begin{equation}\label{eq:nor_pi_hat}
\widehat{\pi}(\theta)\, = \,g(\theta; \alpha,\beta)\big[1 - 0.70T_1(\theta;G) + 0.83T_2(\theta;G)- 0.53T_3(\theta;G)\big].
\vspace{-.1em}
\end{equation} 
In Figure \ref{fig:appH_covar}(a), the U-function clearly indicates potential prior-data conflict when using $\pi(\te)={\rm Gamma}(6.02,0.20)$. Figure \ref{fig:appH_covar}(b) displays the DS prior (red) along with the parametric EB (blue) and the Kiefer-Wolfowitz NPMLE estimate (green).  We see a definite bimodality for $\hat{\pi}(\theta)$, indicating that there are two distinct groups of risk profiles.  The macroinference plot in Figure \ref{fig:appH_covar}(c) reinforces the structured heterogeneity of the data.  In terms of risk-profile, we consider the mode at $0.59$ as occupational categories with comparatively lower risk; these are occupations less likely to make a claim based on their risk exposure.  The mode at $1.46$ represents those occupations at a higher risk, thus more likely to make a claim based on their exposure. 
\begin{figure}[t]
  \centering
  \begin{subfigure}{.31\textwidth}
    \includegraphics[width=\linewidth,trim=1cm .5cm 0cm 1cm]{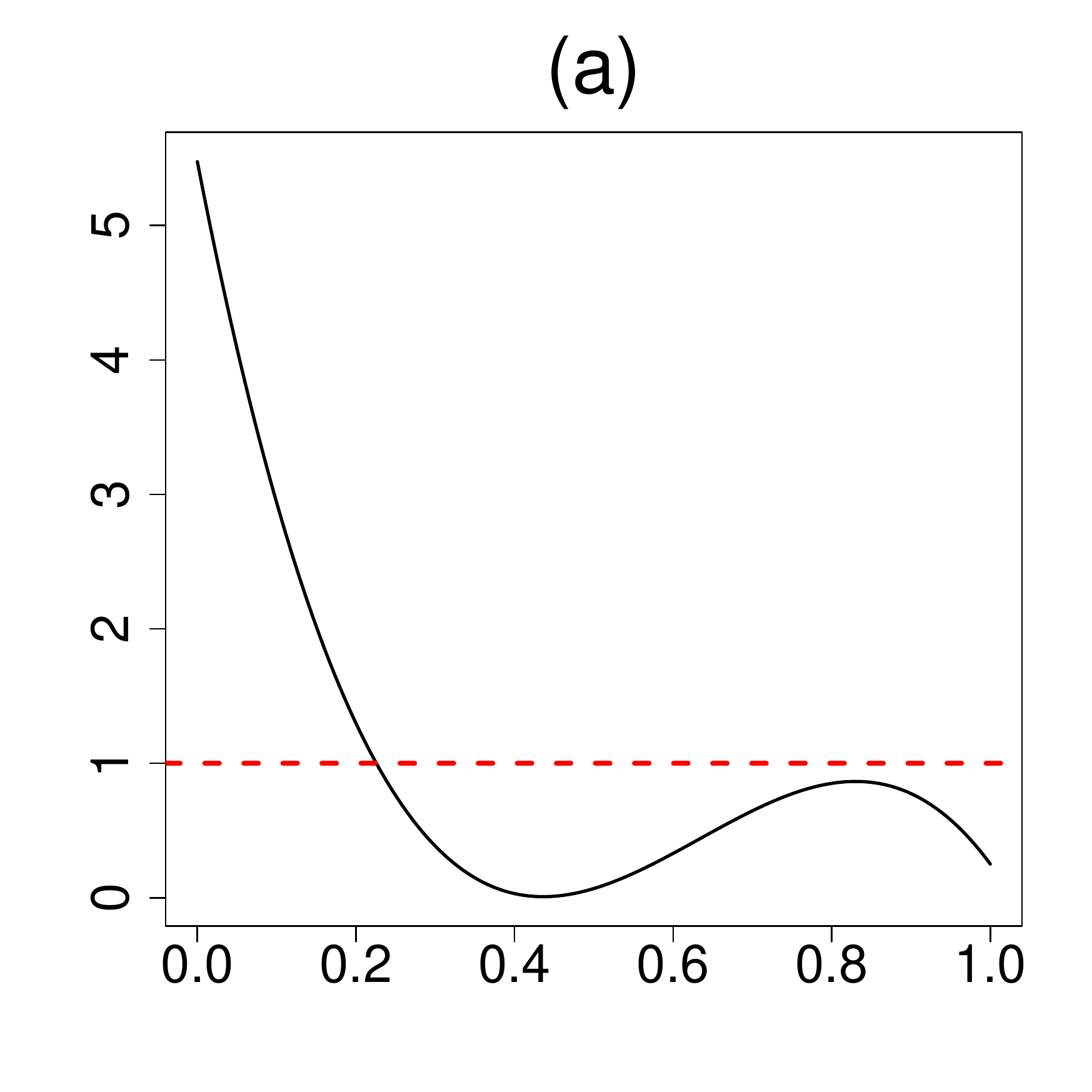}
  \end{subfigure}\hspace*{2.5mm}
\begin{subfigure}{.31\textwidth}
  \includegraphics[width=\linewidth,trim=1cm .5cm 0cm 1cm]{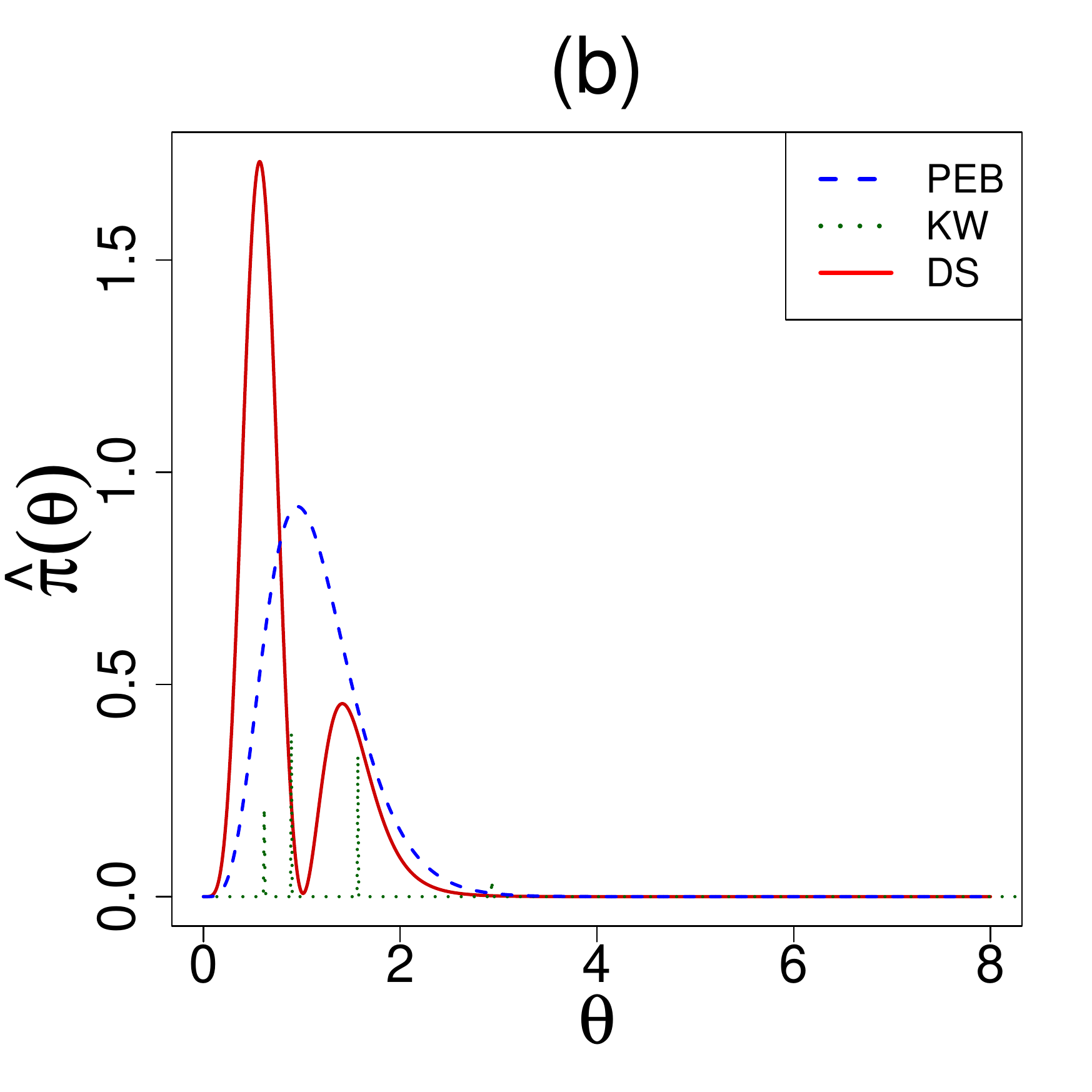}
  \end{subfigure}\hspace*{2.5mm}
  \begin{subfigure}{.31\textwidth}
    \includegraphics[width=\linewidth,trim=1cm .5cm 0cm 1cm]{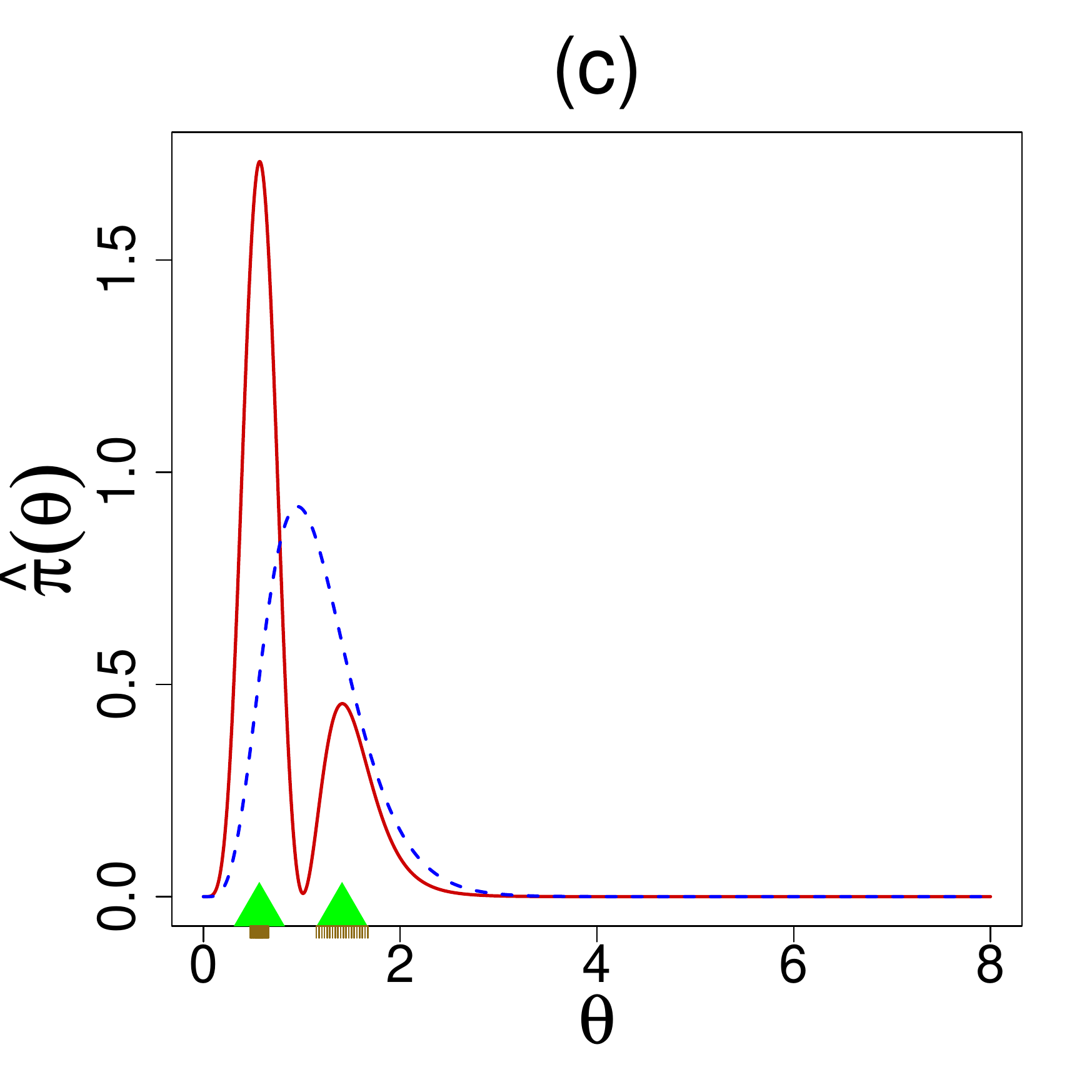}
  \end{subfigure}
  \par\medskip
    \begin{subfigure}{.31\textwidth}
    \includegraphics[width=\linewidth,trim=1cm .5cm 0cm 1cm]{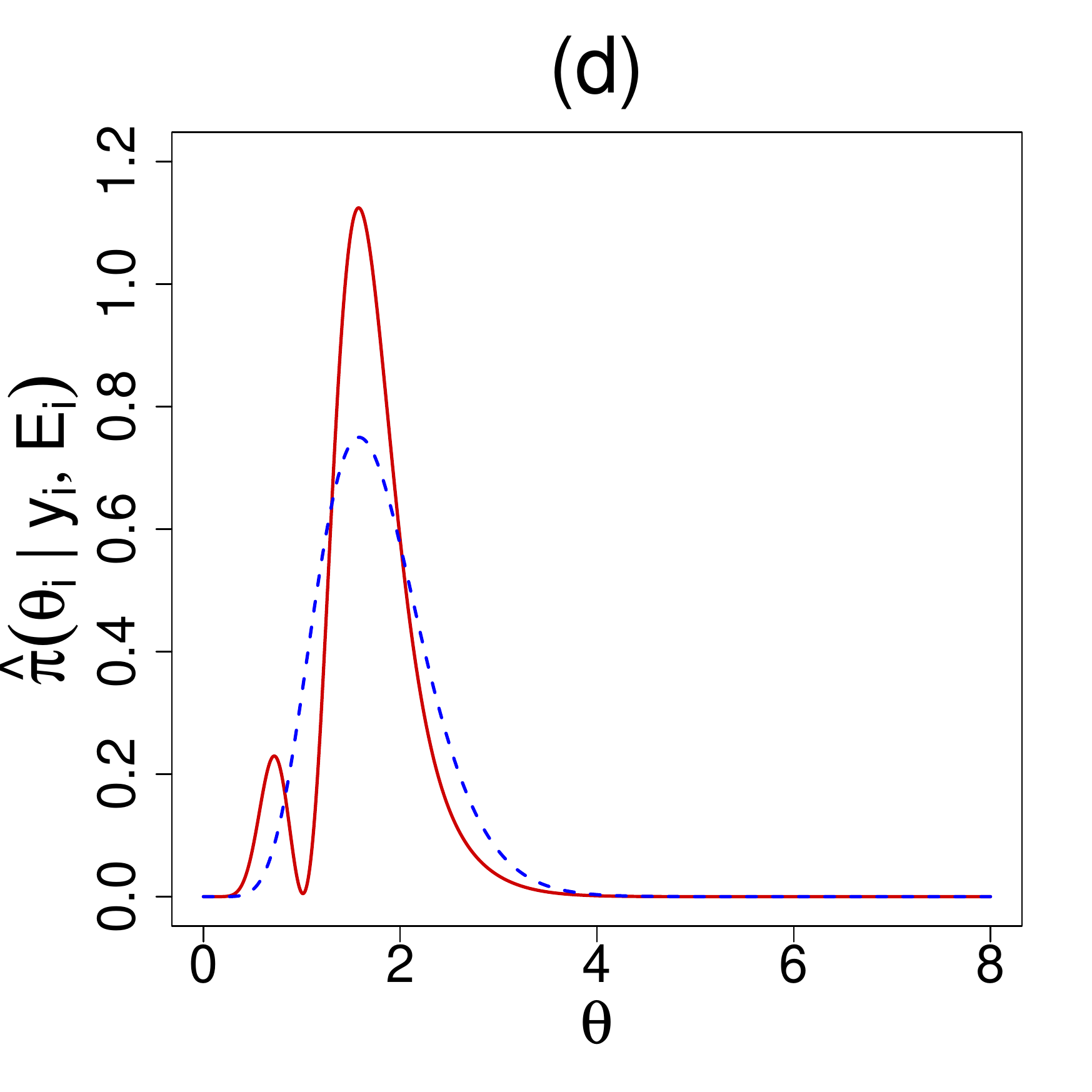}
  \end{subfigure}\hspace*{2.5mm}
\begin{subfigure}{.31\textwidth}
  \includegraphics[width=\linewidth,trim=1cm .5cm 0cm 1cm]{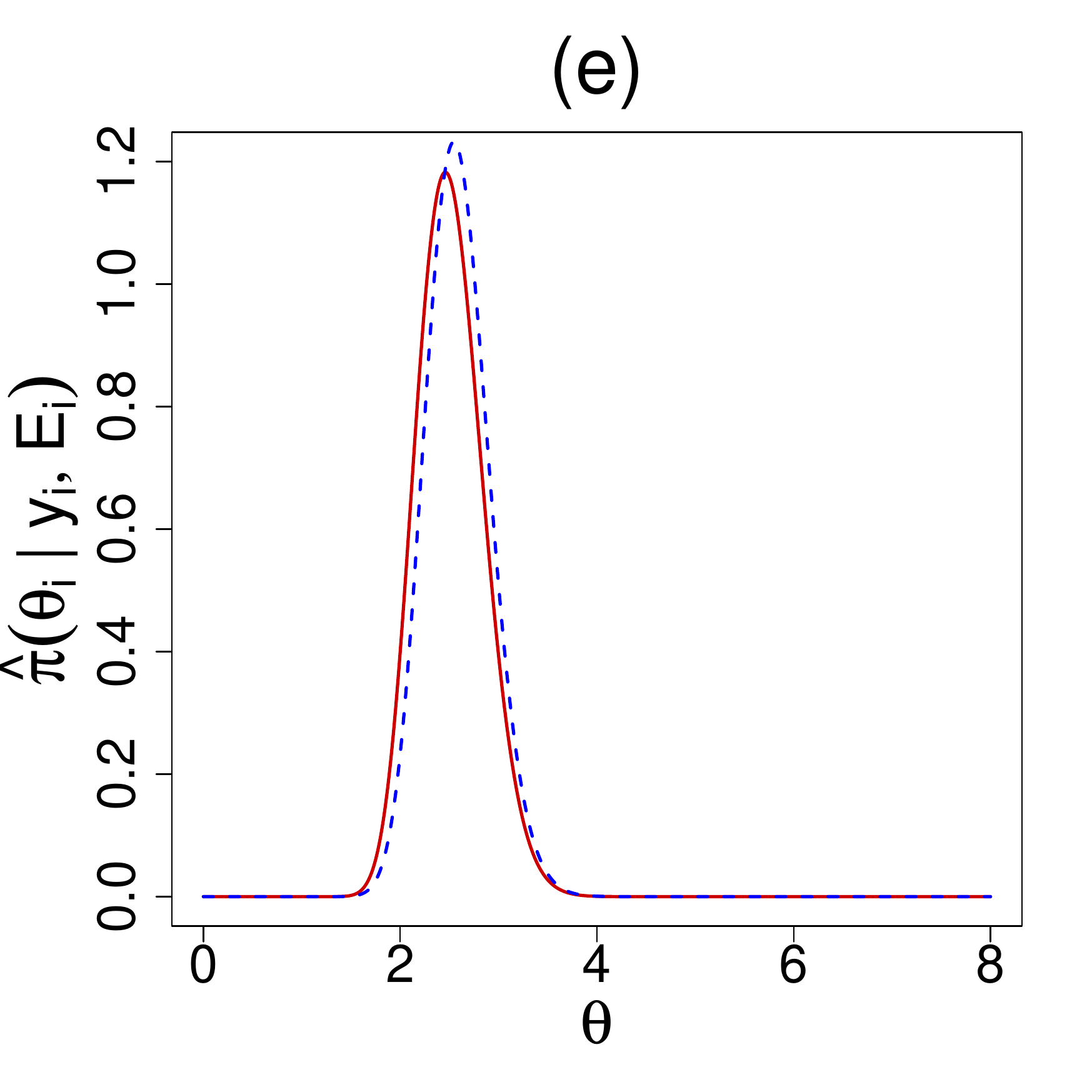}
  \end{subfigure}\hspace*{2.5mm}
  \begin{subfigure}{.31\textwidth}
    \includegraphics[width=\linewidth,trim=1cm .5cm 0cm 1cm]{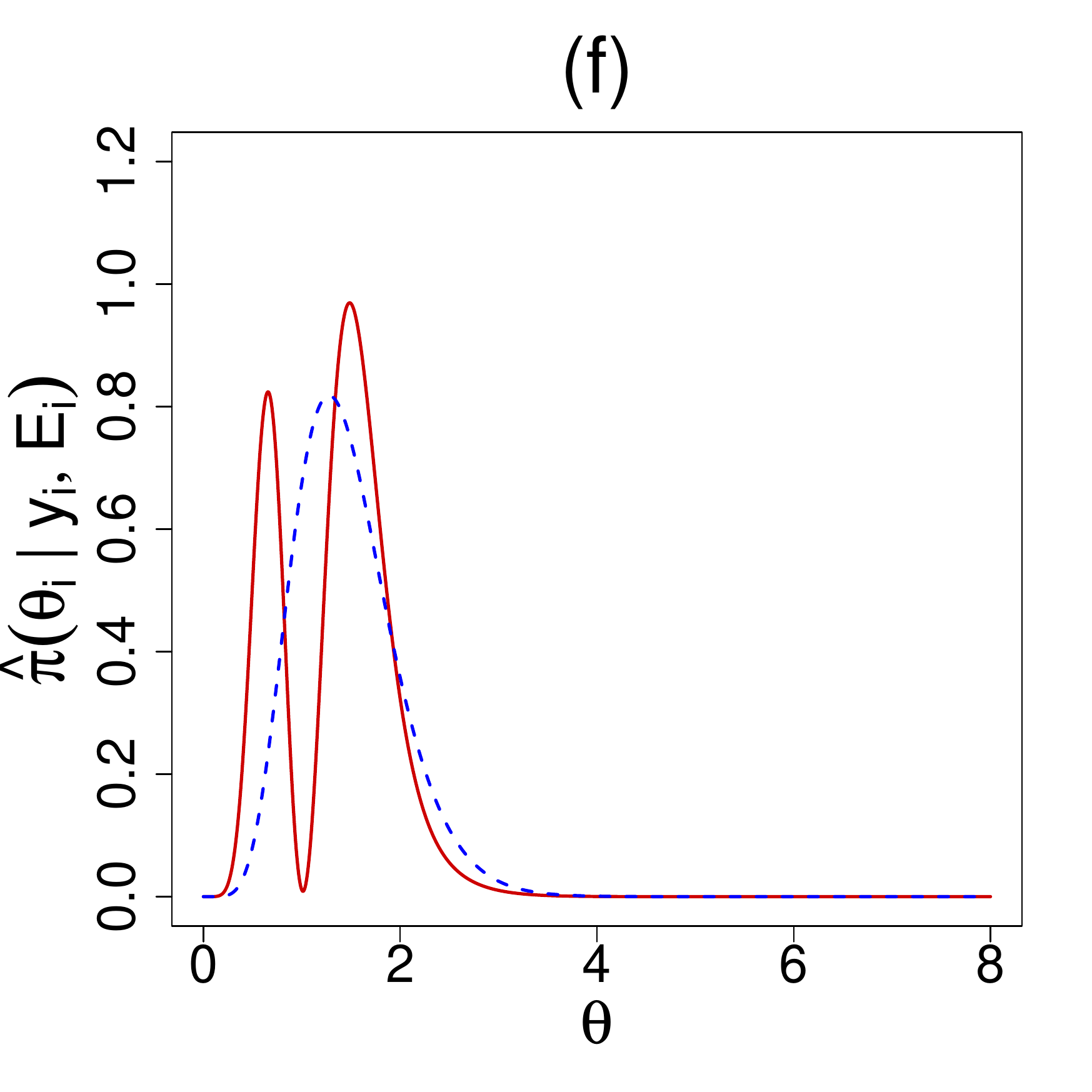}
  \end{subfigure}
  \vskip.4em
\caption{Demonstration of DS-Bayes with covariates on the Norberg insurance dataset.  In (a), we display the U-function.  Panel (b) shows the DS-prior (red), the PEB prior (blue) and the Kiefer-Wolfowitz NPMLE prior (green). Panel (c) shows the macroinference with standard errors (using smooth bootstrap): two modes located at $0.57 (\pm 0.094)$ and $1.41 (\pm 0.261)$. Panels (d) through (f) show microinference for occupational groups 13, 22, and 53 (respectively).}
\label{fig:appH_covar}
\end{figure}
Of particular interest are panels (d), (e), and (f).  These panels show the microinference for three specific occupational groups: group 13 ($Y_{13} = 4$, $E_{13} = 0.45$), group 22 ($Y_{22} = 57 $, $E_{22} = 19.1$), and group 53 ($Y_{53} = 2$, $E_{53} = 0.25$).  In Figure \ref{fig:appH_covar}(d), we have an occupational category that identifies as higher risk with a small lower risk component.  The unimodality in Figure \ref{fig:appH_covar}(e) clearly indicates that category is a higher risk of claim based on exposure. Finally, the occupational category in Figure \ref{fig:appH_covar}(f) is tricky.  Here, we have bimodality with an almost equal probability of being a high or low-risk occupation.  While the other two groups provide clear alternatives for an insurance company, the occupational group $53$ needs the company's judgment in assigning the policy.  
\begin{center} 
{\large I. MAXIMUM-ENTROPY ENHANCEMENT}
\end{center}
For more enhanced result, we offer an extension to maximum entropy $\DS(G,m)$ model, which assumes the following representation of the prior distribution: 
\beq \label{eq:LPme}
\breve{\pi}(\te)\,=\,g(\te;\al,\be)\, \exp \Big[ c_0 +\sum_j c_j \,T_j(\te;G)   \Big],
\eeq
where $c_0$ is some normalizing constant and the $c_j$'s are the LP-maximum entropy coefficients.  The following algorithm outlines the process to solve for the unknown $c_j$'s starting from the $\mathcal{L}^2$ estimate.

\makebox[\textwidth]{\textbf{Orthogonal Series to Maximum Entropy Estimator}}
\rule{\textwidth}{.8pt}
\texttt{Step 0.} Input: BIC-smoothed LP-Fourier ($\mathcal{L}^2$) coefficients $\widehat{\LP}[j;G,\Pi]$, $j = 1, \ldots, m$.
\vskip.4em
\texttt{Step 1.}  Define the set $\mathcal{J} = \left\{j: | \widehat{\LP}[j;G,\Pi]| > 0\right \}$, collection of $j$'s for which we have significant non-zero $\mathcal{L}^2$ orthogonal coefficients.
\vskip.4em
\texttt{Step 2.}  To estimate the maximum entropy coefficients $c_j$ in $\breve{\pi}(\theta)$ of \eqref{eq:LPme}, solve the following sets of moment equality constraints:
\begin{equation}
\widehat{\LP}[j;G,\Pi] = \int T_j(\theta;G) \breve{\pi}(\theta) d\theta,~~\text{for}~j \in \mathcal{J}.\\
\end{equation}
\vskip.4em
\texttt{Step 3.} Output: $\big(\hat c_0, \{\hat c_j\}_{j \in \mathcal{J}}\big)$; accordingly the estimated maximum entropy $\breve{d}$ and $\breve{\pi}$.\\
\rule{\textwidth}{.8pt}
\vskip.5em
{\bf Two Data Examples}. Here we carry out the maximum entropy analysis for \texttt{rat} (binomial variate) and \texttt{galaxy} data (normal variate). The \texttt{galaxy} data consists of $k=324$ observed rotation velocities $y_i$ and their uncertainties of Low Surface Brightness (LSB) galaxies \cite{de2001high}. 
\begin{enumerate}[label=(\alph*)]
\item Rat Tumor data, $g$ is beta distribution with MLE $\alpha = 2.30$, $\beta= 14.08$:
\begin{equation}\label{eq:rat_me_pi_hat}
\breve{\pi}(\theta) = g(\theta; \alpha,\beta)\exp\big[-0.13 - 0.52T_3(\theta;G) \big].
\end{equation}
\item Galaxy data, $g$ is normal distribution with MLE $\mu = 85.5$, $\tau^2= 3304$:
\begin{equation}\label{eq:gal_me_pi_hat}
\breve{\pi}(\theta) = g(\theta; \mu,\tau^2)\exp\big[-0.15 + 0.26T_3(\theta;G) - 0.28T_4(\theta;G) + 0.46T_5(\theta;G)\big].
\end{equation}
\end{enumerate}

\begin{figure}[t]
  \centering
  \begin{subfigure}{.44\textwidth}
  \includegraphics[width=\linewidth,trim=1cm .5cm 0cm 1cm]{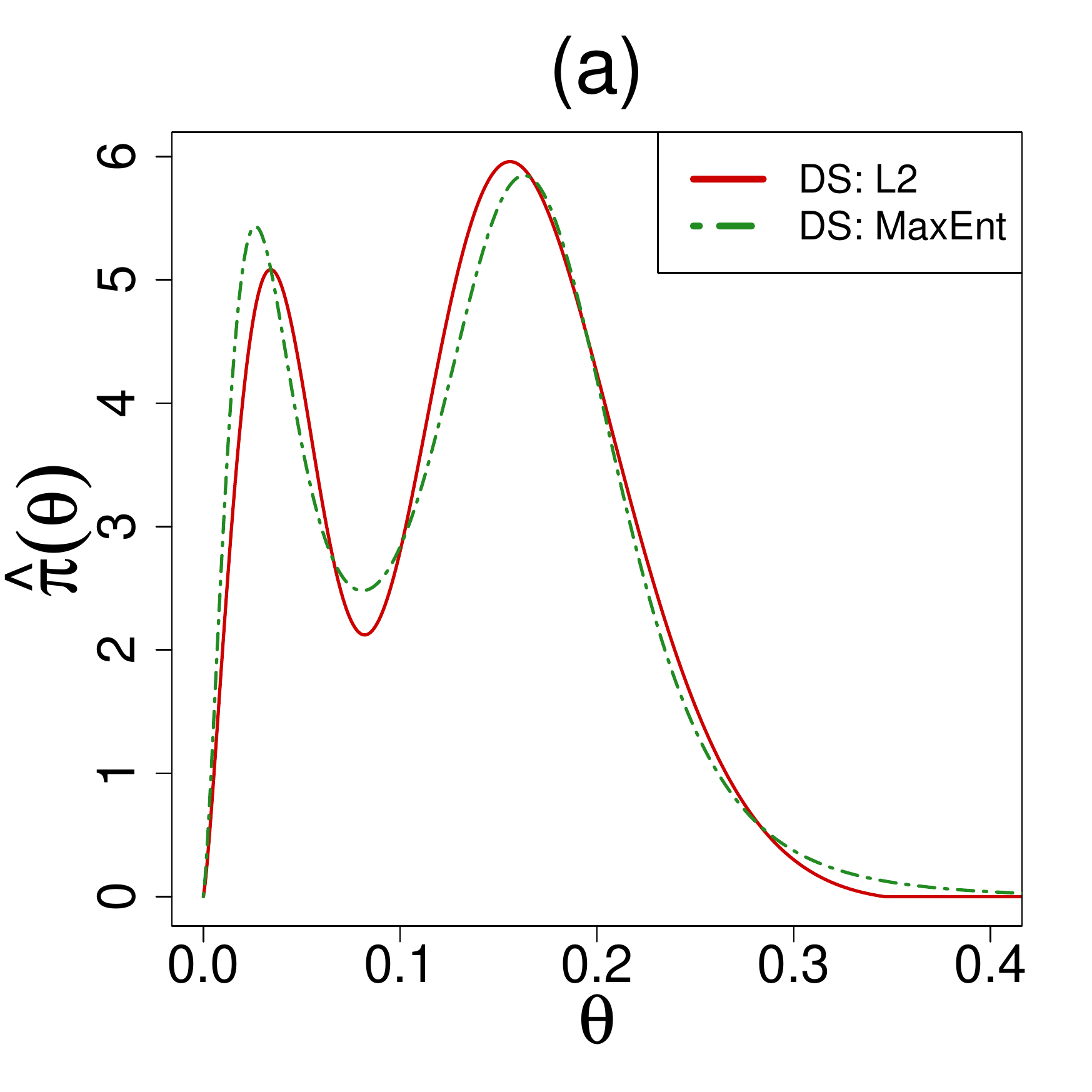}
  \end{subfigure}\hspace{12mm}%
  \begin{subfigure}{.44\textwidth}
    \includegraphics[width=\linewidth,trim=1cm .5cm 0cm 1cm]{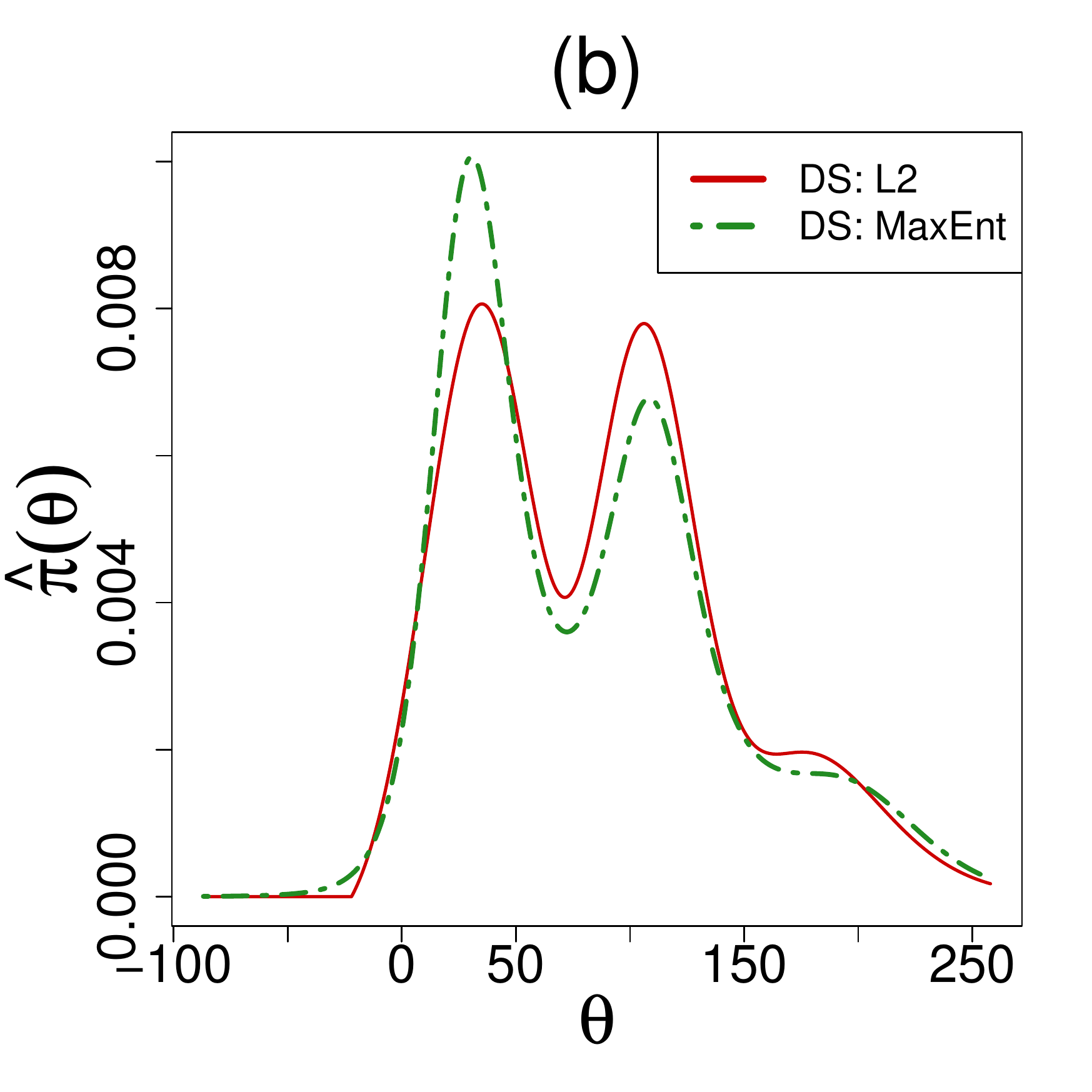}
  \end{subfigure}
  \vskip.25em
\caption{Comparison of $\mathcal{L}^2$ (solid red line) and maximum entropy (two-dash green line) estimates of DS prior.  Panel (a) shows the comparison for the \texttt{rat} tumor data, while panel (b) illustrates the difference (in modal shapes) for the \texttt{galaxy} data.}
\label{fig:appI_MaxEnt}
\vskip1em
\end{figure}
The resulting LP-maximum-entropy $\DS(G,m)$ priors are shown in Figure \ref{fig:appI_MaxEnt}.  In both examples, we see the maximum entropy estimates (green dashed lines) are very similar to the $\mathcal{L}^2$ with some adjustments to the modal shapes.

The \texttt{BayesGOF} package in R implements this algorithm as an option for the \texttt{DS.prior} function. The following code demonstrates how to generate both the $\mathcal{L}^2$ and maximum entropy representations of the LP coefficients for both the rat tumor and galaxy data sets.  
\vskip.75em
\begin{tcolorbox}[colback=gray!5!white,colframe=black]
 \begin{verbatim}
library(BayesGOF)
#---Rat Tumor Data
data(rat)
rat.start <- gMLE.bb(rat$y, rat$n)$estimate
rat.ds.L2 <- DS.prior(rat, max.m = 4, g.par = rat.start, 
                      family = "Binomial", LP.type = "L2")
## Shown in Figure 15(a) as solid red line                
rat.ds.ME <- DS.prior(rat, max.m = 4, g.par = rat.start, 
                      family = "Binomial", LP.type = "MaxEnt")
## Shown in Figure 15(a) as two-dashed green line
#---Galaxy Data
data(galaxy)
gal.start <- gMLE.nn(galaxy$y, galaxy$se)$estimate
gal.ds.L2 <- DS.prior(galaxy, max.m = 5, g.par = gal.start, 
                      family = "Normal", LP.type = "L2")
## Shown in Figure 15(b) as solid red line
gal.ds.ME <- DS.prior(galaxy, max.m = 5, g.par = gal.start, 
                      family = "Normal", LP.type = "MaxEnt")
## Shown in Figure 15(b) as two-dashed green line
\end{verbatim}
\end{tcolorbox}
\vskip.2em

 
\putbib[app-ref]
\end{bibunit}
\end{document}